%% file: main.tex
\documentclass[11pt,a4paper]{article}
\usepackage[utf8]{inputenc}
\usepackage[T1]{fontenc}
\usepackage{amsmath,amssymb,amsthm}
\usepackage{booktabs,array}
\usepackage{xcolor}
\usepackage{tikz}
\usepackage[margin=2.5cm]{geometry}
\usepackage{hyperref}
\hypersetup{colorlinks=true,linkcolor=blue!60!black,citecolor=blue!60!black,urlcolor=blue!60!black}

\newcommand{\doi}[1]{doi:\href{https://doi.org/#1}{\nolinkurl{#1}}}

\newtheorem{theorem}{Theorem}[section]
\newtheorem{lemma}[theorem]{Lemma}
\newtheorem{proposition}[theorem]{Proposition}
\newtheorem{corollary}[theorem]{Corollary}
\newtheorem{observation}[theorem]{Observation}
\newtheorem{fact}[theorem]{Fact}
\theoremstyle{definition}
\newtheorem{definition}[theorem]{Definition}
\theoremstyle{remark}
\newtheorem{remark}[theorem]{Remark}
\newtheorem{example}[theorem]{Example}

\newtheorem*{mainA}{Theorem A}
\newtheorem*{mainB}{Theorem B}
\newtheorem*{mainC}{Theorem C}

\newcommand{\VC}{\operatorname{VC}}
\newcommand{\B}{B}
\newcommand{\G}{\mathcal{G}}
\newcommand{\dt}{\tilde d}
\newcommand{\dB}{d_{\B}}
\newcommand{\dbr}{d_{\mathrm{br}}}
\newcommand{\Pos}{F}
\newcommand{\Lift}{L}
\newcommand{\sgn}{\operatorname{sgn}}
\newcommand{\topp}{\operatorname{top}_3}
\newcommand{\DIST}{\textup{\textsc{Dist}}}
\newcommand{\FLIP}{\textup{\textsc{Flip}}}
\newcommand{\sw}[1]{\langle #1\rangle}
\newcommand{\pl}[1]{\texttt{#1}}

\title{Computing distances in braid-move graphs,\\ higher Bruhat orders, and oriented-matroid\\ mutation graphs is NP-hard}
\author{Tilen Marc\\[2pt]
{\small Faculty of Mathematics and Physics, University of Ljubljana, Slovenia}\\
{\small Institute of Mathematics, Physics and Mechanics, Ljubljana, Slovenia}\\
{\small Abelium d.o.o., Ljubljana, Slovenia}}
\date{23 October 2026}

\begin{document}
\maketitle

\begin{abstract}
A wiring diagram is an arrangement of $m$ pseudolines, every two of which cross exactly once. Reading its crossings
from left to right encodes it as a reduced word of the permutation $w_0\in S_m$ that reverses $1,\dots,m$, unique up
to commutations, and a braid move $\sigma_i\sigma_{i+1}\sigma_i\leftrightarrow\sigma_{i+1}\sigma_i\sigma_{i+1}$
flips one of its triangles. We prove that deciding whether two wiring diagrams are within a given braid-move
distance is NP-complete. Equivalently, deciding flip distance is NP-complete in the higher Bruhat order
$\B(m,2)$, for triangle flips of marked arrangements of pseudolines, and for rhombic tilings of a $2m$-gon; this
answers an open question on flip graphs of pseudoline arrangements (SODA 2024). The reduction
is from \textup{\textsc{Vertex Cover}}: for every graph $G$ we construct two wiring diagrams at braid-move distance
$|D|+2\VC(G)$, where $\VC(G)$ is the vertex cover number of $G$ and the Hamming distance $|D|$ is the number of
triples of wires whose triangles are oriented differently. The construction glues copies of a seven-wire gadget
whose braid-move distance exceeds its Hamming distance by $2$.

Adding a line at infinity transfers the construction to rank-$3$ uniform oriented matroids, where two oriented
matroids are adjacent in the mutation graph if their chirotopes differ in one basis; so deciding distance in the
mutation graph is NP-complete in rank~$3$. An iterate of Rambau's expansion of signotopes amplifies the excess and
gives NP-hardness in every fixed rank $r\ge4$, and by duality in every fixed corank at least $3$. The same instances
show that flip distance in the higher Bruhat orders $\B(n,k)$ is NP-complete for every fixed $k\ge2$.
\end{abstract}



\section{Introduction}

\paragraph{Wiring diagrams and braid moves.}
Let $w_0\in S_m$ be the permutation that reverses the order of $1,\dots,m$, that is, $w_0(i)=m+1-i$; it is the
longest element of $S_m$. A \emph{reduced word} of $w_0$ is a shortest way of writing $w_0$ as a product of
adjacent transpositions $\sigma_i=(i\;\,i{+}1)$; it has $\binom m2$ letters. We place $m$ wires at the positions
$1,\dots,m$ and let each letter $\sigma_i$ exchange the wires at the positions $i$ and $i+1$. A reduced word of $w_0$
then exchanges every pair of wires exactly once; drawn as a \emph{wiring diagram}, it is an arrangement of $m$
pseudolines (Figures~\ref{fig:wiring} and~\ref{fig:gadget}). Commutations $\sigma_i\sigma_j=\sigma_j\sigma_i$ ($|i-j|\ge2$) do not change
the drawing up to isotopy, while a \emph{braid move}
$\sigma_i\sigma_{i+1}\sigma_i\leftrightarrow\sigma_{i+1}\sigma_i\sigma_{i+1}$ moves one wire across the crossing
of two others, that is, flips a triangle. Label the wires $1,\dots,m$ from top to bottom on the left. Three
wires $a<b<c$ cross in the order $ab,ac,bc$ (orientation $+$) or in the order $bc,ac,ab$ (orientation $-$); this
is the \emph{orientation} of the triple $abc$, and the orientations of all triples form the \emph{sign vector} of
the diagram. Every braid move changes the orientation of exactly one triple. So the \emph{braid-move distance} of
two wiring diagrams, the least number of braid moves that transform one into the other, is at least their
\emph{Hamming distance}, the number $|D|$ of triples with different orientations. The difference is the
\emph{excess}. Commutation classes of reduced words of $w_0$ correspond bijectively to
marked arrangements of $m$ pseudolines, to elements of the higher Bruhat order $\B(m,2)$ of Manin--Schechtman and
Ziegler~\cite{MS,Ziegler} (the rank-$3$ signotopes of Felsner and Weil~\cite{FW}), and to rhombic tilings of the
$2m$-gon~\cite{Elnitsky}, with braid moves corresponding to flips in each setting.

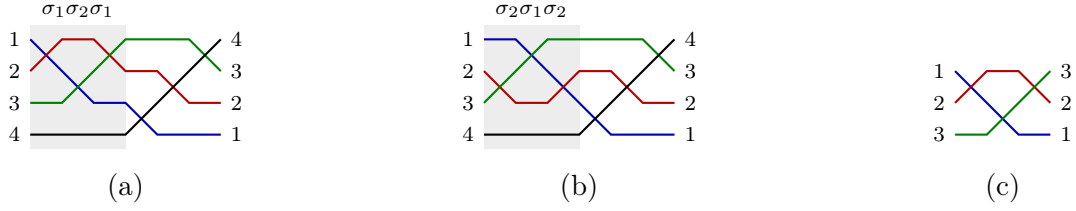
\begin{figure}[t]
\centering
\begin{minipage}[b]{0.3\textwidth}\centering\input{fig_w_a.tex}\par\smallskip(a)\end{minipage}\hfill
\begin{minipage}[b]{0.3\textwidth}\centering\input{fig_w_b.tex}\par\smallskip(b)\end{minipage}\hfill
\begin{minipage}[b]{0.25\textwidth}\centering\input{fig_w_c.tex}\par\smallskip(c)\end{minipage}
\caption{(a) The wiring diagram of the reduced word $\sigma_1\sigma_2\sigma_1\sigma_3\sigma_2\sigma_1$ of
$w_0\in S_4$. The wires are labelled by their positions on the left, position $1$ on top, and each letter
$\sigma_i$ exchanges the wires at positions $i$ and $i+1$. The shaded part contains the triangle of the wires $1,2,3$. All four
triples have orientation $+$. (b) The flip of this triangle, the braid move
$\sigma_1\sigma_2\sigma_1\to\sigma_2\sigma_1\sigma_2$. It changes the orientation of $123$ to $-$ and no other
orientation, so (a) and (b) have Hamming distance $1$. (c) Deleting the wire $4$ from (a) leaves the word
$\sigma_1\sigma_2\sigma_1$ on the wires $1,2,3$.}\label{fig:wiring}
\end{figure}

\paragraph{Flip distance.}
Flip graphs, whose vertices are combinatorial or geometric structures and whose edges are local moves, are studied
throughout discrete geometry, for their connectivity, their diameter and their distances; triangulations with edge
flips are the best-known example. Deciding flip distance is NP-hard for triangulations of point sets~\cite{LubiwPathak,Pilz} and of simple
polygons~\cite{AMP}; for convex polygons the question is classical, and a recent
preprint~\cite{Dorfer} announces NP-completeness. For pseudoline arrangements, Radtke et
al.~\cite{Radtke} study the triangle-flip graph of marked arrangements of $n$ pseudolines and ask
(their Question~5) what the computational complexity is of computing a shortest path between any pair of vertices
in this graph. By the correspondence above, this graph is isomorphic to the graph on the commutation classes of reduced words of
$w_0\in S_n$, with braid moves as edges.

\paragraph{Flip distance versus Hamming distance.}
For rhombic tilings, Bodini, Fernique, Rao and R\'emila~\cite{BFRR} compare the flip distance with the Hamming
distance $|D|$. The two agree for tilings with at most four directions and for certain families with five, but
already in $\B(6,2)$ there are $16$ pairs whose flip distance exceeds $|D|$ by $2$~\cite{BFR,BFRR}; one of them appears earlier in~\cite[Sect.~3]{FW00}. Since $|D|$ is easy to compute, determining the flip distance amounts to determining the excess,
and this is how our reduction works: it encodes a graph so that the excess is exactly twice its vertex cover number.

\paragraph{Related results.}
Several global parameters of these flip graphs are known. The diameter and the radius of $\B(m,2)$ are
$\binom m3$~\cite[Thms.~9 and~14]{GMS}. Diameters have also been determined in more general settings. Gutierres,
Mamede and Santos consider an arbitrary permutation $w$ in place of $w_0$: the graph on the commutation classes of the
reduced words of $w$, with braid moves as edges, has diameter equal to the number of $321$-patterns of $w$, that is,
of triples of wires that pairwise cross~\cite[Thm.~5.1]{GMS22}; for $w=w_0$ every triple is such a pattern. Reiner
and Roichman~\cite{ReinerRoichman} consider the longest elements of Coxeter groups, and more generally minimal
galleries in supersolvable hyperplane arrangements; for the longest elements of types $A$ (the symmetric groups) and
$B$ (the signed permutations) they determine the diameter of the graph whose vertices are the reduced words
themselves, with commutations as well as braid moves as edges. For the triangle-flip graph of marked arrangements of
$n$ pseudolines, Radtke et al.~\cite{Radtke} show that the diameter is $\Theta(n^3)$ and that the connectivity equals
the minimum degree $n-2$; its number of vertices, like the number of arrangements of $n$ pseudolines, is
$2^{\Theta(n^2)}$~\cite{FelsnerValtr,DumitrescuMandal}. Many questions about these graphs have remained open.
Besides the distance problem studied here, they include the mixing time of the natural random walk of triangle
flips, which is governed by the expansion of the graph~\cite[Question~4]{Radtke}, and the exact constant in the
exponent $\Theta(n^2)$ of the number of vertices.

\paragraph{Results for wiring diagrams.}
Let $\textup{\textsc{Braid}}$ be the problem: given two reduced words of $w_0\in S_m$ and an integer $K$ (in binary),
decide whether their braid-move distance is at most $K$.

\begin{mainA}
For every graph $G$ with $n$ vertices one can compute in polynomial time two reduced words $W^s_G,W^v_G$ of
$w_0\in S_m$, $m=3n+|E(G)|$, with sign vectors $s_G$ and $v_G$, whose braid-move distance is
\[
\dbr\bigl([W^s_G],[W^v_G]\bigr)\;=\;|D|+2\,\VC(G),\qquad |D|=3n(m-3)+6|E(G)| ,
\]
where $D=D(s_G,v_G)$ is the set of triples of wires whose orientations differ and $\VC(G)$ is the minimum size
of a vertex cover of $G$ (Theorem~\ref{thm:wiring}).
\end{mainA}

\begin{mainB}
The following problems are NP-complete, with the bound $K$ given in binary:
\begin{enumerate}
\item[(i)] $\textup{\textsc{Braid}}$, the braid-move distance of two reduced words of $w_0\in S_m$;
\item[(ii)] flip distance in the higher Bruhat order $\B(m,2)$;
\item[(iii)] triangle-flip distance of two marked arrangements of $m$ pseudolines (Question~5 of~\cite{Radtke});
\item[(iv)] flip distance of two rhombic tilings of the regular $2m$-gon.
\end{enumerate}
Moreover, there is $\delta>0$ such that approximating the excess $\dbr-|D|$ within a factor $1+\delta$ is NP-hard
(Theorem~\ref{thm:braidNP}, Corollary~\ref{cor:reform}, Remark~\ref{rem:apx}).
\end{mainB}

\paragraph{Oriented matroids.}
Oriented matroids generalize arrangements of pseudolines. Adding the line at infinity to a wiring diagram gives a
simple arrangement of pseudolines in the projective plane, and uniform oriented matroids of rank $3$ are exactly
such arrangements (with orientations), up to homeomorphism; uniform oriented matroids of rank $r$ are, in the same way,
simple arrangements of pseudospheres in $S^{r-1}$, and the realizable ones come from vectors in general position
in $\mathbb R^r$~\cite{BLSWZ}. A uniform oriented matroid is encoded by its chirotope: a sign $\chi(B)\in\{+,-\}$ for
each $r$-subset $B$ (a \emph{basis}), where $\chi$ and $-\chi$ describe the same oriented matroid. A
\emph{mutation} changes the sign of a single basis so that the result is again a chirotope; geometrically, one
pseudosphere moves across a vertex of a simplicial cell, and in rank $3$ a mutation is a triangle flip. The
\emph{mutation graph} $\G_{n,r}$ has the uniform oriented matroids of rank $r$ on $[n]$ as vertices and mutations
as edges.

Mutation graphs have long been central to the theory of oriented matroids: two of its best-known conjectures are
statements about them. For uniform oriented matroids, Las Vergnas' conjecture (1980) that every simple oriented
matroid has a simplicial tope says that the mutation graph has no isolated vertices, and Cordovil and Las Vergnas
(1988) conjectured that it is connected (see~\cite{KnauerMarc,GuKnauer}); in rank~$3$ it is, by Ringel's
theorem~\cite{Ringel} (see also Proposition~\ref{prop:diam}). Mutations have also been studied for their number and
their relation to realizability and Euclideanness (see e.g.~\cite{BokowskiRohlfs,Wilhelmi,KnauerMarc}). Recently, Gu
and Knauer~\cite{GuKnauer} disproved both conjectures: they found an oriented matroid of rank~$7$ on $24$ elements
without a simplicial tope, and disconnected mutation graphs $\G_{n,r}$ for all $r\ge7$ and $n\ge r+17$. This makes the
structure of mutation graphs all the more interesting. We approach them from the computational side: how hard is it to
compute the distance of two given oriented matroids? To our knowledge, distances in these graphs have not been
studied. The question makes sense in every rank, and it is the subject of Part~\ref{part:om}.

Let $\DIST_r$ be the problem: given two uniform chirotopes $\chi_1,\chi_2$ of rank $r$ on an $n$-element set (as
sign vectors of length $\binom nr$) and an integer $K$, decide whether the oriented matroids they represent have
distance at most $K$ in the mutation graph. We reach oriented matroids through the higher Bruhat orders $\B(n,k)$
of Manin and Schechtman~\cite{MS,Ziegler}, whose elements are the signotopes of rank $k+1$
(Definition~\ref{def:signotope}) and which generalize $\B(m,2)$: adding a line at infinity embeds $\B(n,k)$ into
the mutation graph $\G_{n+1,k+1}$ as an induced subgraph.

\begin{mainC}
If $G$ has at least three vertices, the two uniform oriented matroids of rank $3$ obtained from the wiring
diagrams of Theorem~A by adding a line at infinity have distance $|D|+2\VC(G)$ in the mutation graph
$\G_{m+1,3}$. So $\DIST_3$ is NP-complete, with $K$ in binary. For every fixed $r\ge 4$ the problem $\DIST_r$ is
NP-hard, and NP-complete when $K$ is given in unary. The same holds in every fixed corank $r\ge3$, i.e.\ for uniform oriented matroids of rank
$n-r$ on $n$ elements (Theorems~\ref{thm:rank3}, \ref{thm:np3} and~\ref{thm:allranks}). The same instances show that
deciding flip distance in the higher Bruhat order $\B(n,k)$ is NP-complete for every fixed $k\ge2$, with $K$ in binary
(Corollary~\ref{cor:hbo}).
\end{mainC}

\paragraph{Ideas.}
The reduction is from \textup{\textsc{Vertex Cover}}. Every vertex $u$ of $G$ gets a switch: a block of three wires
whose small triangle, the \emph{bead}, can be \emph{turned over} by one braid move. Every edge $uw$ gets a
\emph{gadget}, a piece of the two wiring diagrams that acts as an OR gate: it costs two braid moves more than the
Hamming distance unless the bead of $u$ or the bead of $w$ is turned over. Turning a bead over at the start and back
at the end costs two braid moves and serves all edges at that vertex. So the braid-move distance is $|D|$ plus twice
the least number of beads that serve all edges, that is, $|D|+2\VC(G)$.

The gadget is analysed with a classical local criterion, the \emph{packet criterion} (Lemma~\ref{lem:packet}):
whether two wiring diagrams have excess $0$, i.e.\ are at braid-move distance exactly $|D|$, can be decided by
looking at four wires at a time. The gadget has seven wires, the blocks $U$ and $V$ of the two endpoints and a
private wire between them (Section~\ref{sec:gadget}). In the construction (Section~\ref{sec:construction}) every edge
gets its own copy of the gadget, and two gadgets share at most one triple, a bead. For the lower bound we delete all
wires outside a single gadget, which isolates the cost of each edge; the upper bound is an explicit sequence of
braid moves.

For oriented matroids, adding the line at infinity keeps the upper bound. The lower bound needs more care, since a
sequence of mutations may move the line at infinity. We do not know whether this can ever help, but a sequence
without extra mutations never moves it, and this suffices for the gadget argument (Section~\ref{sec:rank3}). For
higher rank we lift the rank-$3$ instance by an iterate of Rambau's expansion of signotopes~\cite{FW,Rambau}. The
lift contains many copies of the rank-$3$ instance as minors, one for each set of $r-3$ new elements, and every
mutation of the lift changes at most one of them. So the excess over the Hamming distance is multiplied by roughly
the number of copies, up to a ratio close to $1$, and it still determines the vertex cover number
(Section~\ref{sec:higher}).

The packet criterion, the correspondence between reduced words and signotopes, the line at infinity (the positive fibre) and Rambau's
expansion are classical and are attributed where they are used; the gadget, the construction and the amplification
of the excess under the lift are new.

\paragraph{Formal verification.}
The main results are also verified in Lean~4 (see \hyperref[sec:computations]{\emph{Computations}}); the proofs in
this paper do not depend on this verification.

\paragraph{Organization.}
Part~\ref{part:wiring} treats wiring diagrams: Section~\ref{sec:prelim} gives the formal definitions,
recalls the correspondence with signotopes of rank~$3$ and the graph $\B(m,2)$, and proves the
packet criterion (Section~\ref{sec:packets}); Sections~\ref{sec:gadget}--\ref{sec:braid} contain the gadget, the
construction and the theorems on braid-move distance. Part~\ref{part:om} treats oriented matroids and higher Bruhat orders: Section~\ref{sec:om} introduces chirotopes,
signotopes and the higher Bruhat orders, and the positive fibre that connects them; Section~\ref{sec:rank3} treats
rank $3$, and Section~\ref{sec:higher} higher ranks, for both mutation graphs and higher Bruhat orders.
Section~\ref{sec:remarks} lists open problems.

\part{Wiring diagrams}\label{part:wiring}

\section{Wiring diagrams and braid moves}\label{sec:prelim}

Throughout, $[m]=\{1,\dots,m\}$, and signs are $\pm1$, written $+$ and $-$. For a finite totally ordered
set $S$, $\binom Sr$ is the set of $r$-subsets of $S$; we write a subset $\{a_1<\dots<a_r\}$ as the string
$a_1\cdots a_r$. For maps $u,v\colon\binom Sr\to\{\pm\}$ let $D(u,v)=\{X: u(X)\ne v(X)\}$, and for
$X\in\binom Sr$ let $u^X$ be $u$ with the value at $X$ negated.

\subsection{Reduced words, braid moves and signotopes}\label{sec:words}

The material of this section is standard. We fix our sign convention carefully, since the computations of
Section~\ref{sec:gadget} depend on it.

\begin{definition}\label{def:words}
Fix $m\ge2$. A \emph{word} is a finite sequence of \emph{letters} $\sigma_i$ with $1\le i\le m-1$. Its
\emph{wiring diagram} starts with wires $1,\dots,m$ at positions $1,\dots,m$, position $1$ on top; the letter
$\sigma_i$ exchanges the wires currently at positions $i$ and $i+1$, and we say that it \emph{swaps} them
and write $\sw{xy}$ for the swap of the pair $\{x,y\}$. A word is \emph{reduced} if it swaps every pair of
wires exactly once. (These are exactly the reduced words of the longest permutation $w_0\in S_m$; we use no
other reduced words.) In a reduced word $W$, the three swaps of wires $a<b<c$ occur either in the order
$\sw{ab},\sw{ac},\sw{bc}$ or in the order $\sw{bc},\sw{ac},\sw{ab}$: two wires can be swapped only when they are
adjacent, so the first of the three swaps is $\sw{ab}$ or $\sw{bc}$, and if it is $\sw{ab}$ the wires are then in the
order $b,a,c$, so that $\sw{ac}$ comes before $\sw{bc}$; the other case is symmetric. We put $s_W(abc)=+$ in the
first case and $s_W(abc)=-$ in the second; this is the \emph{orientation} of the triple $abc$ in $W$, and
$s_W\colon\binom{[m]}3\to\{\pm\}$ is the \emph{sign vector} of $W$.
\begin{itemize}
\item A \emph{commutation} exchanges two adjacent letters $\sigma_i\sigma_j$ with $|i-j|\ge2$. Two words are
\emph{equivalent}, $W\equiv W'$, if one arises from the other by commutations; the equivalence classes of
reduced words are their \emph{commutation classes} $[W]$.
\item A \emph{braid move} replaces three consecutive letters $\sigma_i\sigma_{i+1}\sigma_i$ by
$\sigma_{i+1}\sigma_i\sigma_{i+1}$, or conversely.
\item The \emph{braid-move distance} $\dbr([W],[W'])$ is the least number of braid moves in a sequence of
commutations and braid moves that transforms $W$ into $W'$; it depends only on the commutation classes.
\end{itemize}
\end{definition}

Any two reduced words of $w_0$ are related by commutations and braid moves~\cite{Matsumoto,Tits}, so $\dbr$ is
finite. By a \emph{sequence of braid moves} between two reduced words we always mean a sequence of commutations and
braid moves; only its braid moves are counted. As in the introduction, commutation classes are wiring diagrams up to isotopy and braid
moves are triangle flips (Remark~\ref{rem:triangle}). The convention of~\cite{Radtke} for
orientations differs from ours by a global sign.

\paragraph{Packets and signotopes.} For a map $u\colon\binom S3\to\{\pm\}$ and a $4$-set $P=abcd\subseteq S$, the
\emph{packet} of $P$ is the sequence
\[
u[P]=\bigl(u(abc),u(abd),u(acd),u(bcd)\bigr)
\]
of the values on the four triples of $P$ in lexicographic order; for the sign vector of a reduced word it lists the
orientations of the four triples of wires in $P$. A sequence of signs has \emph{at most one sign change} if, read
from left to right, it switches between $+$ and $-$ at most once, that is, if it is constant or consists of a block of
one sign followed by a block of the other. Equivalently, it has no \emph{alternating} subsequence of length three,
that is, no $x_i=x_k\ne x_j$ with $i<j<k$: such a subsequence forces two sign changes, and the first two sign changes,
at positions $i<j$, give the alternating subsequence $(x_i,x_{i+1},x_{j+1})$. Of the sixteen sign sequences of
length $4$ this leaves eight, which we denote by
\[
\begin{array}{llll}
v_0=\pl{++++},& v_1=\pl{+++-{}},& v_2=\pl{++-{}-{}},& v_3=\pl{+-{}-{}-{}},\\
v_4=\pl{-{}-{}-{}-{}},& v_5=\pl{-{}-{}-{}+},& v_6=\pl{-{}-{}++},& v_7=\pl{-{}+++};
\end{array}
\]
they are the sign vectors of the eight wiring diagrams on four wires in Figure~\ref{fig:C8}, which form the cycle
of Observation~\ref{obs:C8}. A map $u\colon\binom S3\to\{\pm\}$ all of whose packets have at most
one sign change is a \emph{signotope of rank $3$} on $S$, in the sense of Felsner and Weil~\cite{FW};
Definition~\ref{def:signotope} gives the definition for all ranks.

\begin{theorem}[\cite{MS,Ziegler,FW,Elnitsky}]\label{thm:corr}
Let $m\ge3$.
\begin{enumerate}
\item[(a)] Equivalent reduced words have the same sign vector, and $[W]\mapsto s_W$ is a bijection from the
commutation classes of reduced words of $w_0\in S_m$ onto the signotopes of rank~$3$ on $[m]$.
\item[(b)] If a braid move acting on the wires of a triple $T$ transforms a reduced word $W$ into $W'$, then $W'$ is
reduced and $s_{W'}=s_W^T$. Conversely, if $s_W^T$ is a signotope, then some word equivalent to $W$ admits a braid
move acting on the wires of $T$. So two commutation classes are related by a single braid move if and only if their
sign vectors differ in exactly one triple.
\item[(c)] The commutation classes correspond to the marked arrangements of $m$ pseudolines, with braid moves
corresponding to triangle flips, and to the rhombic tilings of the regular $2m$-gon, with braid moves
corresponding to flips of hexagons.
\end{enumerate}
\end{theorem}
For (a) and (b) see~\cite[Sect.~2, Thm.~3]{MS}, \cite[Thm.~4.1 and Cor.~4.2]{Ziegler} and \cite[Thms.~2
and~7]{FW}; for (c) see~\cite[Thms.~2 and~7]{FW} and~\cite{Elnitsky}. We say that a braid move
acting on the wires of $T$ \emph{flips} $T$.

We write $\B(m,2)$ for the graph on the signotopes of rank $3$ on $[m]$ in which two signotopes are adjacent if they
differ in exactly one triple, and $\dB$ for its distance (likewise $\B(S,2)$ for any finite totally ordered set
$S$); this is the Hasse diagram of the higher Bruhat order of Manin--Schechtman and Ziegler~\cite{MS,Ziegler}, taken
as an undirected graph. By Theorem~\ref{thm:corr},
\[
\dbr([W],[W'])=\dB(s_W,s_{W'})\qquad\text{for all reduced words }W,W'.
\]

\begin{remark}[a geometric picture]\label{rem:triangle}
Some word equivalent to $W$ admits a braid move on $T$ exactly when, in the wiring diagram, the three wires of $T$
bound a triangle that no other wire enters. The orientation of a triple $abc$ is $+$ exactly when the crossing of
$a$ and $c$ lies below the wire $b$ (Figure~\ref{fig:wiring}).
\end{remark}

\subsection{Hamming distance, local replacement and deleting wires}\label{sec:tools}

We state three elementary lemmas that are used throughout the paper.

\begin{lemma}[Hamming bound and parity]\label{lem:parity}
Let $u_0,\dots,u_\ell\colon\binom Sr\to\{\pm\}$ be maps such that each $u_i$ equals $u_{i-1}$ or $u_{i-1}^X$ for
some $X$, and let $c_X$ be the number of steps that \emph{flip} $X$, i.e.\ with $u_i=u_{i-1}^X$. Put
$D=D(u_0,u_\ell)$ and $Y=\{X:c_X>[X\in D]\}$, where $[X\in D]$ is $1$ if $X\in D$ and $0$ otherwise. Then $c_X$ is
odd if and only if $X\in D$, and $\sum_Xc_X\ge|D|+2|Y|$.
\end{lemma}
\begin{proof}
Each flip of $X$ changes the value at $X$, which gives the parity; so $c_X>[X\in D]$ implies $c_X\ge[X\in D]+2$.
\end{proof}
The maps $u_i$ can be the sign vectors along a sequence of commutations and braid moves (a commutation leaves the
sign vector unchanged, and a braid move flips one triple by Theorem~\ref{thm:corr}(b)), or, in Part~\ref{part:om},
signotopes or chirotopes along a walk. For words, with $D=D(s_W,s_{W'})$, Lemma~\ref{lem:parity} says that
$\dbr([W],[W'])\ge|D|$, that the number of braid moves of every sequence from $W$ to $W'$ has the parity of $|D|$,
and that a sequence with exactly $|D|$ braid moves flips every triple of $D$ once and nothing else.
We call the difference $\dbr([W],[W'])-|D(s_W,s_{W'})|$ the \emph{excess}.

\begin{lemma}[local replacement]\label{lem:replace}
Let $W=\mathsf A\,L\,\mathsf C$ be a reduced word such that $L$ consists of letters $\sigma_i$ with $q\le i\le
q+k-2$ and swaps every pair of the $k$ wires $Z$ that occupy positions $q,\dots,q+k-1$ after $\mathsf A$. Let
$L'$ be a word in the same letters that swaps every pair of these $k$ wires exactly once. We read $L$ and $L'$ as
reduced words on $k$ wires (positions shifted by $q-1$).
\begin{enumerate}
\item[(a)] $W'=\mathsf A\,L'\,\mathsf C$ is reduced, and $s_{W'}(t)=s_W(t)$ for every triple $t\not\subseteq Z$. If
moreover the wires of $Z$ are in label order after $\mathsf A$, then $s_W|_Z=s_L$ and $s_{W'}|_Z=s_{L'}$.
\item[(b)] A sequence of braid moves from $L$ to $L'$, applied to the letters of $L$ inside $W$, is a sequence of
braid moves from $W$ to $W'$ with the same number of braid moves.
\end{enumerate}
\end{lemma}
\begin{proof}
(a) $L$ moves no wire outside $Z$ and reverses $Z$, since it reverses the relative order of every pair of $Z$;
the same holds for $L'$. So $\mathsf A\,L$ and $\mathsf A\,L'$ end in the same arrangement, and $\mathsf C$
acts in the same way after both. Pairs not inside $Z$ are swapped once, in $\mathsf A$ or in $\mathsf C$, and
pairs inside $Z$ once in $L'$; so $W'$ is reduced. A triple with at most one wire in $Z$ has its swaps in
$\mathsf A$ or $\mathsf C$ in both words, in the same order. A triple with exactly two wires $x,y\in Z$ has
$\sw{xy}$ in $L$ (respectively $L'$) and its other two swaps at the same places in $\mathsf A$ or $\mathsf C$.
For the last statement: all swaps between two wires of $Z$ occur in $L$ (respectively $L'$), and the wires
of $Z$ enter $L$ at positions $q,\dots,q+k-1$ in label order, so $L$, read on $k$ wires, moves the $j$-th
smallest wire of $Z$ as $W$ does; hence every triple of $Z$ has its three swaps in the same order in $W$ as
in $L$.

(b) Shifted by $q-1$, commutations and braid moves among the letters of $L$ are commutations and braid moves of
$W$. The intermediate words have the form $\mathsf A\,L''\,\mathsf C$, where $L''$ is again reduced on $k$ wires
(Theorem~\ref{thm:corr}(b)), so they are reduced by~(a).
\end{proof}

\paragraph{Deleting wires.} Let $W$ be a reduced word on the wires $[m]$ and $A\subseteq[m]$ with $|A|\ge2$. When
$W$ swaps two wires of $A$, they are adjacent in the current arrangement, hence adjacent among the wires of $A$.
Listing these swaps in the order in which they occur, each as the letter $\sigma_j$ with $j$ the position of the
upper of the two wires among the wires of $A$ at that moment, gives a reduced word $W|_A$ on $|A|$ wires (the
wires of $A$ numbered in label order); in the wiring diagram, the wires outside $A$ are erased
(Figure~\ref{fig:wiring}(c)). The three swaps of a triple of $A$ occur in the same order in $W|_A$ as in $W$, so
$s_{W|_A}$ is the restriction of $s_W$ to the triples of $A$.

\begin{lemma}[deleting wires]\label{lem:delete}
Let $W,W'$ be reduced words on $[m]$ and $A\subseteq[m]$ with $|A|\ge3$. A sequence of braid moves from $W$ to $W'$
yields a sequence of braid moves from $W|_A$ to $W'|_A$ whose braid moves flip, in the same order, the triples
inside $A$ that the original sequence flips.
\end{lemma}
\begin{proof}
Let $V,V'$ be consecutive words of the sequence. By Theorem~\ref{thm:corr}(a),(b), $s_{V'}=s_V^T$ if the step is a
braid move flipping $T$, and $s_{V'}=s_V$ otherwise. Since $s_{V|_A}$ and $s_{V'|_A}$ are the restrictions of $s_V$
and $s_{V'}$ to the triples of $A$, they are equal unless the step flips a triple $T\subseteq A$, and then
$s_{V'|_A}=s_{V|_A}^T$. By Theorem~\ref{thm:corr}(a),(b), $[V|_A]=[V'|_A]$ in the first case, and $[V|_A]$ and
$[V'|_A]$ are related by one braid move flipping $T$ in the second.
\end{proof}

\subsection{Four wires and the packet criterion}\label{sec:packets}

This subsection proves a useful tool, the packet criterion (Lemma~\ref{lem:packet}): a sequence with exactly $|D|$ braid moves exists if and only if the differing triples can
be ordered so that, on every four wires, they are flipped in the order of an arc of an $8$-cycle. The criterion
is classical in substance; we state it in the form we use.

Deleting all but four wires $a<b<c<d$ from a reduced word $W$ gives a reduced word on four wires, whose sign
vector is the packet $s_W[abcd]$. There are eight wiring diagrams on four wires, one for each of the packets
$v_0,\dots,v_7$ (Figure~\ref{fig:C8}).

\begin{figure}[t]
\centering
\input{fig_c8.tex}
\caption{The eight wiring diagrams on four wires used in the packet criterion, and their sign vectors
$v_0,\dots,v_7$, listed as
$(s(123),s(124),s(134),s(234))$. Two of them are related by a braid move exactly when they are neighbours on the
cycle; each edge is labelled by the triple it flips.}\label{fig:C8}
\end{figure}
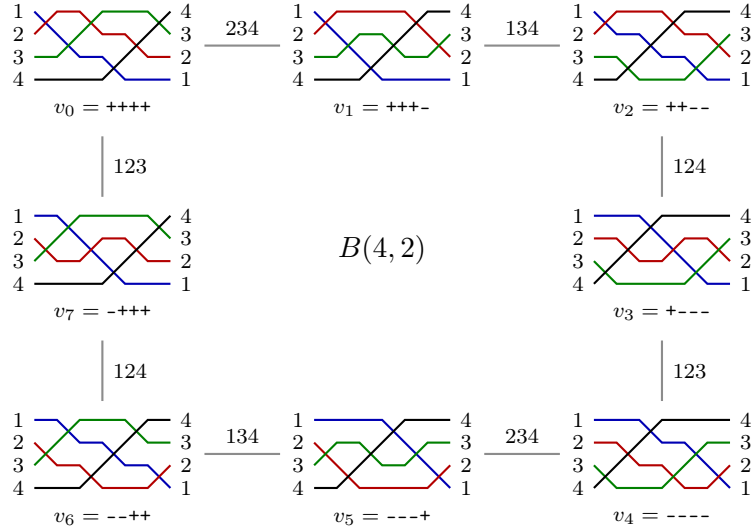

\begin{observation}\label{obs:C8}
Two wiring diagrams on four wires are related by a braid move exactly when they are neighbours in
Figure~\ref{fig:C8}, so $\B(4,2)$ is a cycle $C_8$ of length $8$.
\end{observation}

Now let $W,W'$ be reduced words on $[m]$ and $D=D(s_W,s_{W'})$. For a $4$-set $P=abcd\subseteq[m]$ we identify
the coordinates $1,2,3,4$ of packets with the triples $abc,abd,acd,bcd$. Deleting the wires outside $P$ projects
every sequence of braid moves from $W$ to $W'$ to a walk on $C_8$ from $s_W[P]$ to $s_{W'}[P]$, with one step for
each braid move that flips a triple inside $P$ (Lemma~\ref{lem:delete}). If the sequence has exactly $|D|$ braid
moves, it flips every triple of $D$ once and no other triple (Lemma~\ref{lem:parity}). Then the projected walk flips
each of the $k$ triples of $D$ inside $P$ once, so it never turns back, and it runs along an arc of $C_8$ of
length $k$.

If $x,y\in C_8$ have Hamming distance $k\ge1$, the two arcs of $C_8$ from $x$ to $y$ have lengths $k$ and $8-k$
(Figure~\ref{fig:C8}). An \emph{arc order} of $(x,y)$ is the order in which an arc of length $k$ from $x$ to $y$
flips its coordinates; it is unique if $k\le3$, and if $k=4$ there are two, reverse to
each other. An enumeration $t_1,\dots,t_{|D|}$ of $D$ is \emph{admissible} for $(W,W')$ if for every $4$-set
$P\subseteq[m]$ containing at least two triples of $D$, the triples of $D$ inside $P$ occur in an arc order of
$(s_W[P],s_{W'}[P])$. So a sequence with exactly $|D|$ braid moves flips the triples of $D$ in an admissible order.
The packet criterion is a local tool for the converse: whether there is a sequence of length exactly $|D|$ can be
decided by looking at four wires at a time.

\begin{lemma}[packet criterion]\label{lem:packet}
Let $W,W'$ be reduced words on $[m]$ and $D=D(s_W,s_{W'})$. Then $\dbr([W],[W'])=|D|$ if and only if $D$ has an
admissible enumeration.
\end{lemma}
\begin{proof}
Suppose first that $\dbr([W],[W'])=|D|$, and let $t_1,\dots,t_{|D|}$ be the triples flipped by the braid moves of a
shortest sequence from $W$ to $W'$, in this order. By Lemma~\ref{lem:parity} every triple of $D$ appears exactly
once, so $t_1,\dots,t_{|D|}$ is an enumeration of $D$. For a $4$-set $P$ containing $k\ge2$ triples of $D$, deleting
the wires outside $P$ projects the sequence to a walk on $C_8$ that flips each of these $k$ triples once
(Lemma~\ref{lem:delete}). This walk never turns back, so it runs along an arc of length $k$, and the $k$ triples
appear in an arc order. Hence the enumeration is admissible.

Conversely, let $t_1,\dots,t_{|D|}$ be an admissible enumeration of $D$, and put $s_0=s_W$ and $s_i=s_{i-1}^{t_i}$.
For a $4$-set $P$ the packets $s_0[P],s_1[P],\dots$ change exactly at the steps with $t_i\subseteq P$. If $P$
contains at least two triples of $D$, they run along an arc of $C_8$; if it contains at most one, they take only the
values $s_W[P]$ and $s_{W'}[P]$. In both cases every packet has at most one sign change, therefore each of
$s_0,\dots,s_{|D|}$ is a signotope, and $s_0,\dots,s_{|D|}$ is a walk of length $|D|$ in $\B(m,2)$ from $s_W$ to
$s_{W'}$. By Theorem~\ref{thm:corr}, $\dbr([W],[W'])=\dB(s_W,s_{W'})\le|D|$, and equality holds by
Lemma~\ref{lem:parity}.
\end{proof}

For the two commutation classes with constant sign vectors, sequences with $\binom m3$ braid moves are the maximal
chains of $\B(m,2)$, and the proof of Lemma~\ref{lem:packet} becomes the classical description of maximal chains, in which the
triples of every $4$-set are flipped in lexicographic or reverse lexicographic
order~\cite[Sect.~2, Thm.~3(c)]{MS}, \cite[Thm.~4.1(A)]{Ziegler}; Ziegler~\cite[Thm.~4.5]{Ziegler} treats
comparable pairs, and the argument for arbitrary pairs is the same.

\section{The gadget}\label{sec:gadget}

The goal of Sections~\ref{sec:gadget}--\ref{sec:braid} is to turn a graph $G$ into two reduced words whose excess
is twice the vertex cover number of $G$. The gadget encodes one edge of $G$. It is a pair of reduced words on seven wires whose braid-move distance exceeds the Hamming
distance by $2$; turning over either of its two beads, in both words, gives a pair without excess
(Proposition~\ref{prop:gadget}(b)).
It is built from three ingredients that reappear in the construction of Section~\ref{sec:construction}: two
blocks of three wires, a bead in each block, and a private wire.

\paragraph{Blocks and the private wire.}
The seven wires are, in label order,
\[
U=\{1,2,3\},\qquad p=4,\qquad V=\{5,6,7\}.
\]
The \emph{blocks} $U$ and $V$ stand for the two endpoints $u<w$ of an edge $uw$ (for a fixed linear order $<$ of
the vertices of $G$), and $p$ is the \emph{private
wire} of the edge. In the construction every vertex $u$ receives a block $X_u$ of three consecutive wires, which
plays the role of $U$ in the gadget of every edge $uw$ with $w>u$ and the role of $V$ in the gadget of every edge
$wu$ with $w<u$; every edge receives its own private wire. Gadgets of different edges share at most one block,
and hence at most one triple (Proposition~\ref{prop:lowerW}).
\begin{center}\small
\begin{tabular}{llll}\toprule
wires & in the gadget & bead triple & in the construction\\\midrule
$U=\{1,2,3\}$ & block of the endpoint $u$ & $123$ & the block $X_u$ of the vertex $u$\\
$p=4$ & private wire & -- & the wire $p_e$ of the edge $e=uw$\\
$V=\{5,6,7\}$ & block of the endpoint $w$ & $567$ & the block $X_w$ of the vertex $w$\\\bottomrule
\end{tabular}
\end{center}

\paragraph{Beads.}
In every word we build, the gadget words below as well as the words of Section~\ref{sec:construction}, the three
wires of a block are swapped with each other in a single factor of three letters, during which they occupy three
consecutive positions. This factor is the \emph{bead} of the block; in the wiring diagram it is a small triangle
(Figure~\ref{fig:gadget}). It is one of the following two words. For a position $q$ let
\[
\beta_q=\sigma_q\sigma_{q+1}\sigma_q,\qquad\beta'_q=\sigma_{q+1}\sigma_q\sigma_{q+1},\qquad
\beta^0_q=\beta_q,\quad\beta^1_q=\beta'_q .
\]
Both reverse the three wires at positions $q,q+1,q+2$. If these are $x<y<z$ in this order (which is automatic
in a reduced word, since the bead swaps all three pairs), $\beta$ swaps them in the order
$\sw{xy},\sw{xz},\sw{yz}$ and $\beta'$ in the reverse order. We say that $\beta^\epsilon_q$ has \emph{type} $\epsilon$. The two types give opposite signs to the \emph{bead triple} $xyz$, and by Theorem~\ref{thm:corr}(b):
\begin{itemize}
\item[(B)] exchanging $\beta_q$ and $\beta'_q$ inside a reduced word is a braid move; it changes $s_W$ exactly on
the bead triple.
\end{itemize}
We say that a bead of type $1$ is \emph{turned over}. In the reduction, the beads that are turned over form
a vertex cover.

\paragraph{The cell.}
The private wire is threaded through the two blocks by the word (Figure~\ref{fig:kappa})
\[
\kappa=\sigma_4\sigma_3\sigma_2\sigma_1\cdot\sigma_4\sigma_3\cdot\sigma_5\sigma_4\sigma_3\sigma_2\cdot
\sigma_6\sigma_5\sigma_4\sigma_3\cdot\sigma_4 .
\]

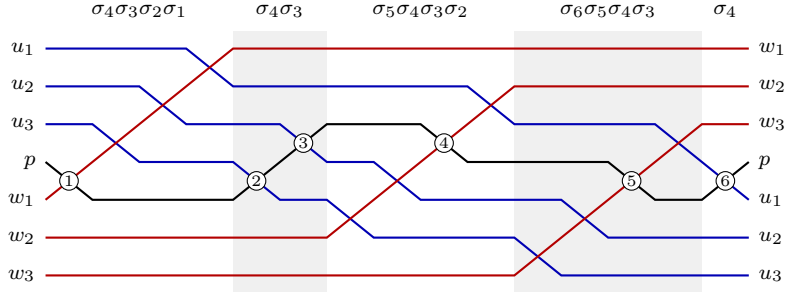
\begin{figure}[t]
\centering
\input{fig_kappa.tex}
\caption{The word $\kappa$ on the wires $u_1,u_2,u_3$ (the block $U$, blue), $p$ (black) and $w_1,w_2,w_3$ (the
block $V$, red), named by their positions $1,\dots,7$ when $\kappa$ starts; the five parts of $\kappa$ are marked
above. Every pair of wires from different groups crosses once and no pair inside a group crosses (K1); at the end
$V$ is on top and $U$ at the bottom, each in its original order (K2); the numbers mark the crossings of $p$, which
meets $w_1,u_3,u_2,w_2,w_3,u_1$ in this order (K3).}\label{fig:kappa}
\end{figure}

Let $\kappa$ act on seven wires, whatever their labels. Call the wires at positions $1,2,3$ when $\kappa$ starts
$u_1,u_2,u_3$ (the $U$-wires), the wire at position $4$ the $p$-wire, and the wires at positions $5,6,7$
$w_1,w_2,w_3$ (the $V$-wires). Following $\kappa$ letter by letter (Figure~\ref{fig:kappa}) shows:
\begin{itemize}
\item[(K1)] $\kappa$ swaps every pair of wires from two different of these three groups exactly once, and no pair
inside a group;
\item[(K2)] $\kappa$ ends with the $V$-wires at positions $1$--$3$, the $p$-wire at $4$ and the $U$-wires at
$5$--$7$, each group in its previous internal order;
\item[(K3)] the $p$-wire meets the six block wires in the order $w_1,u_3,u_2,w_2,w_3,u_1$.
\end{itemize}

\begin{definition}[the gadget pairs]\label{def:gadget}
For $\epsilon=(\epsilon_1,\epsilon_2)\in\{0,1\}^2$ the \emph{gadget words} are
$\beta^{\epsilon_1}_1\,\beta^{\epsilon_2}_5\,\kappa$ and $\kappa\,\beta^{\epsilon_2}_1\,\beta^{\epsilon_1}_5$,
with sign vectors
\[
g^s_\epsilon=s_{\beta^{\epsilon_1}_1\,\beta^{\epsilon_2}_5\,\kappa},\qquad
g^v_\epsilon=s_{\kappa\,\beta^{\epsilon_2}_1\,\beta^{\epsilon_1}_5}.
\]
Both words are reduced by (K1). In both, $U$ carries a bead of type $\epsilon_1$ and $V$ one
of type $\epsilon_2$ (after $\kappa$, $V$ is at positions $1$--$3$ and $U$ at $5$--$7$), placed before
$\kappa$ in the first word and after it in the second. We write $g^s=g^s_{00}$ and $g^v=g^v_{00}$.
Figure~\ref{fig:gadget} shows the two wiring diagrams.
\end{definition}
\begin{figure}[t]
\centering
\input{fig_gs.tex}\par\medskip\input{fig_gv.tex}
\caption{The wiring diagrams of $\beta_1\beta_5\kappa$ (top, giving $g^s$) and $\kappa\beta_1\beta_5$
(below, giving $g^v$). Wires $1,2,3$ (the block $U$) are blue, wire $4$ ($p$) is black, and wires $5,6,7$
(the block $V$) are red. The shaded regions are the beads.}\label{fig:gadget}
\end{figure}
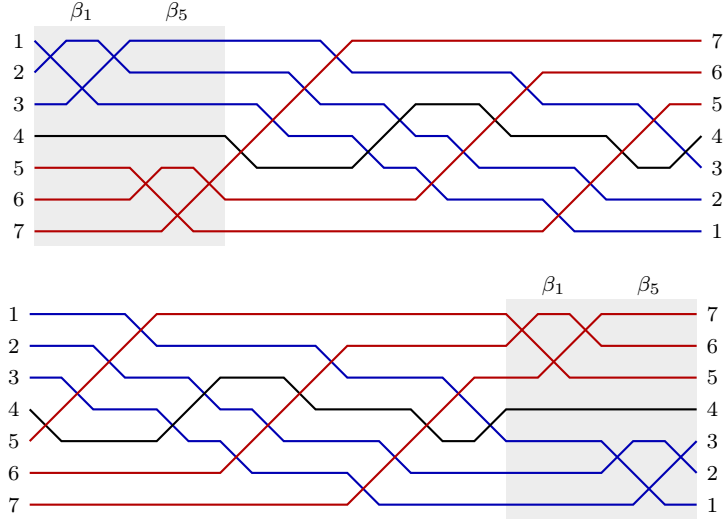

\begin{proposition}[the gadget]\label{prop:gadget}
\begin{enumerate}
\item[(a)] For every $\epsilon$, the pair $(g^s_\epsilon,g^v_\epsilon)$ differs on the same set $D_\kappa$ of $30$
triples.
\item[(b)] \emph{(OR)} The braid-move distance of the two words is
\begin{gather*}
\dbr\bigl([\beta_1\beta_5\kappa],[\kappa\beta_1\beta_5]\bigr)=|D_\kappa|+2=32,\\
\dbr\bigl([\beta^{\epsilon_1}_1\beta^{\epsilon_2}_5\kappa],[\kappa\beta^{\epsilon_2}_1\beta^{\epsilon_1}_5]\bigr)
=|D_\kappa|=30\quad\text{for }\epsilon\ne(0,0).
\end{gather*}
\end{enumerate}
\end{proposition}

The proposition concerns eight explicit reduced words on seven wires; we prove it below. \texttt{check.py} also
confirms the four distances by breadth-first search in $\B(7,2)$, which has $24698$ vertices
(see \hyperref[sec:computations]{\emph{Computations}} at the end of the paper).

In words, (b) says that the gadget is an OR of its two beads: if neither bead is turned over, every sequence of braid moves
between the two words needs two extra braid moves (Lemma~\ref{lem:parity}), and turning over either bead removes
them.

\begin{proof}[Proof of Proposition~\ref{prop:gadget}(a)]
Each row of Table~\ref{tab:seq30} below lists the same $30$ triples; this is the set $D_\kappa$. That $g^s$ and $g^v$
differ exactly on these triples is a direct check with Definition~\ref{def:words}: the orientation of a triple is
read off from the order of its three swaps in the two words, or from the two diagrams of Figure~\ref{fig:gadget};
\texttt{check.py} does the same check (see \hyperref[sec:computations]{\emph{Computations}}). The bead triples $123$ and $567$ are not in
$D_\kappa$, since their pairs are swapped only by the beads, which have the same type in both words. Changing the type
of a bead changes the sign vector exactly on its bead triple, in both words (property~(B)), so
$D(g^s_\epsilon,g^v_\epsilon)=D_\kappa$ for every $\epsilon$.
\end{proof}

The lower bound in (b) comes from a cycle of forced orders. A sequence without excess flips the differing triples
of every four wires in an arc order (Section~\ref{sec:packets}), and each of the six $4$-sets of
Table~\ref{tab:cycle} has exactly three differing triples, so that its arc order is unique.

\begin{table}[t]
\centering\small
\begin{tabular}{cccl}
\toprule
$P$ & $g^s[P]$ & $g^v[P]$ & forced order\\
\midrule
1246 & \pl{++++} & \pl{-{}-{}-{}+} & $124\prec126\prec146$\\
1467 & \pl{+-{}-{}-{}} & \pl{-{}-{}++} & $146\prec467\prec167$\\
4567 & \pl{-{}-{}-{}+} & \pl{++++} & $467\prec457\prec456$\\
2456 & \pl{++-{}-{}} & \pl{-{}+++} & $256\prec456\prec245$\\
2345 & \pl{+++-{}} & \pl{-{}-{}-{}-{}} & $245\prec235\prec234$\\
1234 & \pl{++++} & \pl{+-{}-{}-{}} & $234\prec134\prec124$\\
\bottomrule
\end{tabular}
\caption{Six packets of the gadget with $|D_\kappa\cap\binom P3|=3$ and their unique arc orders
(Figure~\ref{fig:C8}); coordinates are ordered $abc,abd,acd,bcd$ for $P=abcd$.}\label{tab:cycle}
\end{table}

\begin{lemma}\label{lem:gadgetlower}
$\dbr([\beta_1\beta_5\kappa],[\kappa\beta_1\beta_5])>|D_\kappa|=30$.
\end{lemma}
\begin{proof}
By the packet criterion (Lemma~\ref{lem:packet}), it suffices to show that $D_\kappa$ has no admissible enumeration for the pair, that is, no order in which the triples of $D_\kappa$ inside
every $4$-set $P$ containing at least two of them appear in an arc order of $(g^s[P],g^v[P])$. Table~\ref{tab:cycle} lists six such
$4$-sets $P$ with their packets $g^s[P]$ and $g^v[P]$; each packet is the sign vector of the word on four wires
obtained from a gadget word by deleting the other three wires (Section~\ref{sec:words}), and can be read off
Figure~\ref{fig:gadget}. In each of the six rows the two packets differ in exactly three triples, so the arc order is
unique (Figure~\ref{fig:C8}); it is the order in the last column. The six orders chain into the cycle
$124\prec146\prec467\prec456\prec245\prec234\prec124$, so no such order of $D_\kappa$ exists.
\end{proof}

Only two of the six $4$-sets of Table~\ref{tab:cycle} contain a bead triple: $1234$, which forces
$234\prec134\prec124$, and $4567$, which forces $467\prec457\prec456$. Turning over the bead of $U$ (respectively
$V$) negates $123$ (respectively $567$) in both sign vectors and reverses the corresponding forced order, which breaks
the cycle. This is why the gadget is an OR of its two beads.

The upper bounds come from the explicit orders of Table~\ref{tab:seq30}.

\begin{table}[t]
\centering\small
\begin{tabular}{cl}
\toprule
$\epsilon$ & order of the $30$ flips\\
\midrule
$(1,0)$ & $127,124,126,125,137,146,136,145,135,134,247,267,257,256,467,$\\
 & $457,456,367,357,347,356,346,245,235,234,236,237,167,157,156$\\
$(0,1)$ & $237,137,127,156,256,456,356,346,245,145,235,135,125,234,134,$\\
 & $124,236,136,126,146,457,357,257,157,347,367,247,267,467,167$\\
$(1,1)$ & $127,124,126,125,137,146,136,145,135,134,247,256,257,267,456,$\\
 & $457,467,356,357,367,347,346,245,235,234,236,237,156,157,167$\\
\bottomrule
\end{tabular}
\caption{For each $\epsilon\ne(0,0)$, an admissible order of $D_\kappa$ for the gadget pair of $\epsilon$; by the
packet criterion, $\beta^{\epsilon_1}_1\beta^{\epsilon_2}_5\kappa$ and $\kappa\beta^{\epsilon_2}_1\beta^{\epsilon_1}_5$
are therefore at braid-move distance $30$.}\label{tab:seq30}
\end{table}

\begin{proof}[Proof of Proposition~\ref{prop:gadget}(b)]
By Lemma~\ref{lem:parity}, every sequence of braid moves between the two words of a gadget pair has at least
$|D_\kappa|=30$ braid moves, and their number has the parity of $30$.

\emph{Distance $30$ for $\epsilon\ne(0,0)$.} The row of Table~\ref{tab:seq30} for $\epsilon$ is an enumeration of
$D_\kappa$, and a direct check shows that it is admissible for the pair: for every $4$-set $P\subseteq[7]$
containing at least two triples of $D_\kappa$, these triples appear in an arc order of
$(g^s_\epsilon[P],g^v_\epsilon[P])$. (This finite check is also done by \texttt{check.py}.) By the packet criterion (Lemma~\ref{lem:packet}), the
distance is $30$.

\emph{Distance $32$ for $\epsilon=(0,0)$.} By Lemma~\ref{lem:gadgetlower} and parity, the distance is at least
$32$. Conversely, turning over the bead of $U$ in the first word is one braid move (property~(B)) and gives the first
word for $\epsilon=(1,0)$; then $30$ braid moves lead to the second word $\kappa\beta_1\beta'_5$ for
$\epsilon=(1,0)$; and turning the bead of $U$, now at positions $5$--$7$, back is one more braid move. This gives
$32$ braid moves.
\end{proof}

\paragraph{Why seven wires.} Pairs of wiring diagrams with positive excess already exist on six
wires~\cite{FW00,BFR,BFRR}. The construction needs more: two disjoint triples whose flip in both words removes the
excess, here the bead triples $123$ and $567$. None of the $16$ pairs on six wires with positive excess has two
disjoint such triples. Among the $60$ possible cells formed by two blocks with beads and a private wire, $\kappa$ is
the only one with the OR property. Both facts were found by computer (see
\hyperref[sec:computations]{\emph{Computations}}).

\section{The construction}\label{sec:construction}

Let $G$ be a graph with vertex set $[n]$ and edge set $E$; we write an edge as $uw$ with $u<w$, and $\VC(G)$ for the minimum size of a vertex cover of $G$. We construct
two reduced words $W^s_G$ and $W^v_G$ on $m=3n+|E|$ wires, with sign vectors $s_G$ and $v_G$. Both come from one
word $\widehat\Phi$, in which the wires of each block are never swapped with each other; the three swaps
inside a block are supplied by its bead, which stands at the very beginning of $W^s_G$ and at the very end of
$W^v_G$. The
diagram $\widehat\Phi$ is built in two steps. First we draw every block and every private wire as a single
wire and describe how these wires cross (the \emph{shadow sweep} $\Phi$, Section~\ref{sec:shadow}); then we
replace each of these single wires by the wires of its group (the \emph{blow-up}, Section~\ref{sec:blowup}).
Figure~\ref{fig:p3} shows $\Phi$ and the two words for the path $P_3$.

\subsection{Wires, groups and bundles}

For each vertex $u$ there is a \emph{block} $X_u$ of three wires, and for each edge $e$ a \emph{private wire}
$p_e$. Blocks and private wires are called \emph{groups}. The private wire of an edge belongs to its smaller
endpoint: for $u\in[n]$ let $w_1<\dots<w_{k_u}$ be the neighbours of $u$ that are larger than $u$, and let the
\emph{bundle} of $u$ be the list of groups
\[
B_u=(X_u,\Pi_u),\qquad \Pi_u=(p_{uw_1},\dots,p_{uw_{k_u}}).
\]
Listing the bundles $B_1,B_2,\dots,B_n$ one after another gives an order of the groups, the \emph{label order}.
We label the $m=3n+|E|$ wires $1,\dots,m$ following it: the three wires of a block receive three consecutive labels,
and a private wire receives one label.

\begin{example}\label{ex:p3wires}
For the path $P_3$ with edges $12$ and $23$ (Figure~\ref{fig:p3}), the bundles are
$B_1=(X_1,p_{12})$, $B_2=(X_2,p_{23})$ and $B_3=(X_3)$, so the groups in label order are
$X_1,p_{12},X_2,p_{23},X_3$, and the wires are labelled $X_1=\{1,2,3\}$, $p_{12}=4$, $X_2=\{5,6,7\}$, $p_{23}=8$,
$X_3=\{9,10,11\}$.
\end{example}

For an edge $e=uw$ let $A_e=X_u\cup\{p_e\}\cup X_w$. In label order these seven wires are the three wires of
$X_u$, then $p_e$ (which lies in $B_u$), then the three wires of $X_w$ (which lies in the later bundle $B_w$).
We identify them with $U,p,V$ of Section~\ref{sec:gadget}, and $A_e$ with $[7]$ in label order.

\subsection{The shadow sweep}\label{sec:shadow}

In outline, the shadow sweep reverses the order of the bundles $B_1,\dots,B_n$ by exchanging adjacent bundles,
one pair at a time. In round $u$ the later bundles $B_{u+1},\dots,B_n$ pass up through $B_u$ one after another, so
that $B_u$ ends up below all of them. When $uw$ is an edge, the crossing of $B_u$ and $B_w$ is routed so that the
three groups $X_u,p_{uw},X_w$ are exchanged with each other in three consecutive moves, a \emph{cell}; this is
where the gadget of $uw$ will be placed. The rest of this subsection makes the outline precise.

\paragraph{Arrangements of groups and moves.}
In this subsection each group is a single wire. An \emph{arrangement} is the list of all $n+|E|$ groups from
top to bottom; the initial arrangement is the label order. A group $g$ lies \emph{directly above} a group $h$
(and $h$ \emph{directly below} $g$) if $h$ immediately follows $g$ in the arrangement; a set of groups lies directly above
another set if the two sets are consecutive in the list, the first above the second. A \emph{move} exchanges two
groups of which one lies directly above the other; this is a letter $\sigma_i$ acting on the positions of the
groups. If $h_1,\dots,h_k$ are the $k$ groups directly above $g$, with $h_1$ nearest to $g$, then \emph{$g$
passes up through $h_1,\dots,h_k$} means the $k$ moves that exchange $g$ with $h_1$, then with $h_2$, and so on.
Afterwards $g$ lies directly above $h_k$, and $h_1,\dots,h_k$ have moved down by one position in the same order.
The shadow sweep $\Phi$ is a sequence of moves in which every pair of groups is exchanged exactly once, i.e.\ a
reduced word on the groups; at the end the arrangement is the reverse of the label order.

\paragraph{Cells.}
If the three groups $X_u,p_{uw},X_w$ of an edge $uw$ lie in this order at three consecutive positions, the
cell of $uw$ is the sequence of three moves in which $X_w$ passes up through $p_{uw}$ and $X_u$, and
then $p_{uw}$ passes up through $X_u$. As a word it is $\sigma_{i+1}\sigma_i\sigma_{i+1}$, where $i$ is the
position of $X_u$. It exchanges the pairs $\{p_{uw},X_w\}$, $\{X_u,X_w\}$, $\{X_u,p_{uw}\}$ in this order and
leaves the three groups in the order $X_w,p_{uw},X_u$. In the blow-up the three moves of a cell are replaced
together by the word $\kappa$ of Section~\ref{sec:gadget}, and this is where the gadget of the edge is placed.
We therefore describe $\Phi$ so that for every edge $uw$ the three groups $X_u,p_{uw},X_w$ are exchanged with
each other in a cell.

\paragraph{The rounds.}
$\Phi$ consists of rounds $u=1,\dots,n$, followed by a completion. At the start of round $u$ the arrangement is
$B_u,B_{u+1},\dots,B_n$, each bundle in its initial order, followed by the groups of $B_{u-1},\dots,B_1$; for $u=1$
this is the label order. In \emph{round $u$}, for $w=u+1,\dots,n$ in this order, the bundle $B_w$ passes up through
$B_u$, as described in the next paragraph, so that $B_u$ moves down past all later bundles, one bundle at a time.
For $u<w$ let $\Pi_u^{<w}$ and $\Pi_u^{\ge w}$ be the sublists of $\Pi_u$ of the wires $p_{uw'}$ with $w'<w$ and with
$w'\ge w$, respectively; $\Pi_u^{>w}$ is defined in the same way. Every passage starts in the following situation:
\begin{itemize}
\item[(I)] $B_u$ lies directly above $B_w$, the groups of $B_w$ are in their initial order $(X_w,\Pi_w)$, and the
groups of $B_u$ are in the order $(\Pi_u^{<w},X_u,\Pi_u^{\ge w})$.
\end{itemize}
For $w=u+1$ this is the arrangement at the start of the round, since $\Pi_u^{<u+1}$ is empty; for the later passages
it follows from the next paragraph.

\paragraph{One bundle passing up through another.}
Assume (I). \emph{$B_w$ passes up through $B_u$} means that the groups of $B_w$, starting with its top group $X_w$,
pass up one after another through all groups of $B_u$, with one modification when $uw$ is an edge.
\begin{itemize}
\item If $uw$ is not an edge, there is no private wire $p_{uw}$, and the passage consists of these moves. Every group
of $B_w$ is exchanged once with every group of $B_u$, and no other pair is exchanged.
\item If $uw$ is an edge, then $p_{uw}$ is the first wire of $\Pi_u^{\ge w}$, since $\Pi_u$ is in label order, that
is, in the order of the larger endpoints; so $p_{uw}$ lies directly below $X_u$. First $X_w$ passes up through the
other wires of $\Pi_u^{\ge w}$. Then $X_u,p_{uw},X_w$ lie consecutively in this order, and the cell of $uw$ is
applied. Then $X_w$ passes up through $\Pi_u^{<w}$, and the other groups of $B_w$ follow by moves as before. Besides
the pairs of a group of $B_w$ and a group of $B_u$, the cell exchanges the pair $\{X_u,p_{uw}\}$: inside $B_u$ the
wire $p_{uw}$ moves from below $X_u$ to above it.
\end{itemize}
In both cases $B_w$ ends directly above $B_u$, in its initial order, and the groups of $B_u$ end in the order
$(\Pi_u^{<w+1},X_u,\Pi_u^{\ge w+1})$. So (I) holds again when $B_{w+1}$ is about to pass up through $B_u$. At the end of the round, $B_u$ lies below
$B_{u+1},\dots,B_n$ in the order $(\Pi_u,X_u)$, and $B_{u+1},\dots,B_n$ are still in their initial order; this is the
arrangement at the start of round $u+1$.

\paragraph{The completion.}
After round $n$ the arrangement is $B_n,\dots,B_1$, each $B_u$ in the order $(\Pi_u,X_u)$; in particular the
private wires of each $\Pi_u$ lie consecutively, in label order. In the \emph{completion} their order is reversed:
for definiteness, for $j=2,\dots,k_u$ in turn, the wire $p_{uw_j}$ passes up through the private wires of $\Pi_u$
above it. This exchanges each pair of private wires of $\Pi_u$ once, and afterwards the arrangement is the reverse
of the label order. The order does not matter: any sequence of moves that reverses $\Pi_u$ and exchanges each pair once
would do, and the lists $\Pi_u$ for different $u$ are separated by blocks and may be handled in any order.

\begin{example}\label{ex:c4}
For the cycle $C_4$ with edges $12,23,34,14$ the bundles are $B_1=(X_1,p_{12},p_{14})$, $B_2=(X_2,p_{23})$,
$B_3=(X_3,p_{34})$ and $B_4=(X_4)$. The arrangements during the sweep are (top to bottom):
\begin{center}\small
\begin{tabular}{ll}\toprule
start & $X_1,p_{12},p_{14},X_2,p_{23},X_3,p_{34},X_4$\\\midrule
round $1$: $B_2$ through $B_1$; $X_2$ through $p_{14}$, the cell of $12$ & $X_2,p_{12},X_1,p_{14},p_{23},X_3,p_{34},X_4$\\
\phantom{round $1$: $B_2$ through $B_1$;} $p_{23}$ through $p_{14},X_1,p_{12}$ & $X_2,p_{23},p_{12},X_1,p_{14},X_3,p_{34},X_4$\\
\phantom{round $1$: }$B_3$ through $B_1$ (no edge) & $X_2,p_{23},X_3,p_{34},p_{12},X_1,p_{14},X_4$\\
\phantom{round $1$: }$B_4$ through $B_1$; the cell of $14$, $X_4$ through $p_{12}$ & $X_2,p_{23},X_3,p_{34},X_4,p_{12},p_{14},X_1$\\\midrule
round $2$: $B_3$ through $B_2$; the cell of $23$, $p_{34}$ through $B_2$ & $X_3,p_{34},p_{23},X_2,X_4,p_{12},p_{14},X_1$\\
\phantom{round $2$: }$B_4$ through $B_2$ (no edge) & $X_3,p_{34},X_4,p_{23},X_2,p_{12},p_{14},X_1$\\\midrule
round $3$: $B_4$ through $B_3$; the cell of $34$ & $X_4,p_{34},X_3,p_{23},X_2,p_{12},p_{14},X_1$\\\midrule
completion: $p_{14}$ through $p_{12}$ & $X_4,p_{34},X_3,p_{23},X_2,p_{14},p_{12},X_1$\\\bottomrule
\end{tabular}
\end{center}
When $B_4$ passes up through $B_1$, the bundle $B_1$ is in the order $(p_{12},X_1,p_{14})$, so $p_{14}$ lies directly
below $X_1$ and the cell of $14$ comes first; afterwards $X_4$ continues through $\Pi_1^{<4}=(p_{12})$. At the end,
$B_1$ is in the order $(p_{12},p_{14},X_1)$, and the completion reverses $\Pi_1$.
\end{example}

We record what is used later.

\begin{lemma}\label{lem:phi}
$\Phi$ is a reduced word on the groups, and for every edge $e=uw$ the pairs $\{X_u,p_e\}$, $\{X_u,X_w\}$ and
$\{p_e,X_w\}$ are exchanged in the cell of $e$.
\end{lemma}
\begin{proof}
By the description above, every move and every cell acts on groups at consecutive positions. Two groups of
different bundles $B_u$, $B_w$ with $u<w$ are exchanged exactly once, when $B_w$ passes up through $B_u$. The
pair $\{X_u,p_{uw}\}$ is exchanged only in the cell of $uw$, and a pair of private wires of the same $\Pi_u$ only
in the completion. These are all pairs of groups, each exchanged exactly once.
\end{proof}

\begin{figure}[tp]
\centering
\input{fig_p3_phi.tex}

\medskip
{\small $s_G$:}\par\input{fig_p3_blowup.tex}

\medskip
{\small $v_G$:}\par\input{fig_p3_blowup_v.tex}
\caption{The construction for the path $P_3$ ($1-2-3$), with $m=11$ wires (Example~\ref{ex:p3}). Top: the
shadow sweep $\Phi$ on the groups $X_1,p_{12},X_2,p_{23},X_3$. Middle: the word $W^s_G$, the three beads followed
by the blow-up $\widehat\Phi$. Bottom: the word $W^v_G$, the blow-up $\widehat\Phi$ followed by the three beads.
Colours: $X_1=\{1,2,3\}$ blue, $p_{12}=4$ black, $X_2=\{5,6,7\}$ red, $p_{23}=8$ grey, $X_3=\{9,10,11\}$ green;
the shaded regions, labelled in the figure, mark the two cells, the beads and the two copies of $\kappa$.}\label{fig:p3}
\end{figure}
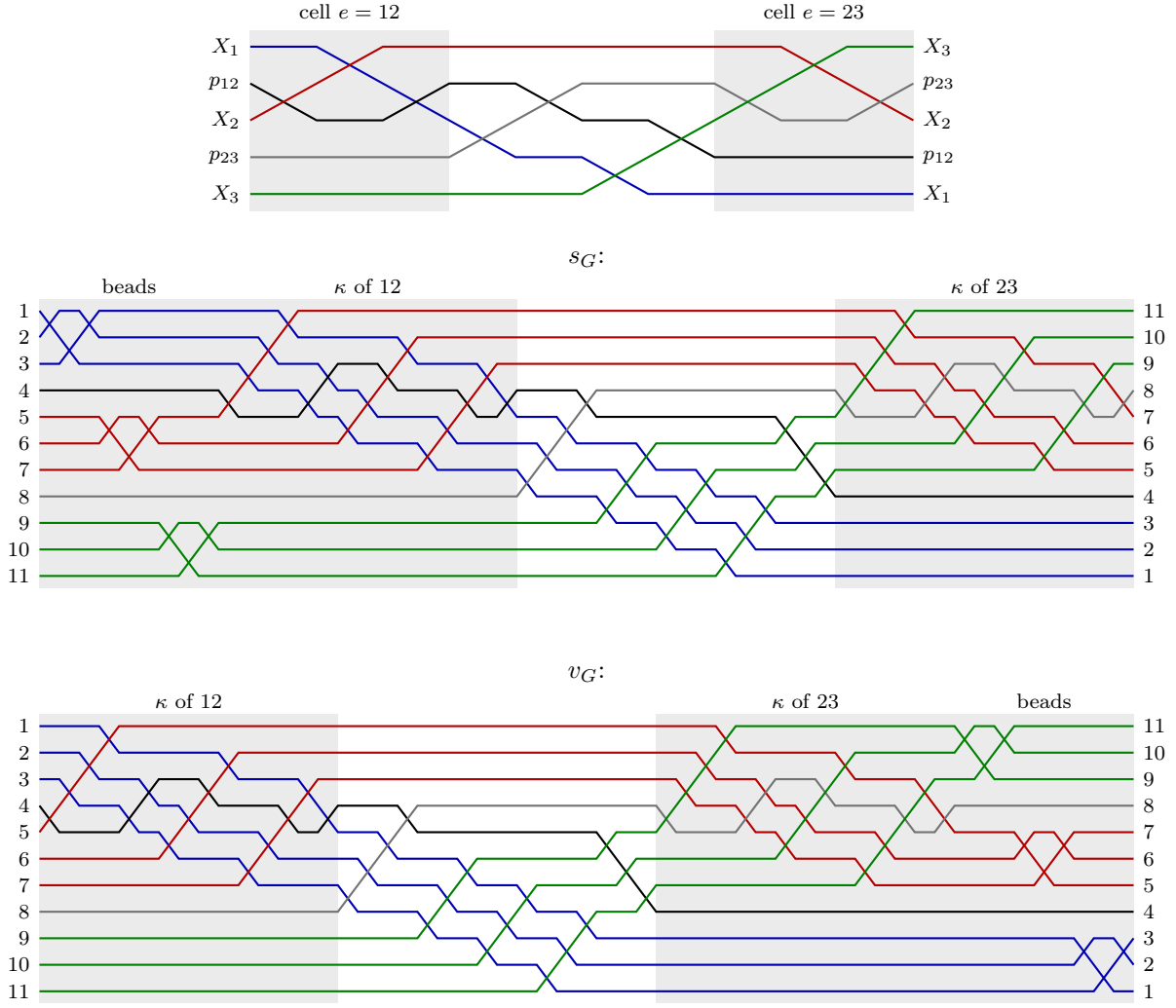

\subsection{The blow-up and the two words}\label{sec:blowup}

We now replace each group by its wires, three for a block and one for a private wire, and each step of $\Phi$ (a
move, or a whole cell) by its \emph{factor}, a word on the positions of these wires: when $h$ passes up through $g$,
the wires of $h$ pass up one after another through the wires of $g$, and a cell becomes $\kappa$. With $q$ the top
position of the upper group, the factors are:
\begin{itemize}
\item[(pp)] a private wire passes up through a private wire: the factor is $\sigma_q$;
\item[(Xp)] a private wire passes up through a block: the factor is $\sigma_{q+2}\sigma_{q+1}\sigma_q$;
\item[(pX)] a block passes up through a private wire: the factor is $\sigma_q\sigma_{q+1}\sigma_{q+2}$;
\item[(XX)] a block passes up through a block: the factor is
$\sigma_{q+2}\sigma_{q+1}\sigma_q\cdot\sigma_{q+3}\sigma_{q+2}\sigma_{q+1}\cdot\sigma_{q+4}\sigma_{q+3}\sigma_{q+2}$;
\item[(cell)] the cell of an edge $e$: the factor is $\kappa$, shifted to the seven positions $q,\dots,q+6$ of $A_e$.
\end{itemize}
Every factor exchanges each pair of wires from different groups once and keeps the order inside each group (for
$\kappa$ by (K1) and (K2)). Call the resulting word $\widehat\Phi=\phi_1\phi_2\cdots\phi_J$.

\paragraph{The beads and the two words.}

In $\widehat\Phi$ the three wires of a block are never exchanged with each other, and between any two factors they
occupy three consecutive positions. This lets us insert the bead of the block between any two factors; it supplies
the three missing swaps. For $0\le k\le J$
and $u\in[n]$ let $q_u(k)$ be the top position of $X_u$ after $\phi_1\cdots\phi_k$; the block occupies the
positions $q_u(k),q_u(k)+1,q_u(k)+2$. For $\epsilon\colon[n]\to\{0,1\}$ put
\[
\mathcal B^\epsilon_k=\prod_{u\in[n]}\beta^{\epsilon_u}_{q_u(k)}\qquad\text{and}\qquad
W^\epsilon_k=\phi_1\cdots\phi_k\;\mathcal B^\epsilon_k\;\phi_{k+1}\cdots\phi_J ,
\]
i.e.\ $W^\epsilon_k$ is the word $\widehat\Phi$ with all beads inserted after its $k$-th factor, the bead of
$u$ of type $\epsilon_u$. Beads of different blocks act on disjoint position triples and their letters differ
by at least $2$, so the order of the product is irrelevant up to $\equiv$. Finally let
\[
W^s_G=W^0_0,\qquad W^v_G=W^0_J,\qquad s_G=s_{W^s_G},\qquad v_G=s_{W^v_G}.
\]
So all beads, of type $0$, stand at the very beginning of $W^s_G$ and at the very end of $W^v_G$
(Figure~\ref{fig:p3}). In Section~\ref{sec:braid} we move the beads through the diagram one factor at a time,
$W^\epsilon_0,W^\epsilon_1,\dots,W^\epsilon_J$; Section~\ref{sec:onestep} describes what a single step changes.

\begin{lemma}\label{lem:blowup}
In every $W^\epsilon_k$, before and after each factor, the groups occupy consecutive positions, in the current
order of $\Phi$. Every pair of wires is swapped exactly once: a pair from different groups in the factor of the step
of $\Phi$ that exchanges the two groups, and a pair inside a block in its bead. In particular $W^\epsilon_k$ is
reduced.
\end{lemma}
\begin{proof}
Each factor acts on the positions of its groups and swaps exactly the pairs of wires from different ones of them,
keeping the internal orders, and a bead reverses its block in place. Induction along the word gives the first
statement, and Lemma~\ref{lem:phi} (every pair of groups is exchanged by exactly one step) the second.
\end{proof}

\subsection{One step}\label{sec:onestep}
In Section~\ref{sec:braid} we go from $W^0_0$ to $W^0_J$ by braid moves in three phases: turn over the beads of a
vertex cover, go from $W^\epsilon_0$ to $W^\epsilon_J$ through $W^\epsilon_1,\dots,W^\epsilon_{J-1}$, and turn
the beads back. A step from $W^\epsilon_k$ to $W^\epsilon_{k+1}$ moves the beads past the factor $\phi_{k+1}$, and it
involves only a few wires. Let $Z$ be the set of wires of the groups of $\phi_{k+1}$, and let $b$ and $b^+$ be the
products of the beads of the blocks among these groups, at their positions before and after $\phi_{k+1}$. All other
beads commute with $\phi_{k+1}$, since their letters differ from those of $\phi_{k+1}$ by at least $2$. So the step
only replaces $b\,\phi_{k+1}$ by $\phi_{k+1}b^+$, a word on the wires of $Z$. For a factor of kind (pp) there is no
bead and nothing changes; for the other kinds $Z$ has $4$, $6$ or $7$ wires. Let $D_k=D(s_{W^\epsilon_k},
s_{W^\epsilon_{k+1}})$ be the set of triples whose sign the step changes.

\begin{lemma}[one step]\label{lem:step}
Let $\epsilon\colon[n]\to\{0,1\}$ and $0\le k<J$. Then $s_{W^\epsilon_k}$ and $s_{W^\epsilon_{k+1}}$ agree outside
$\binom Z3$, and $D_k$ and the number of braid moves needed to go from $W^\epsilon_k$ to $W^\epsilon_{k+1}$ are as
follows:
\begin{center}\small
\begin{tabular}{llllc}\toprule
kind of $\phi_{k+1}$ & $|Z|$ & triples of $D_k$ & $|D_k|$ & braid moves\\\midrule
(pp) & $2$ & none & $0$ & $0$\\
(Xp), (pX) & $4$ & two wires in the block, one in the private wire & $3$ & $3$\\
(XX) & $6$ & two wires in one block, one in the other & $18$ & $18$\\
(cell) of $uw$ & $7$ & the copy of $D_\kappa$ in $A_{uw}$ & $30$ & $30$ if $\epsilon_u=1$ or $\epsilon_w=1$\\\bottomrule
\end{tabular}
\end{center}
\end{lemma}
\begin{proof}
\emph{Reduction to $Z$.} The other beads act on positions outside the interval of positions of $Z$, and their
letters differ from those of $\phi_{k+1}$, $b$ and $b^+$ by at least $2$. Using commutations we may therefore write
$W^\epsilon_k\equiv\mathsf A\,L\,\mathsf C$ and $W^\epsilon_{k+1}\equiv\mathsf A\,L'\,\mathsf C$ with $L=b\,\phi_{k+1}$
and $L'=\phi_{k+1}b^+$. After $\mathsf A$ the wires of $Z$ occupy consecutive positions in label order
(Lemma~\ref{lem:blowup}): the groups of $\phi_{k+1}$ have not been exchanged yet, and a block is reversed only by its
bead, which lies in $L$. By Lemma~\ref{lem:replace}, $s_{W^\epsilon_k}$ and $s_{W^\epsilon_{k+1}}$ agree outside
$\binom Z3$, their restrictions to $Z$ are $s_L$ and $s_{L'}$ (reading $L,L'$ as words on the $|Z|$ wires), and a
sequence of braid moves from $L$ to $L'$ gives one from $W^\epsilon_k$ to $W^\epsilon_{k+1}$ with the same number
of braid moves. Since every such sequence has at least $|D_k|$ braid moves, it remains to compare $L$ and $L'$,
words on at most $7$ wires.

\emph{The cases.} (pp): $L=L'$. In the other cases no bead triple changes: $\phi_{k+1}$ swaps no pair inside a block,
so in $L$ and in $L'$ each bead acts on its block in label order and gives its bead triple the sign of its type.

(Xp), (pX), (XX): a triple with two wires $x,x'$ in a block and its third wire in the other group changes sign. Its
swap $\sw{xx'}$ lies in $b$, before its two other swaps, in $L$, and in $b^+$, after them, in $L'$; as the labels
of a block are consecutive, $x,x'$ are the two smallest or the two largest wires of the triple, and by
Definition~\ref{def:words} the sign of the triple records whether their swap comes first or last. These triples
and the bead triples are all triples of $Z$, which gives $D_k$. For (Xp) and (pX), $Z$ is a single $4$-set and $D_k$
consists of three of its triples, so the unique arc order of $(s_L[Z],s_{L'}[Z])$ is admissible, and the packet
criterion (Lemma~\ref{lem:packet}) gives $3$ braid moves. For (XX), let $X$ be the block above $Y$ before
$\phi_{k+1}$; its labels are smaller than those of $Y$. In $L$ the swaps inside the blocks come first, so $s_L$ is
$+$ on the triples $xx'y$ and $-$ on the triples $xyy'$ ($x,x'\in X$, $y,y'\in Y$), and $s_{L'}$ is the opposite.
Flip first, for $y\in Y$ in decreasing order, the three triples $xx'y$ in the arc order of the $4$-set $X\cup\{y\}$,
and then, for $x\in X$ in decreasing order, the three triples $xyy'$ in the arc order of the $4$-set $\{x\}\cup Y$.
This enumeration of $D_k$ is admissible: the $4$-subsets of $Z$ are $X\cup\{y\}$, $\{x\}\cup Y$, and
$\{x,x',y,y'\}$ with $x<x'$ and $y<y'$; the first two kinds are respected by construction, and for the third the
packet goes from $(+,+,-,-)$ to $(-,-,+,+)$, and the order flips $xx'y'$, $xx'y$, $x'yy'$, $xyy'$ in this order,
which is the arc through $(+,-,-,-)$, $(-,-,-,-)$, $(-,-,-,+)$. So the packet criterion gives $18$ braid moves.

(cell): with the positions of $A_{uw}$ shifted to $1,\dots,7$, we have
$L=\beta^{\epsilon_u}_1\beta^{\epsilon_w}_5\kappa$ and $L'=\kappa\,\beta^{\epsilon_w}_1\beta^{\epsilon_u}_5$,
because $\kappa$ moves $X_w$ to the positions $1$--$3$ and $X_u$ to $5$--$7$ (property (K2)). These are the gadget
words of Definition~\ref{def:gadget}, and Proposition~\ref{prop:gadget} gives $D_k$ and, if
$(\epsilon_u,\epsilon_w)\ne(0,0)$, $30$ braid moves.

These finitely many local statements, for all types of the beads involved, were also checked by computer (see
\hyperref[sec:computations]{\emph{Computations}}).
\end{proof}

\paragraph{The steps change disjoint sets of triples.}
The following corollary is what makes the moves of the beads add up to exactly $|D|$ braid moves.

\begin{corollary}\label{cor:disjoint}
For every $\epsilon$ the sets $D_0,\dots,D_{J-1}$ are pairwise disjoint, and their union is $D=D(s_G,v_G)$. Hence
$|D|=3n(m-3)+6|E|$.
\end{corollary}
\begin{proof}
Let $t\in D_k$. By Lemma~\ref{lem:step}, $t\subseteq Z$. The three wires of $t$ do not all lie in one block: the
factor $\phi_{k+1}$ keeps the order inside each block, so the sign of a bead triple is fixed by the type of its bead,
which the step does not change. So $t$ has wires in at least two groups of $\phi_{k+1}$, at least two of its three
pairs join different groups, and $\phi_{k+1}$ swaps both. If $t$ also lay in $D_l$ with $l\ne k$, two of its pairs would
also be swapped in $\phi_{l+1}$; since $t$ has only three pairs, one pair would be swapped twice, contrary to
Lemma~\ref{lem:blowup}. So the sets $D_k$ are pairwise disjoint, every triple changes sign in at most one step, and
$\bigcup_kD_k=D(s_{W^\epsilon_0},s_{W^\epsilon_J})$. This is $D$, because $s_{W^\epsilon_0}$ and $s_{W^\epsilon_J}$
arise from $s_G$ and $v_G$ by negating the same bead triples (property~(B) in Section~\ref{sec:gadget}).

For the count: by Lemma~\ref{lem:phi}, $\widehat\Phi$ has $|E|$ factors (cell), $\binom n2-|E|$ factors (XX), and
$(n-2)|E|$ factors (Xp) or (pX), since each cell contains two of the $n|E|$ pairs of a block and a private wire.
Lemma~\ref{lem:step} gives $|D|=30|E|+18\bigl(\binom n2-|E|\bigr)+3(n-2)|E|=3n(m-3)+6|E|$.
\end{proof}

\begin{corollary}\label{cor:edgeD}
For every edge $e$, deleting all wires outside $A_e$ turns $W^0_0$ into
$\beta_1\beta_5\kappa$ and $W^0_J$ into a word equivalent to $\kappa\beta_1\beta_5$. In particular $s_G|_{A_e}=g^s$
and $v_G|_{A_e}=g^v$.
\end{corollary}
\begin{proof}
Let $e=uw$. By Lemma~\ref{lem:blowup}, the only letters of $W^0_0$ and $W^0_J$ that swap two wires of $A_e$ are those
of the beads of $X_u$ and $X_w$ and of the cell of $e$, before which the beads stand in $W^0_0$ and after which they
stand in $W^0_J$. So deleting the wires outside $A_e$ gives the words $L$ and $L'$ of the step of the cell of $e$
for $\epsilon=0$, up to commuting the two beads, which are $\beta_1\beta_5\kappa$ and $\kappa\beta_1\beta_5$ by the
proof of Lemma~\ref{lem:step}. Their sign vectors are $g^s$ and $g^v$ (Definition~\ref{def:gadget}).
\end{proof}

\begin{example}\label{ex:p3}
We continue Example~\ref{ex:p3wires} (Figure~\ref{fig:p3}). In round $1$ the bundle $B_2=(X_2,p_{23})$ passes up
through $B_1=(X_1,p_{12})$: first the cell of $12$ is applied, then $p_{23}$ passes up through $B_1$. Then $B_3=(X_3)$
passes up through $B_1$. In round $2$ the bundle $B_3$ passes up through $B_2$, with the cell of $23$. The gadgets
$A_{12}=\{1,\dots,7\}$ and $A_{23}=\{5,\dots,11\}$ share only the block $X_2$, and hence only the triple $567$;
$X_2$ plays the role of $V$ in the gadget of $12$ and of $U$ in the gadget of $23$. By Corollary~\ref{cor:disjoint},
$|D|=3\cdot3\cdot8+6\cdot2=84$.
\end{example}

\section{Braid-move distance}\label{sec:braid}

Let $G$ be a graph with vertex set $[n]$ and edge set $E$, let $W^s_G$ and $W^v_G$ be the two reduced words on
$m=3n+|E|$ wires constructed in Section~\ref{sec:construction}, and let $D=D(s_G,v_G)$. This section proves
Theorem~A for these words.

\begin{proposition}[lower bound]\label{prop:lowerW}
For every graph $G$, $\dbr([W^s_G],[W^v_G])\ge|D|+2\VC(G)$.
\end{proposition}
\begin{proof}
Fix a sequence of commutations and braid moves from $W^s_G$ to $W^v_G$. For a triple $T\in\binom{[m]}3$ let $c_T$
be the number of braid moves of the sequence that flip $T$, and let
\[
Y=\bigl\{T\in\tbinom{[m]}3:\ c_T>[T\in D]\bigr\}
\]
as in Lemma~\ref{lem:parity}. By that lemma, a triple of $D$ is flipped an odd number of times and lies in $Y$ if
and only if it is flipped at least three times, and a triple outside $D$ is flipped an even number of times and lies
in $Y$ if and only if it is flipped at least twice. The lemma also shows that the sequence has at least $|D|+2|Y|$
braid moves, so it suffices to show $|Y|\ge\VC(G)$. We construct a vertex cover of $G$ with at most $|Y|$ vertices.

\emph{Every gadget contains a triple of $Y$.} Let $e$ be an edge. By Corollary~\ref{cor:edgeD}, deleting all wires outside $A_e$
turns $W^s_G$ and $W^v_G$ into the gadget words $\beta_1\beta_5\kappa$ and $\kappa\beta_1\beta_5$ (the latter up to commutation). By
Lemma~\ref{lem:delete}, the fixed sequence from $W^s_G$ to $W^v_G$ yields a sequence of commutations and braid
moves from $\beta_1\beta_5\kappa$ to $\kappa\beta_1\beta_5$ with $\sum_{T\subseteq A_e}c_T$ braid moves, and
exactly $|D_\kappa|=30$ triples inside $A_e$ lie in $D$ (Proposition~\ref{prop:gadget}(a)). If no triple inside
$A_e$ lay in $Y$, this sequence would have $30$ braid moves, contrary to Proposition~\ref{prop:gadget}(b). So some
triple inside $A_e$ lies in $Y$.

\emph{Overlaps.} For distinct edges $e,f$, a triple inside both $A_e$ and $A_f$ lies in $A_e\cap A_f$. If
$e\cap f=\{u\}$, then $A_e\cap A_f=X_u$, and the only such triple is the bead triple $X_u$; if $e\cap f=\emptyset$
there is none. Moreover, the only bead triples inside $A_e$, for $e=uw$, are $X_u$ and $X_w$.

\emph{A vertex cover.} Call an edge $uw$ \emph{charged} if $X_u\in Y$ or $X_w\in Y$, that is, if the bead triple of
one of its endpoints lies in $Y$. Let
\[
\mathcal C=\{u\in[n]:X_u\in Y\}\ \cup\ \{\text{one chosen endpoint of each edge that is not charged}\}.
\]
A charged edge has an endpoint in the first set, and every other edge has its chosen endpoint in the second set, so
$\mathcal C$ is a vertex cover. A charged edge may contain further triples of $Y$ besides the bead triple; they are
not needed, since the edge is already covered.

It remains to show $|\mathcal C|\le|Y|$. We assign to each vertex $u$ of the first set the triple $X_u\in Y$, and to
each edge $e=uw$ that is not charged a triple $T_e\in Y$ inside $A_e$, which exists by the first step. Since
$X_u,X_w\notin Y$, the triple $T_e$ is neither $X_u$ nor $X_w$, so it is not a bead triple, and by the overlap step it
lies in no $A_f$ with $f\ne e$. Hence the triples $T_e$ of different edges are different, and they differ from the
bead triples $X_u$. So these assignments are injective into $Y$, and $|\mathcal C|\le|Y|$. Hence
$\VC(G)\le|\mathcal C|\le|Y|$.
\end{proof}

\begin{proposition}[upper bound]\label{prop:upper}
For every graph $G$ and every vertex cover $\mathcal C$ of $G$,
\[
\dbr([W^s_G],[W^v_G])\le|D|+2|\mathcal C|.
\]
\end{proposition}
\begin{proof}
Recall from Section~\ref{sec:construction} that $\widehat\Phi=\phi_1\cdots\phi_J$ is the blow-up of the shadow
sweep, and that for $\epsilon\colon[n]\to\{0,1\}$ and $0\le k\le J$ the word
\[
W^\epsilon_k=\phi_1\cdots\phi_k\;\mathcal B^\epsilon_k\;\phi_{k+1}\cdots\phi_J
\]
is $\widehat\Phi$ with all beads inserted after its $k$-th factor, the bead of each vertex $u$ being of type
$\epsilon_u$, so that it is turned over exactly when $\epsilon_u=1$. In particular
$W^s_G=W^0_0$ has all beads, of type $0$, before the first factor, and $W^v_G=W^0_J$ has them after the last
one. Let $\epsilon$ be the indicator function of $\mathcal C$, so that $\epsilon_u=1$ exactly for $u\in\mathcal C$.
We go from $W^s_G$ to $W^v_G$ in three phases,
\[
W^s_G=W^0_0\ \xrightarrow{\ |\mathcal C|\ }\ W^\epsilon_0\ \xrightarrow{\ |D|\ }\ W^\epsilon_J\
\xrightarrow{\ |\mathcal C|\ }\ W^0_J=W^v_G ,
\]
where the labels are the numbers of braid moves: we turn over the beads of the vertices in $\mathcal C$, move all
beads through $\widehat\Phi$ one factor at a time, and turn the beads of $\mathcal C$ back. Turning over the beads
of a vertex cover is what makes the middle phase possible with $|D|$ braid moves: in every cell at least one of the
two beads involved is turned over, and by Proposition~\ref{prop:gadget}(b) the gadget then has no excess.

\emph{Phases 1 and 3.} In $W^0_0$ the beads stand before $\phi_1$, the bead of $u$ being $\beta_{q_u(0)}$, where
$q_u(0)$ is the top position of the block $X_u$. Replacing it by $\beta'_{q_u(0)}$ is one braid move, which flips
the bead triple of $X_u$ (property~(B) in Section~\ref{sec:gadget}). Doing this for every $u\in\mathcal C$ turns
$W^0_0$ into $W^\epsilon_0$ with $|\mathcal C|$ braid moves. In the same way, replacing $\beta'_{q_u(J)}$ by
$\beta_{q_u(J)}$ for every $u\in\mathcal C$ turns $W^\epsilon_J$ into $W^0_J$ with $|\mathcal C|$ braid moves.

\emph{Phase 2.} We go from $W^\epsilon_0$ to $W^\epsilon_J$ through $W^\epsilon_1,\dots,W^\epsilon_{J-1}$. The step
from $W^\epsilon_k$ to $W^\epsilon_{k+1}$ moves the beads past the factor $\phi_{k+1}$. Since $\mathcal C$ covers
every edge, $\epsilon_u=1$ or $\epsilon_w=1$ for every cell of an edge $uw$, so by Lemma~\ref{lem:step}(b) every step
uses exactly $|D_k|$ braid moves. The sets $D_0,\dots,D_{J-1}$ are pairwise disjoint with union $D$
(Corollary~\ref{cor:disjoint}), so the whole phase uses $\sum_k|D_k|=|D|$ braid moves.

Altogether the three phases use $|\mathcal C|+|D|+|\mathcal C|$ braid moves.
\end{proof}

\begin{theorem}\label{thm:wiring}
For every graph $G$ with $n$ vertices and $m=3n+|E|$,
\[
\dbr([W^s_G],[W^v_G])=\dB(s_G,v_G)=|D|+2\VC(G),\qquad |D|=3n(m-3)+6|E|.
\]
\end{theorem}
\begin{proof}
Propositions~\ref{prop:lowerW} and~\ref{prop:upper}, Corollary~\ref{cor:disjoint} for $|D|$, and
Theorem~\ref{thm:corr} for $\dB$.
\end{proof}

Let $\textup{\textsc{Braid}}$ be the problem: given two reduced words $W,W'$ of $w_0\in S_m$ and an integer $K$ (in
binary), decide whether $\dbr([W],[W'])\le K$. Let $\FLIP_{\B}$ be the same problem for two maps
$s,v\colon\binom{[m]}3\to\{\pm\}$ in $\B(m,2)$, given as sign vectors of length $\binom m3$.

\begin{theorem}\label{thm:braidNP}
$\textup{\textsc{Braid}}$ and $\FLIP_{\B}$ are NP-complete.
\end{theorem}
\begin{proof}
We reduce from \textup{\textsc{Vertex Cover}}, which is NP-complete~\cite{Karp}. Given $(G,k)$ with $k\le n$,
compute $W^s_G$ and $W^v_G$: $\Phi$ has $O((n+|E|)^2)$ steps, the words have $\binom m2$ letters, and $s_G,v_G$
have $\binom m3$ entries. Put $K=|D|+2k$. By Theorem~\ref{thm:wiring}, $\VC(G)\le k$ if and only if the distance is
at most $K$, in both settings.

Membership: given reduced words $W,W'$, the sign vectors $s_W,s_{W'}$ are computed in polynomial time
(Definition~\ref{def:words}), and $\dbr([W],[W'])=\dB(s_W,s_{W'})$ by Theorem~\ref{thm:corr}; so it suffices to
show that $\FLIP_{\B}$ is in NP. The diameter of $\B(m,2)$ is $\binom m3$~\cite[Thm.~9]{GMS} (see
also~\eqref{eq:diamB} in Section~\ref{sec:signotopes}), so the answer is yes if $K\ge\binom m3$. Otherwise a certificate is a sequence $s=s_0,s_1,\dots,s_\ell=v$ with
$\ell\le K$, each term obtained from the previous one by negating one triple. That every term lies in $\B(m,2)$,
i.e.\ has at most one sign change in each packet, is checked in time $O(m^4)$ per term.
\end{proof}

\begin{corollary}\label{cor:reform}
The following problems are NP-complete.
\begin{itemize}
\item Triangle-flip distance between two marked simple arrangements of $m$ pseudolines (Question~5
of~\cite{Radtke}); as in~\cite{Radtke}, an arrangement in the Euclidean plane is \emph{marked} by the choice of an
unbounded cell.
\item Flip distance between two rhombic tilings of the regular $2m$-gon with fixed edge directions.
\end{itemize}
In both problems the objects are given by their sign vectors, or by any standard representation that is
polynomial-time interconvertible with the sign vectors, such as reduced words or wiring diagrams, and for tilings
the list of tiles.
\end{corollary}
\begin{proof}
By Theorem~\ref{thm:corr}(c) and its sources~\cite{MS,Ziegler,Elnitsky,FW,Radtke}, both flip graphs are
isomorphic to $\B(m,2)$ via the sign vectors, so with sign vectors as input both problems are $\FLIP_{\B}$
(Theorem~\ref{thm:braidNP}). For another representation, the reduction outputs the sign vectors $s_G,v_G$ converted
to that representation; for membership, the inputs are converted to sign vectors and the certificate of
Theorem~\ref{thm:braidNP} is used. Both conversions take polynomial time.
\end{proof}

\begin{remark}\label{rem:apx}
The reduction is exact: the excess $\dbr-|D|$ equals $2\VC(G)$. Since \textup{\textsc{Vertex Cover}} is APX-hard
on cubic graphs~\cite{AlimontiKann}, there is $\delta>0$ such that approximating the excess
$\dbr([W],[W'])-|D(s_W,s_{W'})|$ within a factor $1+\delta$ is NP-hard. For cubic graphs,
$|D|=\tfrac23m^2$ while the excess is at most $2n<m$, so on these instances $|D|$ approximates $\dbr$ itself within
a factor $1+O(1/m)$; the hardness concerns the excess only.
\end{remark}

\part{Oriented matroids and higher Bruhat orders}\label{part:om}

\section{Chirotopes, signotopes and the positive fibre}\label{sec:om}

A wiring diagram becomes a simple arrangement of pseudolines in the projective plane when the line at infinity is
added, and uniform oriented matroids are the combinatorial abstraction of such arrangements. By the topological
representation theorem of Folkman and Lawrence, uniform oriented matroids of rank $r$ correspond to simple
arrangements of oriented pseudospheres in $S^{r-1}$, and in rank $3$, up to reorientation, to simple arrangements of
pseudolines in the projective plane~\cite[Sect.~5.2 and Ch.~6]{BLSWZ}. The realizable ones come from arrangements of
oriented central hyperplanes in general position in $\mathbb R^r$, or equivalently from their normal vectors.

The higher Bruhat orders $\B(n,k)$ of Manin and Schechtman, whose elements are the signotopes of rank $k+1$,
generalize the flip graph $\B(n,2)$ of wiring diagrams in another direction. The two generalizations are connected
by the positive fibre, which turns a signotope into a chirotope by adding an element $\infty$ (in rank $3$, the line
at infinity) and makes $\B(n,k)$ an induced subgraph of the mutation graph of rank $k+1$ on $n+1$ elements
(Section~\ref{sec:fibre}). Section~\ref{sec:rank3} transfers the rank-$3$ results of Part~\ref{part:wiring} to
oriented matroids, and Section~\ref{sec:higher} proves hardness in higher ranks for both.

\subsection{Chirotopes and the mutation graph}\label{sec:chirotopes}

For vectors $x_e\in\mathbb R^r$, $e\in\Omega$, in general position, the sign of the determinant of any $r$ of them,
taken in increasing order of their labels, is nonzero; the map that assigns this sign to each $r$-subset is the
\emph{chirotope} of the configuration. Chirotopes in general are defined by the sign conditions of the
Grassmann--Pl\"ucker relations.

\begin{definition}[uniform chirotopes]\label{def:chirotope}
Let $\Omega$ be a finite totally ordered set and $2\le r\le|\Omega|$. A map $\chi\colon\binom\Omega r\to
\{\pm\}$ is extended to ordered $r$-tuples of distinct elements by
$\chi(e_1,\dots,e_r)=\sgn(\pi)\,\chi(\{e_1,\dots,e_r\})$, where $\pi$ sorts $(e_1,\dots,e_r)$. For an
$(r-2)$-subset $\sigma=\{\sigma_1<\dots<\sigma_{r-2}\}$ write
$\chi(\sigma,x,y)=\chi(\sigma_1,\dots,\sigma_{r-2},x,y)$. Then $\chi$ is a \emph{uniform chirotope of rank
$r$} if for every $(r-2)$-subset $\sigma\subseteq\Omega$ and all $a<b<c<d$ in $\Omega\setminus\sigma$
\[
\chi(\sigma,a,b)\chi(\sigma,c,d),\qquad -\chi(\sigma,a,c)\chi(\sigma,b,d),\qquad
\chi(\sigma,a,d)\chi(\sigma,b,c)\quad\text{are not all equal.}\tag{GP}
\]
The $r$-subsets of $\Omega$ are the \emph{bases}.
\end{definition}

(GP) is the sign form of the three-term Grassmann--Pl\"ucker relation
$[\sigma ab][\sigma cd]-[\sigma ac][\sigma bd]+[\sigma ad][\sigma bc]=0$. A permutation of $a,b,c,d$ permutes the
three terms and multiplies them by a common sign, so (GP) holds for all distinct $a,b,c,d$ once it holds for
$a<b<c<d$. The maps satisfying (GP) are exactly the chirotopes of the uniform oriented matroids of rank $r$ on
$\Omega$~\cite[Thm.~3.6.2]{BLSWZ}. A uniform chirotope is \emph{realizable} if $\chi(B)=\sgn\det(x_b)_{b\in B}$ for
some vectors $x_b\in\mathbb R^r$, $b\in\Omega$, with the columns in increasing order; conversely, for any vectors
with all these determinants nonzero, this map satisfies (GP) by the Grassmann--Pl\"ucker identity.

A chirotope $\chi$ and its negative $-\chi$ describe the same oriented matroid (a reflection of $\mathbb R^r$
negates all determinants), and we write $[\chi]=\{\chi,-\chi\}$ for the uniform oriented matroid. A
\emph{mutation} changes the sign of exactly one basis in such a way that the result is again a chirotope.
Geometrically, the $r$ pseudospheres of the basis bound a simplicial cell (a cell with exactly $r$ facets), and the
mutation moves one of them across the vertex in which the other $r-1$ meet; in rank $3$ a mutation is a triangle
flip, as for wiring diagrams.

\begin{definition}[mutation graph]\label{def:mutgraph}
The \emph{mutation graph} $\G_{\Omega,r}$ has the uniform oriented matroids $[\chi]$ of rank $r$ on $\Omega$ as
vertices, with $[\chi]\sim[\psi]$ if some representatives differ in exactly one basis. Its distance
$d([\chi],[\psi])$ is the \emph{mutation distance}. For $\Omega=[n]$ we write $\G_{n,r}$.
\end{definition}

\paragraph{Walks of chirotopes.} In the proofs it is convenient to work with chirotopes rather than with pairs
$\{\chi,-\chi\}$. A \emph{walk of chirotopes} from $\chi$ to $\psi$ is a sequence $\chi=\chi_0,\chi_1,\dots,
\chi_\ell=\psi$ of uniform chirotopes of rank $r$ on $\Omega$ in which consecutive terms differ in exactly one
basis; the step from $\chi_{i-1}$ to $\chi_i=\chi_{i-1}^B$ \emph{flips} $B$. Let $\dt(\chi,\psi)$ be the least
length of such a walk. As for wiring diagrams, the \emph{Hamming distance} of $\chi$ and $\psi$ is $|D(\chi,\psi)|$;
by Lemma~\ref{lem:parity}, $\dt(\chi,\psi)\ge|D(\chi,\psi)|$, with the same parity, and we call the difference the
\emph{excess}. The quantity $\dt$ is only a tool; the mutation distance is recovered from it as follows.

\begin{observation}\label{obs:cover}
For uniform chirotopes $\chi,\psi$ of rank $r$ on $\Omega$,
\[
d([\chi],[\psi])=\min\{\dt(\chi,\psi),\dt(\chi,-\psi)\},\qquad
\dt(\chi,-\psi)\ge\tbinom{|\Omega|}r-|D(\chi,\psi)|.
\]
\end{observation}
Negation commutes with flips, so a walk of oriented matroids starting at $[\chi]$ lifts step by step to a walk of
chirotopes starting at $\chi$, which ends at $\psi$ or at $-\psi$. The inequality is Lemma~\ref{lem:parity}, since
$\chi$ and $-\psi$ differ exactly outside $D(\chi,\psi)$.

As in the introduction, $\DIST_r$ is the problem: given uniform chirotopes $\chi_1,\chi_2$ of rank $r$ on $[n]$, as
tables of their $\binom nr$ values, and an integer $K$, decide whether $d([\chi_1],[\chi_2])\le K$. Whether $K$ is
written in binary or in unary matters in ranks $r\ge4$, where we know no polynomial bound on the distances.

\begin{definition}[minors and symmetries]
Let $\chi$ be a uniform chirotope of rank $r$ on $\Omega$ and $e\in\Omega$. For $A\subseteq\Omega$ with
$|A|\ge r$, the \emph{restriction} $\chi|_A$ is the restriction of $\chi$ to $\binom Ar$, and the
\emph{deletion} is $\chi\setminus e=\chi|_{\Omega\setminus e}$ (if $|\Omega|>r$). For $r\ge3$ the
\emph{contraction} is $(\chi/e)(Y)=\chi(e,y_1,\dots,y_{r-1})$ for $Y=\{y_1<\dots<y_{r-1}\}\subseteq
\Omega\setminus e$. For $R\subseteq\Omega$, the \emph{reorientation} ${}_{-R}\chi$ negates $\chi(B)$ exactly when
$|B\cap R|$ is odd. For a bijection $\pi\colon\Omega\to\Omega'$ onto another totally ordered set, the
\emph{relabelling} $\pi\chi$ is defined on ordered tuples by $(\pi\chi)(\pi e_1,\dots,\pi e_r)=
\chi(e_1,\dots,e_r)$.
\end{definition}

In the arrangement of pseudospheres, deletion removes the pseudosphere $e$, contraction restricts the arrangement
to $e$, and reorientation exchanges the two sides of the elements of $R$. Deletion is the analogue of deleting wires
from a wiring diagram.

\begin{lemma}[{standard, see~\cite[Ch.~3]{BLSWZ}}]\label{lem:minors}
Restrictions and deletions of uniform chirotopes of rank $r$ are uniform chirotopes of rank $r$, contractions are
uniform chirotopes of rank $r-1$, and reorientations and relabellings are uniform chirotopes of rank $r$.
Reorientations and relabellings preserve Hamming distances, commute with $\chi\mapsto-\chi$, preserve realizability,
and preserve $\dt$ and $d$.
\end{lemma}

The next lemma is the analogue for chirotopes of Lemma~\ref{lem:delete}: a walk induces walks on a deletion and on
a contraction, and each of its steps survives in exactly one of the two. It follows from Lemma~\ref{lem:minors},
since flipping $B$ flips $B$ in $\chi\setminus e$ if $e\notin B$, and flips $B\setminus e$ in $\chi/e$ if $e\in B$.

\begin{lemma}[splitting walks]\label{lem:split}
Let $\chi_0,\dots,\chi_\ell$ be a walk of chirotopes on $\Omega$, and $e\in\Omega$. With consecutive repetitions
removed, $(\chi_i\setminus e)_i$ is a walk of chirotopes on $\Omega\setminus e$ (if $|\Omega|>r$) with one step for
each step that flips a basis not containing $e$, and for $r\ge3$, $(\chi_i/e)_i$ is a walk of chirotopes of rank
$r-1$ on $\Omega\setminus e$ with one step for each step that flips a basis containing $e$. Likewise, for
$|A|\ge r$, $(\chi_i|_A)_i$ is a walk with one step for each step that flips a basis contained in $A$; in particular
$\dt(\chi|_A,\psi|_A)\le\dt(\chi,\psi)$.
\end{lemma}

Duality transfers results in rank $r$ to \emph{corank} $r$, that is, to uniform oriented matroids of rank $n-r$ on
$n=|\Omega|$ elements: the map $\chi\mapsto\chi^*$ to the dual chirotope is computable in polynomial time and induces
an isomorphism $\G_{\Omega,r}\cong\G_{\Omega,n-r}$~\cite[Sect.~3.4]{BLSWZ}.

\subsection{Signotopes and higher Bruhat orders}\label{sec:signotopes}

Signotopes of rank $3$ were defined in Section~\ref{sec:words}; they are the sign vectors of wiring diagrams, and
the packet of a $4$-set $P$ is the sign vector of the diagram on the four wires of $P$ (Section~\ref{sec:packets}).
The definition extends to all ranks. In rank
$r$ a packet lists the values on the $r$-subsets of an $(r+1)$-set, and a signotope is a map all of whose packets
change sign at most once. Signotopes of rank $r$ form the higher Bruhat order $\B(n,r-1)$, introduced by Manin and
Schechtman~\cite{MS} as a generalization of the weak order on permutations ($\B(n,1)$) and studied by
Ziegler~\cite{Ziegler}, Felsner and Weil~\cite{FW} and Rambau~\cite{Rambau}.

\begin{definition}[signotopes~\cite{FW}]\label{def:signotope}
Let $S$ be a finite totally ordered set and $r\ge1$. For $u\colon\binom Sr\to\{\pm\}$ and an
$(r+1)$-subset $P=\{p_1<\dots<p_{r+1}\}$ of $S$, the \emph{packet} of $P$ is the sequence
\[
u[P]=\bigl(u(P\setminus p_{r+1}),\,u(P\setminus p_r),\,\dots,\,u(P\setminus p_1)\bigr)
\]
of the values on the $r$-subsets of $P$ in lexicographic order. The map $u$ is a \emph{signotope of rank
$r$} if every packet has at most one sign change. We write $\B(S,r-1)$ for the graph whose vertices are the rank-$r$
signotopes on $S$, with $u$ and $u'$ adjacent (related by a \emph{flip}) if $|D(u,u')|=1$; its distance is $\dB$. We
write $\B(m,k)=\B([m],k)$. Packets of a subset $S'\subseteq S$ are packets of $S$, so the restriction $u|_{S'}$ of a
signotope to $\binom{S'}r$ is a signotope.
\end{definition}

For rank $3$ the packets and the graph $\B(m,2)$ are those of Section~\ref{sec:prelim}, and by
Theorem~\ref{thm:corr} the rank-$3$ signotopes are the sign vectors of wiring diagrams. Every signotope $u$ is
joined to the constant signotope $+$ by a walk of length $|\{X:u(X)=-\}|$, and to the constant signotope $-$ by a
walk of length $|\{X:u(X)=+\}|$~\cite[Sect.~2, Thm.~3(b)]{MS}, \cite[Thm.~4.1(G)]{Ziegler}; going through whichever
is shorter gives
\begin{equation}\label{eq:diamB}
\dB(u,v)\le\tbinom{|S|}r\qquad\text{for all }u,v\in\B(S,r-1).
\end{equation}
Walks in $\B(S,r-1)$, like walks of chirotopes, satisfy Lemma~\ref{lem:parity}. In particular $\dB(u,v)\ge|D(u,v)|$,
and, as in Part~\ref{part:wiring}, we call $\dB(u,v)-|D(u,v)|$ the \emph{excess} of $(u,v)$. Let
$\FLIP_{\B(n,k)}$ be the problem: given two signotopes of rank $k+1$ on $[n]$, as sign vectors of length
$\binom n{k+1}$, and an integer $K$, decide whether their distance in $\B(n,k)$ is at most $K$; for $k=2$ this is the
problem $\FLIP_\B$ of Theorem~\ref{thm:braidNP}.

\subsection{The positive fibre}\label{sec:fibre}

A wiring diagram on the wires $1,\dots,m$ becomes an arrangement of $m+1$ pseudolines in the projective plane when
the line at infinity is added. On the level of sign vectors, this adds an element $\infty$ above all wires and gives
every triple through it the sign $+$, which records that the wires meet the line at infinity in the order
$1,\dots,m$. The same construction works in every rank.

\begin{proposition}[positive fibre~\cite{KV,Ziegler,FW}]\label{prop:fibre}
Let $r\ge2$, let $S$ be a finite totally ordered set with $|S|\ge r$, and let $\Omega=S\cup\{\infty\}$ with
$S<\infty$. For $u\colon\binom Sr\to\{\pm\}$ define $\Pos_r(u)\colon\binom\Omega r\to\{\pm\}$ by
$\Pos_r(u)(B)=+$ if $\infty\in B$ and $\Pos_r(u)(B)=u(B)$ otherwise; we call $\Pos_r(u)$ the \emph{positive fibre}
of $u$. Then $\Pos_r(u)$ is a uniform chirotope if and only if $u$ is a rank-$r$ signotope.
\end{proposition}

If $u$ is a signotope, then so is $\Pos_r(u)$, since its packets are those of $u$ and the sequences
$(u(Q),+,\dots,+)$; and every signotope is a uniform chirotope~\cite[Sect.~2.2.5]{Bergold}, \cite{Miyata}, \cite[Sect.~2.1]{BFS}.
Conversely, the conditions (GP) of $\Pos_r(u)$ at $a<b<c<\infty$ say that no packet of $u$ alternates, that is, that
$u$ is a signotope. The proposition is a form of the correspondence between signotopes and single-element liftings
of the alternating oriented matroid~\cite{KV,Ziegler,FW}; in rank~$3$ it describes $\B(m,2)$ by arrangements of
pseudolines with a line at infinity.

We call a uniform chirotope on $\Omega$ \emph{positive at $\infty$} if it is $+$ on every basis containing $\infty$.
By Proposition~\ref{prop:fibre}, $\Pos_r$ is a bijection from the rank-$r$ signotopes on $S$ onto these chirotopes,
with $D(\Pos_ru,\Pos_rv)=D(u,v)$ and $\Pos_r(u)^X=\Pos_r(u^X)$ for $X\in\binom Sr$. Since $\Pos_r(u)$ and
$-\Pos_r(v)$ differ on all $\binom{|S|}{r-1}\ge2$ bases through $\infty$, it follows that $u\mapsto[\Pos_r(u)]$ is an
isomorphism from $\B(S,r-1)$ onto the subgraph of the mutation graph $\G_{S\cup\{\infty\},r}$ induced by the oriented
matroids that have a chirotope positive at $\infty$. In rank $3$, $\Pos_3(s)$ is the chirotope of the projective
arrangement obtained from the wiring diagram of $s$ by adding the line at infinity, and the flip of a triple of wires
is the mutation of the same basis.

This subgraph is only a small part of the mutation graph. A walk of chirotopes leaves it as soon as it flips a basis
through $\infty$ (in rank $3$, such a mutation moves the line at infinity across a crossing of two other
pseudolines), and we do not know whether this can make walks shorter. The following lemma shows that a walk without
excess never leaves it; its part (b) is the tool that transfers lower bounds from signotopes to chirotopes.

\begin{lemma}[walks and the positive fibre]\label{lem:sheet}
Let $u,v\in\B(S,r-1)$.
\begin{enumerate}
\item[(a)] $\Pos_r$ maps walks in $\B(S,r-1)$ to walks of chirotopes of the same length, so
$\dt(\Pos_ru,\Pos_rv)\le\dB(u,v)$.
\item[(b)] $\dt(\Pos_ru,\Pos_rv)=|D(u,v)|$ if and only if $\dB(u,v)=|D(u,v)|$.
\end{enumerate}
\end{lemma}
\begin{proof}
(a) For $X\in\binom Sr$ we have $\Pos_r(u)^X=\Pos_r(u^X)$. (b) Both distances are at least
$|D(u,v)|=|D(\Pos_ru,\Pos_rv)|$ (Lemma~\ref{lem:parity}), so the ``if'' direction follows from (a). Conversely, by
Lemma~\ref{lem:parity} a walk of chirotopes of length $|D(u,v)|$ from $\Pos_ru$ to $\Pos_rv$ flips only bases in
$D(u,v)$, none of which contains $\infty$. So every chirotope on it is positive at $\infty$, hence the positive fibre
of a signotope (Proposition~\ref{prop:fibre}), and these signotopes form a walk of the same length in $\B(S,r-1)$.
\end{proof}

\section{Rank three}\label{sec:rank3}

Adding the line at infinity to the two wiring diagrams of Theorem~\ref{thm:wiring} gives the uniform chirotopes
$\Pos_3(s_G)$ and $\Pos_3(v_G)$ on $[m]\cup\{\infty\}$, which are positive at $\infty$. We show that the mutation
distance of the oriented matroids $[\Pos_3(s_G)]$ and $[\Pos_3(v_G)]$ is the braid-move distance $|D|+2\VC(G)$ of
the two wiring diagrams, if $G$ has at least three vertices (Theorem~\ref{thm:rank3}). The proof has four steps.
\begin{itemize}
\item \emph{Upper bound.} By Lemma~\ref{lem:sheet}(a), a sequence of braid moves gives a walk of chirotopes of the
same length, so $\dt(\Pos_3(s_G),\Pos_3(v_G))\le|D|+2\VC(G)$.
\item \emph{Lower bound.} A walk of chirotopes from $\Pos_3(s_G)$ to $\Pos_3(v_G)$ may flip bases through $\infty$
and pass through chirotopes that are not positive at $\infty$, so the lower bound of Theorem~\ref{thm:wiring} cannot
simply be quoted; whether such walks can ever be shorter than walks in $\B(m,2)$ is open
(item~\ref{it:isometry} of Section~\ref{sec:remarks}). We avoid this question. The lower bound of
Proposition~\ref{prop:lowerW} only needs that every gadget forces two extra flips, and this can be checked on the
eight elements $A_e\cup\{\infty\}$ (Proposition~\ref{prop:lower}).
\item \emph{From chirotopes to oriented matroids.} By Observation~\ref{obs:cover} it remains to show that
$-\Pos_3(v_G)$ is farther from $\Pos_3(s_G)$ than $\Pos_3(v_G)$; this is a count of the Hamming distance, and it
needs at least three vertices.
\item \emph{Membership in NP.} For $K$ in binary, certificates need a polynomial bound on the mutation distances
between arbitrary rank-$3$ oriented matroids (Section~\ref{sec:diam}).
\end{itemize}

\begin{proposition}[lower bound for walks of chirotopes]\label{prop:lower}
Let $G$ be a graph with vertex set $[n]$ and edge set $E$, let $m=3n+|E|$, let $s_G,v_G\in\B(m,2)$ be the sign
vectors of the words $W^s_G$ and $W^v_G$ of Section~\ref{sec:construction}, and let $D=D(s_G,v_G)$. Then every walk
of chirotopes on $[m]\cup\{\infty\}$ from $\Pos_3(s_G)$ to $\Pos_3(v_G)$ has length at least $|D|+2\VC(G)$; that is,
\[
\dt\bigl(\Pos_3(s_G),\Pos_3(v_G)\bigr)\ge|D|+2\VC(G).
\]
\end{proposition}
\begin{proof}
We follow the proof of Proposition~\ref{prop:lowerW}, with chirotopes in place of words. Let $\chi_s=\Pos_3(s_G)$
and $\chi_v=\Pos_3(v_G)$; both are $+$ on the bases through $\infty$, so $D(\chi_s,\chi_v)=D$. Fix a walk of
chirotopes from $\chi_s$ to $\chi_v$, let $c_B$ be the number of its steps that flip the basis $B$, and define $Y$
as in Lemma~\ref{lem:parity}. The walk has length at least $|D|+2|Y|$, so it suffices to show $|Y|\ge\VC(G)$.

For an edge $e$, restricting the walk to $A_e\cup\{\infty\}$ (Lemma~\ref{lem:split}), with $A_e$ relabelled by
$[7]$, gives a walk from $\Pos_3(g^s)$ to $\Pos_3(g^v)$ (Corollary~\ref{cor:edgeD}) of length
$\sum_{B\subseteq A_e\cup\{\infty\}}c_B$, and exactly $|D_\kappa|=30$ of these bases lie in $D$. If none of them lay
in $Y$, this walk would have length $30$, so $\dt(\Pos_3(g^s),\Pos_3(g^v))=30$, and Lemma~\ref{lem:sheet}(b) would
give $\dB(g^s,g^v)=30$, contrary to
Lemma~\ref{lem:gadgetlower} and Theorem~\ref{thm:corr}. So $Y$ contains a basis inside $A_e\cup\{\infty\}$.

The rest is the argument of Proposition~\ref{prop:lowerW}, with the bead triple $X_u$ replaced by the set
$\mathcal T_u$ of the four bases inside $X_u\cup\{\infty\}$. These sets are pairwise disjoint; for distinct edges
$e,f$, the bases inside both $A_e\cup\{\infty\}$ and $A_f\cup\{\infty\}$ form $\mathcal T_u$ if $e\cap f=\{u\}$, and
there are none if $e\cap f=\emptyset$; and for $e=uw$ only $\mathcal T_u$ and $\mathcal T_w$ contain bases inside
$A_e\cup\{\infty\}$. Hence the vertices $u$ with $Y\cap\mathcal T_u\ne\emptyset$, together with one endpoint of each
edge $uw$ with $Y\cap(\mathcal T_u\cup\mathcal T_w)=\emptyset$, form a vertex cover of $G$ with at most $|Y|$
vertices.
\end{proof}

\begin{theorem}\label{thm:rank3}
Let $G$ be a graph with $n\ge3$ vertices and $m=3n+|E|$. In the mutation graph $\G_{[m]\cup\{\infty\},3}$,
\[
d\bigl([\Pos_3(s_G)],[\Pos_3(v_G)]\bigr)=\dbr([W^s_G],[W^v_G])=|D|+2\VC(G),\qquad |D|=3n(m-3)+6|E|.
\]
\end{theorem}
\begin{proof}
We first determine the distance $\dt$ of the chirotopes; this part holds for every graph $G$. By
Proposition~\ref{prop:lower}, Lemma~\ref{lem:sheet}(a) and Theorem~\ref{thm:wiring},
\[
|D|+2\VC(G)\le\dt\bigl(\Pos_3(s_G),\Pos_3(v_G)\bigr)\le\dB(s_G,v_G)=\dbr([W^s_G],[W^v_G])=|D|+2\VC(G).
\]
For the mutation distance, by
Observation~\ref{obs:cover} it suffices to show $\binom{m+1}3-|D|>|D|+2\VC(G)$. Using $\VC(G)\le n-1$ and $|E|=m-3n$,
\[
2|D|+2\VC(G)\le6n(m-3)+12(m-3n)+2n-2<n(6m-52)+12m .
\]
As $n\ge3$ we have $m\ge9$, so $6m-52>0$, and $n\le m/3$ gives $2|D|+2\VC(G)<2m^2-\tfrac{16}3m$. Finally
$\binom{m+1}3-2m^2+\tfrac{16}3m=\tfrac m6\,(m^2-12m+31)>0$ for $m\ge9$.
\end{proof}

The assumption $n\ge3$ cannot be dropped: for $G=K_2$ there is a walk of chirotopes of length $26<32$ from
$\Pos_3(s_G)$ to $-\Pos_3(v_G)$ (computer check, see \hyperref[sec:computations]{\emph{Computations}}); this attains the lower bound
$\binom83-|D|=56-30$ of Observation~\ref{obs:cover}, so the mutation distance in $\G_{8,3}$ is $26$.

\subsection{A polynomial diameter bound in rank three}\label{sec:diam}

Membership of $\DIST_3$ in NP with $K$ in binary needs a polynomial bound on mutation distances in rank $3$. By
Observation~\ref{obs:cover} it suffices to bound $\dt$. We join every chirotope to a realizable one
(Lemma~\ref{lem:toreal}), and join realizable chirotopes by moving their vectors (Fact~\ref{fact:bridge}).

The following normal form is standard: it is the description of $\B(n-1,2)$ by arrangements of $n$
pseudolines labelled cyclically along the line at infinity~\cite[Sect.~4]{Ziegler}, combined with the
correspondence between signotopes and marked arrangements~\cite[Thms.~2 and~7]{FW}. Geometrically, one takes the
pseudoline $e$ as the line at infinity, reorients the other pseudolines consistently along it, and labels them in
the order in which they meet it.

\begin{fact}[normal form]\label{fact:normal}
Let $\chi$ be a uniform chirotope of rank $3$ on $\Omega$, $|\Omega|=n\ge4$, and $e\in\Omega$. There is a
composition $\gamma$ of reorientations and a relabelling $\Omega\to[n-1]\cup\{\infty\}$ with $e\mapsto\infty$
such that $\gamma\chi=\Pos_3(s)$ for a signotope $s\in\B(n-1,2)$, that is, such that $\gamma\chi$ is positive at
$\infty$.
\end{fact}

\begin{lemma}[every rank-$3$ chirotope is close to a realizable one]\label{lem:toreal}
Let $\chi$ be a uniform chirotope of rank $3$ on an $n$-element set, $n\ge4$. Then there is a realizable uniform
chirotope $\chi'$ with $\dt(\chi,\chi')\le\binom{n-1}3$.
\end{lemma}
\begin{proof}
By Fact~\ref{fact:normal}, $\gamma\chi=\Pos_3(s)$ for a signotope $s\in\B(n-1,2)$ and a composition $\gamma$
of reorientations and a relabelling. By~\eqref{eq:diamB}, $s$ is joined to the constant signotope $+$ by a walk of
length at most $\binom{n-1}3$ in $\B(n-1,2)$, so $\dt(\Pos_3(s),\Pos_3(+))\le\binom{n-1}3$ by
Lemma~\ref{lem:sheet}(a). The chirotope $\Pos_3(+)$ is realized by the vectors $(1,a,a^2)$, $a\in[n-1]$, and
$(0,0,1)$ for $\infty$: for $a<b<c$,
\begin{gather*}
\det\bigl((1,a,a^2),(1,b,b^2),(1,c,c^2)\bigr)=(b-a)(c-a)(c-b)>0,\\
\det\bigl((1,a,a^2),(1,b,b^2),(0,0,1)\bigr)=b-a>0 .
\end{gather*}
Since $\gamma^{-1}$ preserves $\dt$ and realizability (Lemma~\ref{lem:minors}), $\chi'=\gamma^{-1}\Pos_3(+)$ has the
required properties.
\end{proof}

The argument does not extend to higher rank. It rests on Fact~\ref{fact:normal}, which uses that the
contraction $\chi/e$ has rank $2$: every uniform oriented matroid of rank $2$ is, up to reorientation and
relabelling, the alternating one, and this makes $\gamma\chi$ positive at $\infty$. In rank $r\ge4$ the contraction
$\chi/e$ is a uniform oriented matroid of rank $r-1$, which in general is not of this form, so $\chi$ is in general
not a positive fibre up to reorientation and relabelling. In fact, for $r\ge7$ and $n\ge r+17$ the mutation graph
$\G_{n,r}$ is disconnected~\cite{GuKnauer}, while the realizable uniform oriented matroids lie in one of its
components~\cite{RoudneffSturmfels}; so there are uniform oriented matroids that are joined to no realizable one by
mutations.

\begin{fact}[realizable chirotopes are close~\cite{RoudneffSturmfels}, {\cite[Sect.~6.3]{Handbook}}]\label{fact:bridge}
Two realizable uniform chirotopes $\chi,\psi$ of rank $r$ on an $n$-element set satisfy
$\dt(\chi,\psi)\le n\binom{n-1}{r-1}=r\binom nr$.
\end{fact}

Indeed, move the vectors of a realization of $\chi$ to those of a realization of $\psi$, one at a time, along
segments, after a small generic perturbation of both realizations. Then a moving vector crosses each hyperplane
spanned by $r-1$ of the other vectors at most once and one at a time, and each crossing is one mutation; so moving one
vector costs at most $\binom{n-1}{r-1}$ mutations.

\begin{proposition}[rank-three diameter]\label{prop:diam}
For uniform chirotopes $\chi,\psi$ of rank $3$ on an $n$-element set, $n\ge3$,
\[
\dt(\chi,\psi)\le2\binom{n-1}3+3\binom n3<5\binom n3 .
\]
\end{proposition}
\begin{proof}
For $n=3$ the claim is trivial. For $n\ge4$, Lemma~\ref{lem:toreal} gives realizable chirotopes $\chi'$ and $\psi'$
within $\binom{n-1}3$ flips of $\chi$ and $\psi$, and Fact~\ref{fact:bridge} joins $\chi'$ and $\psi'$ by a walk
of at most $3\binom n3$ flips.
\end{proof}

\begin{theorem}\label{thm:np3}
$\DIST_3$ is NP-complete, with $K$ given in binary. The same holds in corank $3$.
\end{theorem}
\begin{proof}
Given a \textup{\textsc{Vertex Cover}} instance $(G,k)$ with $n\ge3$ (smaller instances are decided directly), take
the chirotopes $\Pos_3(s_G)$ and $\Pos_3(v_G)$, where $s_G,v_G$ are the sign vectors constructed in
Section~\ref{sec:construction}; $\Pos_3$ adds $\binom m2$ entries equal to $+$. Put $K=|D|+2k$. By Theorem~\ref{thm:rank3}, $\VC(G)\le k$ holds if and only if the
mutation distance is at most $K$.

Membership: by Proposition~\ref{prop:diam} and Observation~\ref{obs:cover} the answer is yes if $K\ge5\binom{M}3$, with
$M$ the size of the ground set. Otherwise a walk of chirotopes of length at most $K$ from $\chi_1$ to $\chi_2$ or to
$-\chi_2$ is a certificate of polynomial size (Observation~\ref{obs:cover}), and every step can be checked in polynomial
time using Definition~\ref{def:chirotope}. Corank $3$ follows by duality
(Section~\ref{sec:chirotopes}).
\end{proof}

\section{Higher ranks}\label{sec:higher}

This section extends the hardness results to every rank $r\ge4$. We first outline the argument; the details
follow.
\begin{itemize}
\item \emph{The lift.} A rank-$3$ signotope $s$ on $S$ is lifted to a rank-$r$ signotope on $N\cup S$ by adding $c$
new elements below $S$ (Definition~\ref{def:lift}). The lift is built so that contracting any $r-3$ of the new
elements and deleting the others gives back $\Pos_3(s)$.
\item \emph{Lower bound.} A mutation of the lifted instance survives in at most one of these $\binom c{r-3}$ minors,
and each minor must pay the rank-$3$ excess on its own. So the excess is multiplied by at least $\binom c{r-3}$.
\item \emph{Upper bound.} A flip of $s$ lifts to at most $\binom{c+m-3}{r-3}$ flips (Lemma~\ref{lem:liftflip}), so
the excess is multiplied by at most this factor, while the Hamming distance of the lifts is known exactly
(Proposition~\ref{prop:excess}).
\item \emph{Hardness.} For large $c$ the two factors are close, so the distance of the lifts still separates
$\VC(G)\le k$ from $\VC(G)\ge k+1$. The bounds hold both for walks of chirotopes and for walks in the higher Bruhat
order $\B(N\cup S,r-1)$, which gives hardness for the mutation graphs (Theorem~\ref{thm:allranks}) and for the higher
Bruhat orders (Corollary~\ref{cor:hbo}).
\end{itemize}

Throughout this section $S$ is a finite totally ordered set with $|S|=m$, and $N$ is a set of $c$ new elements
placed below $S$. For a set $X\subseteq N\cup S$ of at least three elements, $\topp X$ denotes the set of its three largest
elements; since $N<S$, it lies in $S$ if and only if $|X\cap S|\ge3$.

\paragraph{The lift.} The lift adds the $c$ new elements of $N$ and raises the rank from $3$ to $r$.

\begin{definition}[lift]\label{def:lift}
Let $s\in\B(S,2)$, let $N$ be a finite set of new elements placed below $S$, and let $r\ge3$. The \emph{lift}
$\Lift s=\Lift^r_Ns\colon\binom{N\cup S}r\to\{\pm\}$ takes the value $s(\topp X)$ at $X$ if $|X\cap S|\ge3$, and $+$
otherwise.
\end{definition}

The lift is built from two known maps. Both are compatible with the \emph{higher Bruhat
order}~\cite{MS,Ziegler} on the rank-$k$ signotopes on a set: $u\le v$ if $v$ arises from $u$ by a walk of
signotopes in which every step changes the value of one set from $+$ to $-$.
\begin{itemize}
\item \emph{Padding.} For $r=3$ the lift is the \emph{padding} $\hat s$ of $s$, which is $+$ on every triple meeting
$N$. It is the embedding of $\B(m,2)$ as a lower interval of $\B(m+c,2)$, in mirror image~\cite[Prop.~4.3]{Ziegler}.
\item \emph{Expansion.} For a finite totally ordered set $\Omega_1$, a new element $\ell$ placed below $\Omega_1$ and
a map $v\colon\binom{\Omega_1}k\to\{\pm\}$, let $v'$ be the map on the $(k+1)$-subsets of $\{\ell\}\cup\Omega_1$ given
by $v'(Z)=v(Z\setminus\min Z)$. After reversing the order of the ground set, this is the weak extension of Felsner and
Weil, called expansion by Rambau; it maps signotopes of rank $k$ to signotopes of rank $k+1$ and is an
order-preserving embedding $\B(m,k-1)\to\B(m+1,k)$~\cite[Sect.~4.3]{FW}, \cite{Rambau}.
\end{itemize}
Reversing the order of the ground set reverses every packet and keeps the sets with sign $-$, so it maps signotopes
to signotopes and preserves the higher Bruhat order; the mirror images are therefore harmless.

The two maps combine as follows. Let $\ell$ be a further element placed below $N$. For an $(r+1)$-subset $X$ of
$\{\ell\}\cup N\cup S$, the set $X\setminus\min X$ has the same three largest elements as $X$, and $|X\cap S|\ge3$ if
and only if $|(X\setminus\min X)\cap S|\ge3$ (if $\min X\in S$, then $X\subseteq S$). Hence
$\Lift^{r+1}_{\{\ell\}\cup N}s=(\Lift^r_Ns)'$, and for $N=\{\ell_1<\dots<\ell_c\}$ with $c\ge r-3$,
\begin{equation}\label{eq:liftformula}
\Lift^r_Ns=\bigl(\cdots\bigl((\Lift^3_{N'}s)'\bigr)'\cdots\bigr)'\qquad(r-3\text{ primes}),\qquad
N'=N\setminus\{\ell_1,\dots,\ell_{r-3}\},
\end{equation}
where the steps add $\ell_{r-3},\dots,\ell_1$ in this order. So the lift is the padding followed by $r-3$ expansions.
What is new here is its use together with contractions by $(r-3)$-sets of new elements
(Proposition~\ref{prop:excess}).

\begin{lemma}\label{lem:lift}
Let $s\in\B(S,2)$, $r\ge3$ and $q=r-3$, and suppose that $|N|=c\ge q$.
\begin{enumerate}
\item[(a)] $\Lift^r_Ns$ is a signotope of rank $r$, and the lift preserves the higher Bruhat order: if $u\le v$ in
$\B(S,2)$, then $\Lift^r_Nu\le\Lift^r_Nv$.
\item[(b)] If $T\in\binom S3$ and $s^T$ is a signotope, then $\Lift s$ and $\Lift s^T$ differ exactly on the sets
$T\cup Y$ with $Y$ a $q$-subset of $N\cup\{x\in S:x<\min T\}$.
\end{enumerate}
\end{lemma}
\begin{proof}
(a) By~\eqref{eq:liftformula} the lift is the padding followed by $r-3$ expansions, and each of these maps
signotopes to signotopes and preserves the higher Bruhat order. (b) The two lifts differ at $X$ if and only if
$|X\cap S|\ge3$ and $\topp X=T$.
\end{proof}

\begin{lemma}[lifting a flip]\label{lem:liftflip}
Let $s\in\B(S,2)$, $r\ge3$ and $q=r-3$, suppose that $|N|=c\ge q$, and let $T\in\binom S3$ be such that $s^T$ is a
signotope. Then there is a walk in $\B(N\cup S,r-1)$ from $\Lift s$ to $\Lift s^T$ that flips each set $T\cup Y$, with
$Y$ a $q$-subset of $N\cup\{x\in S:x<\min T\}$, exactly once and flips nothing else. So $\dB(\Lift s,\Lift
s^T)=\binom{c+a}q$, where $a=|\{x\in S:x<\min T\}|$.
\end{lemma}
\begin{proof}
The signotopes $s$ and $s^T$ differ only at $T$, so they are comparable in the higher Bruhat order. By
Lemma~\ref{lem:lift}(a), $\Lift s$ and $\Lift s^T$ are comparable, say $u\le v$. A walk from $u$ to $v$ as
in the definition of the order changes each set at most once, from $+$ to $-$, so it flips each set of $D(u,v)$
exactly once and flips nothing else. By Lemma~\ref{lem:lift}(b), these are the $\binom{c+a}q$ sets $T\cup Y$,
and the distance follows from Lemma~\ref{lem:parity}.
\end{proof}

The lower bound comes from the fact that the lift contains many copies of the original chirotope as minors, and a
single mutation changes at most one of them.

\begin{lemma}[copies]\label{lem:copies}
Let $r\ge4$, $q=r-3$ and $|N|=c\ge q$. For a $q$-set $Y=\{y_1<\dots<y_q\}\subseteq N$ and a chirotope $\chi$ of rank
$r$ on $N\cup S\cup\{\infty\}$, let $\chi_Y$ be the rank-$3$ chirotope on $S\cup\{\infty\}$ obtained by contracting
$y_1,\dots,y_q$ in turn and deleting $N\setminus Y$.
\begin{enumerate}
\item[(a)] $(\Pos_r\Lift s)_Y=\Pos_3s$ for every $s\in\B(S,2)$.
\item[(b)] If $\chi_0,\dots,\chi_\ell$ is a walk of chirotopes, then $\bigl((\chi_i)_Y\bigr)_i$, with consecutive
repetitions removed, is a walk of chirotopes with one step for each step that flips a basis $B$ with $B\cap N=Y$.
\end{enumerate}
In particular, for $s,t\in\B(S,2)$, every walk of chirotopes from $\Pos_r\Lift s$ to $\Pos_r\Lift t$ has, for each
of the $\binom cq$ sets $Y$, at least $\dt(\Pos_3s,\Pos_3t)$ steps that flip a basis $B$ with $B\cap N=Y$.
\end{lemma}
\begin{proof}
(a) We have $\chi_Y(Z)=\chi(y_1,\dots,y_q,z_1,z_2,z_3)$ for $Z=\{z_1<z_2<z_3\}\subseteq S\cup\{\infty\}$ (contract
$y_1$, then $y_2$, and so on), and this tuple is increasing because $N<S<\infty$. For $\chi=\Pos_r\Lift s$ the value
is $+$ if $\infty\in Z$, and $s(\topp(Y\cup Z))=s(Z)$ otherwise.

(b) Apply Lemma~\ref{lem:split} $c$ times: contracting $y_j$ keeps exactly the steps that flip a basis containing
$y_j$, and deleting $x\in N\setminus Y$ keeps exactly the steps that flip a basis avoiding $x$.

The last statement follows from (a) and (b).
\end{proof}

The next proposition is the heart of this section. In it, $h$ and $H$ are the Hamming distances before and after
the lift, and $\lambda$ and $\mu$ are the least and the largest number of flips into which the lift turns one flip of
$s$: by Lemma~\ref{lem:liftflip}, a flip of $T$ becomes $\binom{c+a(T)}q$ flips, where $0\le a(T)\le m-3$. The
Hamming distance $H$ of the lifts is paid exactly, and only the excess, the number of extra flips, is multiplied, by a
factor between $\lambda$ and $\mu$. This is why a ratio $\mu/\lambda$ close to $1$ suffices in Theorem~\ref{thm:allranks}
to separate $\VC(G)\le k$ from $\VC(G)\ge k+1$.

\begin{proposition}[excess amplification]\label{prop:excess}
Let $S$ be a totally ordered set with $|S|=m$, let $N$ be a set of $c$ new elements placed below $S$, let $r\ge4$
and $q=r-3$ with $c\ge q$, and let $s,t\in\B(S,2)$. Put
\[
h=|D(s,t)|,\qquad H=|D(\Lift s,\Lift t)|,\qquad \lambda=\binom cq,\qquad \mu=\binom{c+m-3}q .
\]
Then, with $\Pos_3$ on $S\cup\{\infty\}$ and $\Pos_r$ on $N\cup S\cup\{\infty\}$,
\[
\lambda\bigl(\dt(\Pos_3s,\Pos_3t)-h\bigr)\;\le\;\dt(\Pos_r\Lift s,\Pos_r\Lift t)-H\;\le\;\dB(\Lift s,\Lift t)-H\;\le\;
\mu\bigl(\dB(s,t)-h\bigr).
\]
\end{proposition}
\begin{proof}
The middle inequality is Lemma~\ref{lem:sheet}(a), and $\Lift s,\Lift t$ are signotopes by
Lemma~\ref{lem:lift}(a).

\emph{Upper bound.} Take a walk of length $\dB(s,t)$ from $s$ to $t$. For a triple $T\in\binom S3$ let $n_T$ be
the number of its steps flipping $T$, and let $\mu_T=\binom{c+a(T)}q$ with $a(T)=|\{x\in S:x<\min T\}|$.
Replacing each step flipping $T$ by the walk of Lemma~\ref{lem:liftflip}, of length $\mu_T$, gives a walk from
$\Lift s$ to $\Lift t$ of length $\sum_{T\in\binom S3}\mu_Tn_T$. Since $(\Lift s)(X)=\hat s(\topp X)$, the lifts
$\Lift s$ and $\Lift t$ differ exactly on the sets $T\cup Y$ with $T\in D(s,t)$ and $Y$ a $q$-subset of
$N\cup\{x\in S:x<\min T\}$ (as in Lemma~\ref{lem:lift}(b)), and different triples $T$ give different sets, since
$T=\topp(T\cup Y)$. Hence $H=\sum_{T\in D(s,t)}\mu_T$, and
\[
\dB(\Lift s,\Lift t)\;\le\;\sum_{T\in\binom S3}\mu_Tn_T\;=\;H+\sum_{T\in\binom S3}\mu_T\bigl(n_T-[T\in D(s,t)]\bigr)
\;\le\;H+\mu\bigl(\dB(s,t)-h\bigr),
\]
since every $n_T-[T\in D(s,t)]$ is nonnegative (Lemma~\ref{lem:parity}), these numbers sum to $\dB(s,t)-h$, and
$\mu_T\le\mu$.

\emph{Lower bound.} Take a walk of chirotopes from $\Pos_r\Lift s$ to $\Pos_r\Lift t$, and for an arbitrary basis
$B\in\binom{N\cup S\cup\{\infty\}}r$ denote by $n_B$ the number of its steps flipping $B$. The walk has length
$\sum_Bn_B$, and we split this sum according to whether $|B\cap N|=q$.

By Lemma~\ref{lem:copies}, $\sum_{B\cap N=Y}n_B\ge\dt(\Pos_3s,\Pos_3t)$ for each of the $\lambda$ sets
$Y\in\binom Nq$. Every basis $B$ with $|B\cap N|=q$ appears in exactly one of these sums, namely the one for
$Y=B\cap N$, so
\[
\sum_{B:\,|B\cap N|=q}n_B\;\ge\;\lambda\,\dt(\Pos_3s,\Pos_3t).
\]
For the other bases, $n_B\ge1$ whenever $B\in D(\Lift s,\Lift t)$ (Lemma~\ref{lem:parity}). By the description of
$D(\Lift s,\Lift t)$ in the upper bound, its elements with $|B\cap N|=q$ are exactly the $\lambda h$ sets $T\cup Y$
with $T\in D(s,t)$ and $Y\in\binom Nq$, so $H-\lambda h$ of its elements have $|B\cap N|\ne q$, and
\[
\sum_{B:\,|B\cap N|\ne q}n_B\;\ge\;H-\lambda h.
\]
Adding the two inequalities, the walk has length at least $\lambda\,\dt(\Pos_3s,\Pos_3t)+H-\lambda h$.
\end{proof}

\begin{theorem}\label{thm:allranks}
For every fixed $r\ge4$, the problem $\DIST_r$ of deciding whether two uniform oriented matroids of rank $r$ have
distance at most $K$ in the mutation graph is NP-hard, and it is NP-complete when $K$ is given in unary. The same
holds in every fixed corank $r\ge4$.
\end{theorem}
\begin{proof}
Let $(G,k)$ be an instance of \textup{\textsc{Vertex Cover}} with $1\le k\le|V(G)|$ and $|V(G)|\ge2$ (other
instances are decided directly), let $m=3|V(G)|+|E(G)|$, $s=s_G$, $t=v_G$, $q=r-3$, and
\[
c=qm(k+1)+q,\qquad K=H+2\mu k,
\]
with $h,H,\lambda,\mu$ as in Proposition~\ref{prop:excess}, $S=[m]$ and $|N|=c$. By the proof of
Theorem~\ref{thm:rank3},
$\dB(s,t)-h=\dt(\Pos_3s,\Pos_3t)-h=2\VC(G)$, so Proposition~\ref{prop:excess} gives
\[
H+2\lambda\,\VC(G)\;\le\;\dt(\Pos_r\Lift s,\Pos_r\Lift t)\;\le\;\dB(\Lift s,\Lift t)\;\le\;H+2\mu\,\VC(G).
\]
If $\VC(G)\le k$, the distance is at most $K$. If $\VC(G)\ge k+1$, it is at least $H+2\lambda(k+1)>K$, because
\[
\frac\mu\lambda=\prod_{j=0}^{q-1}\Bigl(1+\frac{m-3}{c-j}\Bigr)\le\Bigl(1+\frac{m-3}{c-q+1}\Bigr)^{q}
\le\frac1{1-q\,\frac{m-3}{c-q+1}}<\frac1{1-\frac1{k+1}}=\frac{k+1}k
\]
by Bernoulli's inequality $(1+x)^{-q}\ge1-qx$ and $c-q+1>qm(k+1)$. So $\dt(\Pos_r\Lift s,\Pos_r\Lift t)\le K$ if
and only if $\VC(G)\le k$.

As in rank $3$, the same holds for the mutation distance of $[\Pos_r\Lift s]$ and $[\Pos_r\Lift t]$, because
$-\Pos_r\Lift t$ is far from $\Pos_r\Lift s$: it differs from $\Pos_r\Lift s$ on all but $H$ of the
$\binom{c+m+1}r$ bases, and since $c$ is large compared with $m$ and $k$, this number exceeds $K$. Indeed, let
$M=c+m\ge(r-1)m$. Then $\binom Mr=\mu\binom M3/\binom r3$ and $\binom M3\ge(r-1)^3\binom m3$, and
$(r-1)^3/\binom r3=6(r-1)^2/(r(r-2))>4$. Hence $\binom Mr>4\mu\binom m3\ge2H+2\mu k$, as $H\le\mu\binom m3$ and
$k\le\binom m3$. So Observation~\ref{obs:cover} gives the claim.

\emph{Size.} The ground set has $c+m+1=O(rm^2)$ elements, and an entry of $\Pos_r\Lift s$ is $+$ if the basis
contains $\infty$ or one of its three largest elements lies in $N$, and otherwise the value of $s_G$ (or $v_G$)
at its three largest elements. So the $\binom{c+m+1}r$ entries are computed in polynomial time for fixed $r$,
and $K$ is polynomial in $m$; for $r=4$ the chirotopes have $O(m^2)$ elements and are given by $O(m^8)$ signs.
If $K$ is given in unary, the problem is in NP: as in the proof of Theorem~\ref{thm:np3}, a walk of chirotopes of
length at most $K$ from $\chi_1$ to $\chi_2$ or to $-\chi_2$ is a certificate of polynomial size, and every step can
be checked in polynomial time. The statement for corank $r$ follows by duality (Section~\ref{sec:chirotopes}).
\end{proof}

\begin{corollary}[higher Bruhat orders]\label{cor:hbo}
For every fixed $r\ge3$, deciding flip distance in the higher Bruhat order $\B(n,r-1)$, that is, the problem
$\FLIP_{\B(n,r-1)}$ of Section~\ref{sec:signotopes}, is NP-complete, with $K$ given in binary.
\end{corollary}
\begin{proof}
For $r=3$ this is Theorem~\ref{thm:braidNP}. Let $r\ge4$, and take the instances of the proof of
Theorem~\ref{thm:allranks}, with the same side conditions on $(G,k)$: the signotopes $\Lift s_G$ and $\Lift v_G$ of
rank $r$ on $N\cup[m]$, relabelled by $[c+m]$ in increasing order, and $K=H+2\mu k$. They are computed in polynomial
time. Both bounds of Proposition~\ref{prop:excess} apply to $\dB(\Lift s_G,\Lift v_G)$: the upper bound is proved by
a walk of signotopes, and the lower bound holds for all walks of chirotopes, which include the walks of signotopes
(Lemma~\ref{lem:sheet}(a)). Since the rank-$3$ excess is $2\VC(G)$ for both distances (Theorem~\ref{thm:wiring} and the proof of
Theorem~\ref{thm:rank3}), this gives
\[
H+2\lambda\,\VC(G)\;\le\;\dB(\Lift s_G,\Lift v_G)\;\le\;H+2\mu\,\VC(G),
\]
and the estimate of $\mu/\lambda$ in the proof of Theorem~\ref{thm:allranks} shows that
$\dB(\Lift s_G,\Lift v_G)\le K$ if and only if $\VC(G)\le k$. Membership: by~\eqref{eq:diamB} the answer is yes if $K\ge\binom nr$, where $n$ is the size of the
ground set. Otherwise a certificate is a sequence of at most $K$ signotopes, each obtained from the previous one by
negating one $r$-set; that each term is a signotope is checked on its $\binom n{r+1}$ packets, in polynomial time
for fixed $r$.
\end{proof}

\phantomsection\addcontentsline{toc}{part}{Concluding remarks}
\part*{Concluding remarks}

\section{Remarks and open problems}\label{sec:remarks}

\begin{enumerate}
\item \emph{Pairs without excess.} Can one decide in polynomial time
whether a pair of wiring diagrams has excess $0$, i.e.\ whether $\dbr=|D|$? By the packet criterion
(Lemma~\ref{lem:packet}) this is an
ordering problem: choose one of the two arc orders for every $4$-set $P$ with $|D\cap\binom P3|=4$ so that the union
of all arc orders (unique when $|D\cap\binom P3|\in\{2,3\}$) is acyclic. It has the flavour of
\textup{\textsc{Betweenness}}, which is NP-complete~\cite{Opatrny}, but the instances are very structured.
\item\label{it:isometry} \emph{Isometry of the positive fibre.} Is $\dt(\Pos_3s,\Pos_3v)=\dB(s,v)$ for all
$s,v\in\B(m,2)$? In other words, can a walk of chirotopes between two chirotopes positive at $\infty$ be shortened by
leaving the chirotopes positive at $\infty$? We prove the equality for our instances (proof of
Theorem~\ref{thm:rank3}).
\item \emph{Binary $K$ in rank $\ge4$.} Is $\DIST_r$ in NP for $r\ge4$ when $K$ is given in binary? For the higher
Bruhat orders this holds by the diameter bound~\eqref{eq:diamB} (Corollary~\ref{cor:hbo}); for mutation graphs it would
follow from a polynomial bound on the diameters of the components of $\G_{n,r}$, as in
Proposition~\ref{prop:diam} for rank $3$.
\item \emph{Reachability.} For $r\ge7$ and $n\ge r+17$ the mutation graphs $\G_{n,r}$ are disconnected~\cite{GuKnauer}. What is the
complexity of deciding whether two uniform oriented matroids are connected by mutations?
\end{enumerate}

\section*{Computations}\phantomsection\label{sec:computations}\addcontentsline{toc}{section}{Computations}

No proof in this paper depends on a computer. The finite statements used in the proofs were additionally checked
by computer, among them Proposition~\ref{prop:gadget} (also by breadth-first search in $\B(7,2)$), the construction
and the sequence of Proposition~\ref{prop:upper} for $27$ graphs, and the lemmas of Section~\ref{sec:higher} in small
cases. A few side remarks, which are not used in any proof, were found by computer and are not proved by hand: the
facts about six wires and about the $60$ cells at the end of Section~\ref{sec:gadget}, and the walk of length $26$ for $G=K_2$ after
Theorem~\ref{thm:rank3}. The programs are \texttt{check.py} (sign vectors, signotopes and chirotopes) and \texttt{check\_words.py} (the
same at the level of words). They are written from the
definitions of this paper and need only the Python standard library.

In addition, the main results are formally verified in Lean~4: the distance formulas of Theorems~\ref{thm:wiring}
and~\ref{thm:rank3} (Theorem~A and the first statement of Theorem~C), the excess amplification of
Proposition~\ref{prop:excess}, and the equivalences behind Theorem~\ref{thm:allranks} and Corollary~\ref{cor:hbo},
that is, the distance is at most $K$ if and only if $\VC(G)\le k$, without the complexity statements themselves.
The formalization uses only core Lean and contains no unproved steps.

The programs, with their outputs and a description of each check, and the Lean formalization are available at
\url{https://github.com/tilenmarc/braid-bruhat-mutation-distance}.

\section*{Use of AI tools}\phantomsection\addcontentsline{toc}{section}{Use of AI tools}

The author used two large language model agents, Claude Fable~5.1 and Claude Opus~5.5 (Anthropic), in
the development of this work. They assisted with drafting and polishing proofs and played an important role in the
discovery of the gadget in Section~\ref{sec:gadget}, through computer search and analysis of the resulting
candidates. They were also used to explore and evaluate several approaches for extending the construction to oriented
matroids of higher rank; this investigation led to the lift developed in Section~\ref{sec:higher}. In addition,
ChatGPT~6.0 (OpenAI) was used to review drafts of the manuscript. The author
retained full responsibility for the mathematical development, selection of arguments, and final presentation of the
work, as well as for all the results and claims in the paper.

\section*{Acknowledgements}\phantomsection\addcontentsline{toc}{section}{Acknowledgements}

The author was supported by the Slovenian Research and Innovation Agency (ARIS), research program P1-0297.

\phantomsection

\end{document}

%% file: fig_w_a.tex
\begin{tikzpicture}[x=0.42cm,y=0.42cm]
\fill[black!7] (0,0.55) rectangle (3,4+0.45); \node[font=\scriptsize] at (1.5,4+0.9) {$\sigma_1\sigma_2\sigma_1$};
\draw[blue!70!black,thick] (0,4) -- (1,3) -- (2,2) -- (3,2) -- (4,1) -- (5,1) -- (6,1);
\node[left,font=\scriptsize] at (0,4) {1};
\node[right,font=\scriptsize] at (6,1) {1};
\draw[red!70!black,thick] (0,3) -- (1,4) -- (2,4) -- (3,3) -- (4,3) -- (5,2) -- (6,2);
\node[left,font=\scriptsize] at (0,3) {2};
\node[right,font=\scriptsize] at (6,2) {2};
\draw[green!50!black,thick] (0,2) -- (1,2) -- (2,3) -- (3,4) -- (4,4) -- (5,4) -- (6,3);
\node[left,font=\scriptsize] at (0,2) {3};
\node[right,font=\scriptsize] at (6,3) {3};
\draw[black,thick] (0,1) -- (1,1) -- (2,1) -- (3,1) -- (4,2) -- (5,3) -- (6,4);
\node[left,font=\scriptsize] at (0,1) {4};
\node[right,font=\scriptsize] at (6,4) {4};
\end{tikzpicture}

%% file: fig_w_b.tex
\begin{tikzpicture}[x=0.42cm,y=0.42cm]
\fill[black!7] (0,0.55) rectangle (3,4+0.45); \node[font=\scriptsize] at (1.5,4+0.9) {$\sigma_2\sigma_1\sigma_2$};
\draw[blue!70!black,thick] (0,4) -- (1,4) -- (2,3) -- (3,2) -- (4,1) -- (5,1) -- (6,1);
\node[left,font=\scriptsize] at (0,4) {1};
\node[right,font=\scriptsize] at (6,1) {1};
\draw[red!70!black,thick] (0,3) -- (1,2) -- (2,2) -- (3,3) -- (4,3) -- (5,2) -- (6,2);
\node[left,font=\scriptsize] at (0,3) {2};
\node[right,font=\scriptsize] at (6,2) {2};
\draw[green!50!black,thick] (0,2) -- (1,3) -- (2,4) -- (3,4) -- (4,4) -- (5,4) -- (6,3);
\node[left,font=\scriptsize] at (0,2) {3};
\node[right,font=\scriptsize] at (6,3) {3};
\draw[black,thick] (0,1) -- (1,1) -- (2,1) -- (3,1) -- (4,2) -- (5,3) -- (6,4);
\node[left,font=\scriptsize] at (0,1) {4};
\node[right,font=\scriptsize] at (6,4) {4};
\end{tikzpicture}

%% file: fig_w_c.tex
\begin{tikzpicture}[x=0.42cm,y=0.42cm]
\draw[blue!70!black,thick] (0,3) -- (1,2) -- (2,1) -- (3,1);
\node[left,font=\scriptsize] at (0,3) {1};
\node[right,font=\scriptsize] at (3,1) {1};
\draw[red!70!black,thick] (0,2) -- (1,3) -- (2,3) -- (3,2);
\node[left,font=\scriptsize] at (0,2) {2};
\node[right,font=\scriptsize] at (3,2) {2};
\draw[green!50!black,thick] (0,1) -- (1,1) -- (2,2) -- (3,3);
\node[left,font=\scriptsize] at (0,1) {3};
\node[right,font=\scriptsize] at (3,3) {3};
\end{tikzpicture}

%% file: fig_c8.tex
\begin{tikzpicture}
\begin{scope}[shift={(0.0cm,0.0cm)},x=0.3cm,y=0.3cm]
\draw[blue!70!black,thick] (0,4) -- (1,3) -- (2,2) -- (3,2) -- (4,1) -- (5,1) -- (6,1);
\node[left,font=\scriptsize] at (0,4) {1};
\node[right,font=\scriptsize] at (6,1) {1};
\draw[red!70!black,thick] (0,3) -- (1,4) -- (2,4) -- (3,3) -- (4,3) -- (5,2) -- (6,2);
\node[left,font=\scriptsize] at (0,3) {2};
\node[right,font=\scriptsize] at (6,2) {2};
\draw[green!50!black,thick] (0,2) -- (1,2) -- (2,3) -- (3,4) -- (4,4) -- (5,4) -- (6,3);
\node[left,font=\scriptsize] at (0,2) {3};
\node[right,font=\scriptsize] at (6,3) {3};
\draw[black,thick] (0,1) -- (1,1) -- (2,1) -- (3,1) -- (4,2) -- (5,3) -- (6,4);
\node[left,font=\scriptsize] at (0,1) {4};
\node[right,font=\scriptsize] at (6,4) {4};
\node[font=\scriptsize] at (3,-0.3) {$v_0=\pl{++++}$};
\end{scope}
\begin{scope}[shift={(3.7cm,0.0cm)},x=0.3cm,y=0.3cm]
\draw[blue!70!black,thick] (0,4) -- (1,3) -- (2,2) -- (3,1) -- (4,1) -- (5,1) -- (6,1);
\node[left,font=\scriptsize] at (0,4) {1};
\node[right,font=\scriptsize] at (6,1) {1};
\draw[red!70!black,thick] (0,3) -- (1,4) -- (2,4) -- (3,4) -- (4,4) -- (5,3) -- (6,2);
\node[left,font=\scriptsize] at (0,3) {2};
\node[right,font=\scriptsize] at (6,2) {2};
\draw[green!50!black,thick] (0,2) -- (1,2) -- (2,3) -- (3,3) -- (4,2) -- (5,2) -- (6,3);
\node[left,font=\scriptsize] at (0,2) {3};
\node[right,font=\scriptsize] at (6,3) {3};
\draw[black,thick] (0,1) -- (1,1) -- (2,1) -- (3,2) -- (4,3) -- (5,4) -- (6,4);
\node[left,font=\scriptsize] at (0,1) {4};
\node[right,font=\scriptsize] at (6,4) {4};
\node[font=\scriptsize] at (3,-0.3) {$v_1=\pl{+++-{}}$};
\end{scope}
\begin{scope}[shift={(7.4cm,0.0cm)},x=0.3cm,y=0.3cm]
\draw[blue!70!black,thick] (0,4) -- (1,3) -- (2,3) -- (3,2) -- (4,2) -- (5,1) -- (6,1);
\node[left,font=\scriptsize] at (0,4) {1};
\node[right,font=\scriptsize] at (6,1) {1};
\draw[red!70!black,thick] (0,3) -- (1,4) -- (2,4) -- (3,4) -- (4,3) -- (5,3) -- (6,2);
\node[left,font=\scriptsize] at (0,3) {2};
\node[right,font=\scriptsize] at (6,2) {2};
\draw[green!50!black,thick] (0,2) -- (1,2) -- (2,1) -- (3,1) -- (4,1) -- (5,2) -- (6,3);
\node[left,font=\scriptsize] at (0,2) {3};
\node[right,font=\scriptsize] at (6,3) {3};
\draw[black,thick] (0,1) -- (1,1) -- (2,2) -- (3,3) -- (4,4) -- (5,4) -- (6,4);
\node[left,font=\scriptsize] at (0,1) {4};
\node[right,font=\scriptsize] at (6,4) {4};
\node[font=\scriptsize] at (3,-0.3) {$v_2=\pl{++-{}-{}}$};
\end{scope}
\begin{scope}[shift={(7.4cm,-2.7cm)},x=0.3cm,y=0.3cm]
\draw[blue!70!black,thick] (0,4) -- (1,4) -- (2,4) -- (3,3) -- (4,2) -- (5,1) -- (6,1);
\node[left,font=\scriptsize] at (0,4) {1};
\node[right,font=\scriptsize] at (6,1) {1};
\draw[red!70!black,thick] (0,3) -- (1,3) -- (2,2) -- (3,2) -- (4,3) -- (5,3) -- (6,2);
\node[left,font=\scriptsize] at (0,3) {2};
\node[right,font=\scriptsize] at (6,2) {2};
\draw[green!50!black,thick] (0,2) -- (1,1) -- (2,1) -- (3,1) -- (4,1) -- (5,2) -- (6,3);
\node[left,font=\scriptsize] at (0,2) {3};
\node[right,font=\scriptsize] at (6,3) {3};
\draw[black,thick] (0,1) -- (1,2) -- (2,3) -- (3,4) -- (4,4) -- (5,4) -- (6,4);
\node[left,font=\scriptsize] at (0,1) {4};
\node[right,font=\scriptsize] at (6,4) {4};
\node[font=\scriptsize] at (3,-0.3) {$v_3=\pl{+-{}-{}-{}}$};
\end{scope}
\begin{scope}[shift={(7.4cm,-5.4cm)},x=0.3cm,y=0.3cm]
\draw[blue!70!black,thick] (0,4) -- (1,4) -- (2,4) -- (3,3) -- (4,3) -- (5,2) -- (6,1);
\node[left,font=\scriptsize] at (0,4) {1};
\node[right,font=\scriptsize] at (6,1) {1};
\draw[red!70!black,thick] (0,3) -- (1,3) -- (2,2) -- (3,2) -- (4,1) -- (5,1) -- (6,2);
\node[left,font=\scriptsize] at (0,3) {2};
\node[right,font=\scriptsize] at (6,2) {2};
\draw[green!50!black,thick] (0,2) -- (1,1) -- (2,1) -- (3,1) -- (4,2) -- (5,3) -- (6,3);
\node[left,font=\scriptsize] at (0,2) {3};
\node[right,font=\scriptsize] at (6,3) {3};
\draw[black,thick] (0,1) -- (1,2) -- (2,3) -- (3,4) -- (4,4) -- (5,4) -- (6,4);
\node[left,font=\scriptsize] at (0,1) {4};
\node[right,font=\scriptsize] at (6,4) {4};
\node[font=\scriptsize] at (3,-0.3) {$v_4=\pl{-{}-{}-{}-{}}$};
\end{scope}
\begin{scope}[shift={(3.7cm,-5.4cm)},x=0.3cm,y=0.3cm]
\draw[blue!70!black,thick] (0,4) -- (1,4) -- (2,4) -- (3,4) -- (4,3) -- (5,2) -- (6,1);
\node[left,font=\scriptsize] at (0,4) {1};
\node[right,font=\scriptsize] at (6,1) {1};
\draw[red!70!black,thick] (0,3) -- (1,2) -- (2,1) -- (3,1) -- (4,1) -- (5,1) -- (6,2);
\node[left,font=\scriptsize] at (0,3) {2};
\node[right,font=\scriptsize] at (6,2) {2};
\draw[green!50!black,thick] (0,2) -- (1,3) -- (2,3) -- (3,2) -- (4,2) -- (5,3) -- (6,3);
\node[left,font=\scriptsize] at (0,2) {3};
\node[right,font=\scriptsize] at (6,3) {3};
\draw[black,thick] (0,1) -- (1,1) -- (2,2) -- (3,3) -- (4,4) -- (5,4) -- (6,4);
\node[left,font=\scriptsize] at (0,1) {4};
\node[right,font=\scriptsize] at (6,4) {4};
\node[font=\scriptsize] at (3,-0.3) {$v_5=\pl{-{}-{}-{}+}$};
\end{scope}
\begin{scope}[shift={(0.0cm,-5.4cm)},x=0.3cm,y=0.3cm]
\draw[blue!70!black,thick] (0,4) -- (1,4) -- (2,3) -- (3,3) -- (4,2) -- (5,2) -- (6,1);
\node[left,font=\scriptsize] at (0,4) {1};
\node[right,font=\scriptsize] at (6,1) {1};
\draw[red!70!black,thick] (0,3) -- (1,2) -- (2,2) -- (3,1) -- (4,1) -- (5,1) -- (6,2);
\node[left,font=\scriptsize] at (0,3) {2};
\node[right,font=\scriptsize] at (6,2) {2};
\draw[green!50!black,thick] (0,2) -- (1,3) -- (2,4) -- (3,4) -- (4,4) -- (5,3) -- (6,3);
\node[left,font=\scriptsize] at (0,2) {3};
\node[right,font=\scriptsize] at (6,3) {3};
\draw[black,thick] (0,1) -- (1,1) -- (2,1) -- (3,2) -- (4,3) -- (5,4) -- (6,4);
\node[left,font=\scriptsize] at (0,1) {4};
\node[right,font=\scriptsize] at (6,4) {4};
\node[font=\scriptsize] at (3,-0.3) {$v_6=\pl{-{}-{}++}$};
\end{scope}
\begin{scope}[shift={(0.0cm,-2.7cm)},x=0.3cm,y=0.3cm]
\draw[blue!70!black,thick] (0,4) -- (1,4) -- (2,3) -- (3,2) -- (4,1) -- (5,1) -- (6,1);
\node[left,font=\scriptsize] at (0,4) {1};
\node[right,font=\scriptsize] at (6,1) {1};
\draw[red!70!black,thick] (0,3) -- (1,2) -- (2,2) -- (3,3) -- (4,3) -- (5,2) -- (6,2);
\node[left,font=\scriptsize] at (0,3) {2};
\node[right,font=\scriptsize] at (6,2) {2};
\draw[green!50!black,thick] (0,2) -- (1,3) -- (2,4) -- (3,4) -- (4,4) -- (5,4) -- (6,3);
\node[left,font=\scriptsize] at (0,2) {3};
\node[right,font=\scriptsize] at (6,3) {3};
\draw[black,thick] (0,1) -- (1,1) -- (2,1) -- (3,1) -- (4,2) -- (5,3) -- (6,4);
\node[left,font=\scriptsize] at (0,1) {4};
\node[right,font=\scriptsize] at (6,4) {4};
\node[font=\scriptsize] at (3,-0.3) {$v_7=\pl{-{}+++}$};
\end{scope}
\draw[black!45,thick] (2.25cm,0.75cm) -- (3.25cm,0.75cm) node[midway,above,font=\scriptsize,text=black] {$234$};
\draw[black!45,thick] (5.95cm,0.75cm) -- (6.95cm,0.75cm) node[midway,above,font=\scriptsize,text=black] {$134$};
\draw[black!45,thick] (8.3cm,-0.45cm) -- (8.3cm,-1.2500000000000002cm) node[midway,right,font=\scriptsize,text=black] {$124$};
\draw[black!45,thick] (8.3cm,-3.1500000000000004cm) -- (8.3cm,-3.95cm) node[midway,right,font=\scriptsize,text=black] {$123$};
\draw[black!45,thick] (5.95cm,-4.65cm) -- (6.95cm,-4.65cm) node[midway,above,font=\scriptsize,text=black] {$234$};
\draw[black!45,thick] (2.25cm,-4.65cm) -- (3.25cm,-4.65cm) node[midway,above,font=\scriptsize,text=black] {$134$};
\draw[black!45,thick] (0.8999999999999999cm,-3.1500000000000004cm) -- (0.8999999999999999cm,-3.95cm) node[midway,right,font=\scriptsize,text=black] {$124$};
\draw[black!45,thick] (0.8999999999999999cm,-0.45cm) -- (0.8999999999999999cm,-1.2500000000000002cm) node[midway,right,font=\scriptsize,text=black] {$123$};
\node at (4.6cm,-1.9500000000000002cm) {$\B(4,2)$};
\end{tikzpicture}

%% file: fig_kappa.tex
\begin{tikzpicture}[x=0.62cm,y=0.5cm]
\fill[black!6] (4,-0.45) rectangle (6,6+0.45);
\fill[black!6] (10,-0.45) rectangle (14,6+0.45);
\draw[blue!70!black,thick] (0,6) -- (1,6) -- (2,6) -- (3,6) -- (4,5) -- (5,5) -- (6,5) -- (7,5) -- (8,5) -- (9,5) -- (10,4) -- (11,4) -- (12,4) -- (13,4) -- (14,3) -- (15,2);
\node[left,font=\scriptsize] at (0,6) {$u_1$};
\node[right,font=\scriptsize] at (15,2) {$u_1$};
\draw[blue!70!black,thick] (0,5) -- (1,5) -- (2,5) -- (3,4) -- (4,4) -- (5,4) -- (6,3) -- (7,3) -- (8,2) -- (9,2) -- (10,2) -- (11,2) -- (12,1) -- (13,1) -- (14,1) -- (15,1);
\node[left,font=\scriptsize] at (0,5) {$u_2$};
\node[right,font=\scriptsize] at (15,1) {$u_2$};
\draw[blue!70!black,thick] (0,4) -- (1,4) -- (2,3) -- (3,3) -- (4,3) -- (5,2) -- (6,2) -- (7,1) -- (8,1) -- (9,1) -- (10,1) -- (11,0) -- (12,0) -- (13,0) -- (14,0) -- (15,0);
\node[left,font=\scriptsize] at (0,4) {$u_3$};
\node[right,font=\scriptsize] at (15,0) {$u_3$};
\draw[black,thick] (0,3) -- (1,2) -- (2,2) -- (3,2) -- (4,2) -- (5,3) -- (6,4) -- (7,4) -- (8,4) -- (9,3) -- (10,3) -- (11,3) -- (12,3) -- (13,2) -- (14,2) -- (15,3);
\node[left,font=\scriptsize] at (0,3) {$p$};
\node[right,font=\scriptsize] at (15,3) {$p$};
\draw[red!70!black,thick] (0,2) -- (1,3) -- (2,4) -- (3,5) -- (4,6) -- (5,6) -- (6,6) -- (7,6) -- (8,6) -- (9,6) -- (10,6) -- (11,6) -- (12,6) -- (13,6) -- (14,6) -- (15,6);
\node[left,font=\scriptsize] at (0,2) {$w_1$};
\node[right,font=\scriptsize] at (15,6) {$w_1$};
\draw[red!70!black,thick] (0,1) -- (1,1) -- (2,1) -- (3,1) -- (4,1) -- (5,1) -- (6,1) -- (7,2) -- (8,3) -- (9,4) -- (10,5) -- (11,5) -- (12,5) -- (13,5) -- (14,5) -- (15,5);
\node[left,font=\scriptsize] at (0,1) {$w_2$};
\node[right,font=\scriptsize] at (15,5) {$w_2$};
\draw[red!70!black,thick] (0,0) -- (1,0) -- (2,0) -- (3,0) -- (4,0) -- (5,0) -- (6,0) -- (7,0) -- (8,0) -- (9,0) -- (10,0) -- (11,1) -- (12,2) -- (13,3) -- (14,4) -- (15,4);
\node[left,font=\scriptsize] at (0,0) {$w_3$};
\node[right,font=\scriptsize] at (15,4) {$w_3$};
\node[circle,draw,fill=white,inner sep=0.6pt,font=\tiny] at (0.5,2.5) {1};
\node[circle,draw,fill=white,inner sep=0.6pt,font=\tiny] at (4.5,2.5) {2};
\node[circle,draw,fill=white,inner sep=0.6pt,font=\tiny] at (5.5,3.5) {3};
\node[circle,draw,fill=white,inner sep=0.6pt,font=\tiny] at (8.5,3.5) {4};
\node[circle,draw,fill=white,inner sep=0.6pt,font=\tiny] at (12.5,2.5) {5};
\node[circle,draw,fill=white,inner sep=0.6pt,font=\tiny] at (14.5,2.5) {6};
\node[font=\scriptsize] at (2.0,6+1.0) {$\sigma_4\sigma_3\sigma_2\sigma_1$};
\node[font=\scriptsize] at (5.0,6+1.0) {$\sigma_4\sigma_3$};
\node[font=\scriptsize] at (8.0,6+1.0) {$\sigma_5\sigma_4\sigma_3\sigma_2$};
\node[font=\scriptsize] at (12.0,6+1.0) {$\sigma_6\sigma_5\sigma_4\sigma_3$};
\node[font=\scriptsize] at (14.5,6+1.0) {$\sigma_4$};
\end{tikzpicture}

%% file: fig_gs.tex
\begin{tikzpicture}[x=0.42cm,y=0.42cm]
\fill[black!7] (0,0.55) rectangle (3,7+0.45); \node[font=\scriptsize] at (1.5,7+0.9) {$\beta_1$};
\fill[black!7] (3,0.55) rectangle (6,7+0.45); \node[font=\scriptsize] at (4.5,7+0.9) {$\beta_5$};
\draw[blue!70!black,thick] (0,7) -- (1,6) -- (2,5) -- (3,5) -- (4,5) -- (5,5) -- (6,5) -- (7,5) -- (8,4) -- (9,4) -- (10,4) -- (11,3) -- (12,3) -- (13,2) -- (14,2) -- (15,2) -- (16,2) -- (17,1) -- (18,1) -- (19,1) -- (20,1) -- (21,1);
\node[left,font=\scriptsize] at (0,7) {1};
\node[right,font=\scriptsize] at (21,1) {1};
\draw[blue!70!black,thick] (0,6) -- (1,7) -- (2,7) -- (3,6) -- (4,6) -- (5,6) -- (6,6) -- (7,6) -- (8,6) -- (9,5) -- (10,5) -- (11,5) -- (12,4) -- (13,4) -- (14,3) -- (15,3) -- (16,3) -- (17,3) -- (18,2) -- (19,2) -- (20,2) -- (21,2);
\node[left,font=\scriptsize] at (0,6) {2};
\node[right,font=\scriptsize] at (21,2) {2};
\draw[blue!70!black,thick] (0,5) -- (1,5) -- (2,6) -- (3,7) -- (4,7) -- (5,7) -- (6,7) -- (7,7) -- (8,7) -- (9,7) -- (10,6) -- (11,6) -- (12,6) -- (13,6) -- (14,6) -- (15,6) -- (16,5) -- (17,5) -- (18,5) -- (19,5) -- (20,4) -- (21,3);
\node[left,font=\scriptsize] at (0,5) {3};
\node[right,font=\scriptsize] at (21,3) {3};
\draw[black,thick] (0,4) -- (1,4) -- (2,4) -- (3,4) -- (4,4) -- (5,4) -- (6,4) -- (7,3) -- (8,3) -- (9,3) -- (10,3) -- (11,4) -- (12,5) -- (13,5) -- (14,5) -- (15,4) -- (16,4) -- (17,4) -- (18,4) -- (19,3) -- (20,3) -- (21,4);
\node[left,font=\scriptsize] at (0,4) {4};
\node[right,font=\scriptsize] at (21,4) {4};
\draw[red!70!black,thick] (0,3) -- (1,3) -- (2,3) -- (3,3) -- (4,2) -- (5,1) -- (6,1) -- (7,1) -- (8,1) -- (9,1) -- (10,1) -- (11,1) -- (12,1) -- (13,1) -- (14,1) -- (15,1) -- (16,1) -- (17,2) -- (18,3) -- (19,4) -- (20,5) -- (21,5);
\node[left,font=\scriptsize] at (0,3) {5};
\node[right,font=\scriptsize] at (21,5) {5};
\draw[red!70!black,thick] (0,2) -- (1,2) -- (2,2) -- (3,2) -- (4,3) -- (5,3) -- (6,2) -- (7,2) -- (8,2) -- (9,2) -- (10,2) -- (11,2) -- (12,2) -- (13,3) -- (14,4) -- (15,5) -- (16,6) -- (17,6) -- (18,6) -- (19,6) -- (20,6) -- (21,6);
\node[left,font=\scriptsize] at (0,2) {6};
\node[right,font=\scriptsize] at (21,6) {6};
\draw[red!70!black,thick] (0,1) -- (1,1) -- (2,1) -- (3,1) -- (4,1) -- (5,2) -- (6,3) -- (7,4) -- (8,5) -- (9,6) -- (10,7) -- (11,7) -- (12,7) -- (13,7) -- (14,7) -- (15,7) -- (16,7) -- (17,7) -- (18,7) -- (19,7) -- (20,7) -- (21,7);
\node[left,font=\scriptsize] at (0,1) {7};
\node[right,font=\scriptsize] at (21,7) {7};
\end{tikzpicture}

%% file: fig_gv.tex
\begin{tikzpicture}[x=0.42cm,y=0.42cm]
\fill[black!7] (15,0.55) rectangle (18,7+0.45); \node[font=\scriptsize] at (16.5,7+0.9) {$\beta_1$};
\fill[black!7] (18,0.55) rectangle (21,7+0.45); \node[font=\scriptsize] at (19.5,7+0.9) {$\beta_5$};
\draw[blue!70!black,thick] (0,7) -- (1,7) -- (2,7) -- (3,7) -- (4,6) -- (5,6) -- (6,6) -- (7,6) -- (8,6) -- (9,6) -- (10,5) -- (11,5) -- (12,5) -- (13,5) -- (14,4) -- (15,3) -- (16,3) -- (17,3) -- (18,3) -- (19,2) -- (20,1) -- (21,1);
\node[left,font=\scriptsize] at (0,7) {1};
\node[right,font=\scriptsize] at (21,1) {1};
\draw[blue!70!black,thick] (0,6) -- (1,6) -- (2,6) -- (3,5) -- (4,5) -- (5,5) -- (6,4) -- (7,4) -- (8,3) -- (9,3) -- (10,3) -- (11,3) -- (12,2) -- (13,2) -- (14,2) -- (15,2) -- (16,2) -- (17,2) -- (18,2) -- (19,3) -- (20,3) -- (21,2);
\node[left,font=\scriptsize] at (0,6) {2};
\node[right,font=\scriptsize] at (21,2) {2};
\draw[blue!70!black,thick] (0,5) -- (1,5) -- (2,4) -- (3,4) -- (4,4) -- (5,3) -- (6,3) -- (7,2) -- (8,2) -- (9,2) -- (10,2) -- (11,1) -- (12,1) -- (13,1) -- (14,1) -- (15,1) -- (16,1) -- (17,1) -- (18,1) -- (19,1) -- (20,2) -- (21,3);
\node[left,font=\scriptsize] at (0,5) {3};
\node[right,font=\scriptsize] at (21,3) {3};
\draw[black,thick] (0,4) -- (1,3) -- (2,3) -- (3,3) -- (4,3) -- (5,4) -- (6,5) -- (7,5) -- (8,5) -- (9,4) -- (10,4) -- (11,4) -- (12,4) -- (13,3) -- (14,3) -- (15,4) -- (16,4) -- (17,4) -- (18,4) -- (19,4) -- (20,4) -- (21,4);
\node[left,font=\scriptsize] at (0,4) {4};
\node[right,font=\scriptsize] at (21,4) {4};
\draw[red!70!black,thick] (0,3) -- (1,4) -- (2,5) -- (3,6) -- (4,7) -- (5,7) -- (6,7) -- (7,7) -- (8,7) -- (9,7) -- (10,7) -- (11,7) -- (12,7) -- (13,7) -- (14,7) -- (15,7) -- (16,6) -- (17,5) -- (18,5) -- (19,5) -- (20,5) -- (21,5);
\node[left,font=\scriptsize] at (0,3) {5};
\node[right,font=\scriptsize] at (21,5) {5};
\draw[red!70!black,thick] (0,2) -- (1,2) -- (2,2) -- (3,2) -- (4,2) -- (5,2) -- (6,2) -- (7,3) -- (8,4) -- (9,5) -- (10,6) -- (11,6) -- (12,6) -- (13,6) -- (14,6) -- (15,6) -- (16,7) -- (17,7) -- (18,6) -- (19,6) -- (20,6) -- (21,6);
\node[left,font=\scriptsize] at (0,2) {6};
\node[right,font=\scriptsize] at (21,6) {6};
\draw[red!70!black,thick] (0,1) -- (1,1) -- (2,1) -- (3,1) -- (4,1) -- (5,1) -- (6,1) -- (7,1) -- (8,1) -- (9,1) -- (10,1) -- (11,2) -- (12,3) -- (13,4) -- (14,5) -- (15,5) -- (16,5) -- (17,6) -- (18,7) -- (19,7) -- (20,7) -- (21,7);
\node[left,font=\scriptsize] at (0,1) {7};
\node[right,font=\scriptsize] at (21,7) {7};
\end{tikzpicture}

%% file: fig_p3_phi.tex
\begin{tikzpicture}[x=0.9cm,y=0.5cm]
\fill[black!8] (0,-0.45) rectangle (3,4+0.45); \node[font=\scriptsize] at (1.5,4+0.95) {cell $e=12$};
\fill[black!8] (7,-0.45) rectangle (10,4+0.45); \node[font=\scriptsize] at (8.5,4+0.95) {cell $e=23$};
\draw[blue!70!black,thick] (0,4) -- (1,4) -- (2,3) -- (3,2) -- (4,1) -- (5,1) -- (6,0) -- (7,0) -- (8,0) -- (9,0) -- (10,0);
\node[left,font=\scriptsize] at (0,4) {$X_1$};
\node[right,font=\scriptsize] at (10,0) {$X_1$};
\draw[black,thick] (0,3) -- (1,2) -- (2,2) -- (3,3) -- (4,3) -- (5,2) -- (6,2) -- (7,1) -- (8,1) -- (9,1) -- (10,1);
\node[left,font=\scriptsize] at (0,3) {$p_{12}$};
\node[right,font=\scriptsize] at (10,1) {$p_{12}$};
\draw[red!70!black,thick] (0,2) -- (1,3) -- (2,4) -- (3,4) -- (4,4) -- (5,4) -- (6,4) -- (7,4) -- (8,4) -- (9,3) -- (10,2);
\node[left,font=\scriptsize] at (0,2) {$X_2$};
\node[right,font=\scriptsize] at (10,2) {$X_2$};
\draw[black!55,thick] (0,1) -- (1,1) -- (2,1) -- (3,1) -- (4,2) -- (5,3) -- (6,3) -- (7,3) -- (8,2) -- (9,2) -- (10,3);
\node[left,font=\scriptsize] at (0,1) {$p_{23}$};
\node[right,font=\scriptsize] at (10,3) {$p_{23}$};
\draw[green!50!black,thick] (0,0) -- (1,0) -- (2,0) -- (3,0) -- (4,0) -- (5,0) -- (6,1) -- (7,2) -- (8,3) -- (9,4) -- (10,4);
\node[left,font=\scriptsize] at (0,0) {$X_3$};
\node[right,font=\scriptsize] at (10,4) {$X_3$};
\end{tikzpicture}

%% file: fig_p3_blowup.tex
\begin{tikzpicture}[x=0.27cm,y=0.36cm]
\fill[black!8] (0,-0.45) rectangle (9,10+0.45); \node[font=\scriptsize] at (4.5,10+0.95) {beads};
\fill[black!8] (9,-0.45) rectangle (24,10+0.45); \node[font=\scriptsize] at (16.5,10+0.95) {$\kappa$ of $12$};
\fill[black!8] (40,-0.45) rectangle (55,10+0.45); \node[font=\scriptsize] at (47.5,10+0.95) {$\kappa$ of $23$};
\draw[blue!70!black,thick] (0,10) -- (1,9) -- (2,8) -- (3,8) -- (4,8) -- (5,8) -- (6,8) -- (7,8) -- (8,8) -- (9,8) -- (10,8) -- (11,7) -- (12,7) -- (13,7) -- (14,6) -- (15,6) -- (16,5) -- (17,5) -- (18,5) -- (19,5) -- (20,4) -- (21,4) -- (22,4) -- (23,4) -- (24,4) -- (25,3) -- (26,3) -- (27,3) -- (28,3) -- (29,2) -- (30,2) -- (31,2) -- (32,1) -- (33,1) -- (34,1) -- (35,0) -- (36,0) -- (37,0) -- (38,0) -- (39,0) -- (40,0) -- (41,0) -- (42,0) -- (43,0) -- (44,0) -- (45,0) -- (46,0) -- (47,0) -- (48,0) -- (49,0) -- (50,0) -- (51,0) -- (52,0) -- (53,0) -- (54,0) -- (55,0);
\node[left,font=\scriptsize] at (0,10) {1};
\node[right,font=\scriptsize] at (55,0) {1};
\draw[blue!70!black,thick] (0,9) -- (1,10) -- (2,10) -- (3,9) -- (4,9) -- (5,9) -- (6,9) -- (7,9) -- (8,9) -- (9,9) -- (10,9) -- (11,9) -- (12,8) -- (13,8) -- (14,8) -- (15,7) -- (16,7) -- (17,6) -- (18,6) -- (19,6) -- (20,6) -- (21,5) -- (22,5) -- (23,5) -- (24,5) -- (25,5) -- (26,4) -- (27,4) -- (28,4) -- (29,4) -- (30,3) -- (31,3) -- (32,3) -- (33,2) -- (34,2) -- (35,2) -- (36,1) -- (37,1) -- (38,1) -- (39,1) -- (40,1) -- (41,1) -- (42,1) -- (43,1) -- (44,1) -- (45,1) -- (46,1) -- (47,1) -- (48,1) -- (49,1) -- (50,1) -- (51,1) -- (52,1) -- (53,1) -- (54,1) -- (55,1);
\node[left,font=\scriptsize] at (0,9) {2};
\node[right,font=\scriptsize] at (55,1) {2};
\draw[blue!70!black,thick] (0,8) -- (1,8) -- (2,9) -- (3,10) -- (4,10) -- (5,10) -- (6,10) -- (7,10) -- (8,10) -- (9,10) -- (10,10) -- (11,10) -- (12,10) -- (13,9) -- (14,9) -- (15,9) -- (16,9) -- (17,9) -- (18,9) -- (19,8) -- (20,8) -- (21,8) -- (22,8) -- (23,7) -- (24,6) -- (25,6) -- (26,6) -- (27,5) -- (28,5) -- (29,5) -- (30,5) -- (31,4) -- (32,4) -- (33,4) -- (34,3) -- (35,3) -- (36,3) -- (37,2) -- (38,2) -- (39,2) -- (40,2) -- (41,2) -- (42,2) -- (43,2) -- (44,2) -- (45,2) -- (46,2) -- (47,2) -- (48,2) -- (49,2) -- (50,2) -- (51,2) -- (52,2) -- (53,2) -- (54,2) -- (55,2);
\node[left,font=\scriptsize] at (0,8) {3};
\node[right,font=\scriptsize] at (55,2) {3};
\draw[black,thick] (0,7) -- (1,7) -- (2,7) -- (3,7) -- (4,7) -- (5,7) -- (6,7) -- (7,7) -- (8,7) -- (9,7) -- (10,6) -- (11,6) -- (12,6) -- (13,6) -- (14,7) -- (15,8) -- (16,8) -- (17,8) -- (18,7) -- (19,7) -- (20,7) -- (21,7) -- (22,6) -- (23,6) -- (24,7) -- (25,7) -- (26,7) -- (27,7) -- (28,6) -- (29,6) -- (30,6) -- (31,6) -- (32,6) -- (33,6) -- (34,6) -- (35,6) -- (36,6) -- (37,6) -- (38,5) -- (39,4) -- (40,3) -- (41,3) -- (42,3) -- (43,3) -- (44,3) -- (45,3) -- (46,3) -- (47,3) -- (48,3) -- (49,3) -- (50,3) -- (51,3) -- (52,3) -- (53,3) -- (54,3) -- (55,3);
\node[left,font=\scriptsize] at (0,7) {4};
\node[right,font=\scriptsize] at (55,3) {4};
\draw[red!70!black,thick] (0,6) -- (1,6) -- (2,6) -- (3,6) -- (4,5) -- (5,4) -- (6,4) -- (7,4) -- (8,4) -- (9,4) -- (10,4) -- (11,4) -- (12,4) -- (13,4) -- (14,4) -- (15,4) -- (16,4) -- (17,4) -- (18,4) -- (19,4) -- (20,5) -- (21,6) -- (22,7) -- (23,8) -- (24,8) -- (25,8) -- (26,8) -- (27,8) -- (28,8) -- (29,8) -- (30,8) -- (31,8) -- (32,8) -- (33,8) -- (34,8) -- (35,8) -- (36,8) -- (37,8) -- (38,8) -- (39,8) -- (40,8) -- (41,8) -- (42,7) -- (43,7) -- (44,7) -- (45,6) -- (46,6) -- (47,5) -- (48,5) -- (49,5) -- (50,5) -- (51,4) -- (52,4) -- (53,4) -- (54,4) -- (55,4);
\node[left,font=\scriptsize] at (0,6) {5};
\node[right,font=\scriptsize] at (55,4) {5};
\draw[red!70!black,thick] (0,5) -- (1,5) -- (2,5) -- (3,5) -- (4,6) -- (5,6) -- (6,5) -- (7,5) -- (8,5) -- (9,5) -- (10,5) -- (11,5) -- (12,5) -- (13,5) -- (14,5) -- (15,5) -- (16,6) -- (17,7) -- (18,8) -- (19,9) -- (20,9) -- (21,9) -- (22,9) -- (23,9) -- (24,9) -- (25,9) -- (26,9) -- (27,9) -- (28,9) -- (29,9) -- (30,9) -- (31,9) -- (32,9) -- (33,9) -- (34,9) -- (35,9) -- (36,9) -- (37,9) -- (38,9) -- (39,9) -- (40,9) -- (41,9) -- (42,9) -- (43,8) -- (44,8) -- (45,8) -- (46,7) -- (47,7) -- (48,6) -- (49,6) -- (50,6) -- (51,6) -- (52,5) -- (53,5) -- (54,5) -- (55,5);
\node[left,font=\scriptsize] at (0,5) {6};
\node[right,font=\scriptsize] at (55,5) {6};
\draw[red!70!black,thick] (0,4) -- (1,4) -- (2,4) -- (3,4) -- (4,4) -- (5,5) -- (6,6) -- (7,6) -- (8,6) -- (9,6) -- (10,7) -- (11,8) -- (12,9) -- (13,10) -- (14,10) -- (15,10) -- (16,10) -- (17,10) -- (18,10) -- (19,10) -- (20,10) -- (21,10) -- (22,10) -- (23,10) -- (24,10) -- (25,10) -- (26,10) -- (27,10) -- (28,10) -- (29,10) -- (30,10) -- (31,10) -- (32,10) -- (33,10) -- (34,10) -- (35,10) -- (36,10) -- (37,10) -- (38,10) -- (39,10) -- (40,10) -- (41,10) -- (42,10) -- (43,10) -- (44,9) -- (45,9) -- (46,9) -- (47,9) -- (48,9) -- (49,9) -- (50,8) -- (51,8) -- (52,8) -- (53,8) -- (54,7) -- (55,6);
\node[left,font=\scriptsize] at (0,4) {7};
\node[right,font=\scriptsize] at (55,6) {7};
\draw[black!55,thick] (0,3) -- (1,3) -- (2,3) -- (3,3) -- (4,3) -- (5,3) -- (6,3) -- (7,3) -- (8,3) -- (9,3) -- (10,3) -- (11,3) -- (12,3) -- (13,3) -- (14,3) -- (15,3) -- (16,3) -- (17,3) -- (18,3) -- (19,3) -- (20,3) -- (21,3) -- (22,3) -- (23,3) -- (24,3) -- (25,4) -- (26,5) -- (27,6) -- (28,7) -- (29,7) -- (30,7) -- (31,7) -- (32,7) -- (33,7) -- (34,7) -- (35,7) -- (36,7) -- (37,7) -- (38,7) -- (39,7) -- (40,7) -- (41,6) -- (42,6) -- (43,6) -- (44,6) -- (45,7) -- (46,8) -- (47,8) -- (48,8) -- (49,7) -- (50,7) -- (51,7) -- (52,7) -- (53,6) -- (54,6) -- (55,7);
\node[left,font=\scriptsize] at (0,3) {8};
\node[right,font=\scriptsize] at (55,7) {8};
\draw[green!50!black,thick] (0,2) -- (1,2) -- (2,2) -- (3,2) -- (4,2) -- (5,2) -- (6,2) -- (7,1) -- (8,0) -- (9,0) -- (10,0) -- (11,0) -- (12,0) -- (13,0) -- (14,0) -- (15,0) -- (16,0) -- (17,0) -- (18,0) -- (19,0) -- (20,0) -- (21,0) -- (22,0) -- (23,0) -- (24,0) -- (25,0) -- (26,0) -- (27,0) -- (28,0) -- (29,0) -- (30,0) -- (31,0) -- (32,0) -- (33,0) -- (34,0) -- (35,1) -- (36,2) -- (37,3) -- (38,3) -- (39,3) -- (40,4) -- (41,4) -- (42,4) -- (43,4) -- (44,4) -- (45,4) -- (46,4) -- (47,4) -- (48,4) -- (49,4) -- (50,4) -- (51,5) -- (52,6) -- (53,7) -- (54,8) -- (55,8);
\node[left,font=\scriptsize] at (0,2) {9};
\node[right,font=\scriptsize] at (55,8) {9};
\draw[green!50!black,thick] (0,1) -- (1,1) -- (2,1) -- (3,1) -- (4,1) -- (5,1) -- (6,1) -- (7,2) -- (8,2) -- (9,1) -- (10,1) -- (11,1) -- (12,1) -- (13,1) -- (14,1) -- (15,1) -- (16,1) -- (17,1) -- (18,1) -- (19,1) -- (20,1) -- (21,1) -- (22,1) -- (23,1) -- (24,1) -- (25,1) -- (26,1) -- (27,1) -- (28,1) -- (29,1) -- (30,1) -- (31,1) -- (32,2) -- (33,3) -- (34,4) -- (35,4) -- (36,4) -- (37,4) -- (38,4) -- (39,5) -- (40,5) -- (41,5) -- (42,5) -- (43,5) -- (44,5) -- (45,5) -- (46,5) -- (47,6) -- (48,7) -- (49,8) -- (50,9) -- (51,9) -- (52,9) -- (53,9) -- (54,9) -- (55,9);
\node[left,font=\scriptsize] at (0,1) {10};
\node[right,font=\scriptsize] at (55,9) {10};
\draw[green!50!black,thick] (0,0) -- (1,0) -- (2,0) -- (3,0) -- (4,0) -- (5,0) -- (6,0) -- (7,0) -- (8,1) -- (9,2) -- (10,2) -- (11,2) -- (12,2) -- (13,2) -- (14,2) -- (15,2) -- (16,2) -- (17,2) -- (18,2) -- (19,2) -- (20,2) -- (21,2) -- (22,2) -- (23,2) -- (24,2) -- (25,2) -- (26,2) -- (27,2) -- (28,2) -- (29,3) -- (30,4) -- (31,5) -- (32,5) -- (33,5) -- (34,5) -- (35,5) -- (36,5) -- (37,5) -- (38,6) -- (39,6) -- (40,6) -- (41,7) -- (42,8) -- (43,9) -- (44,10) -- (45,10) -- (46,10) -- (47,10) -- (48,10) -- (49,10) -- (50,10) -- (51,10) -- (52,10) -- (53,10) -- (54,10) -- (55,10);
\node[left,font=\scriptsize] at (0,0) {11};
\node[right,font=\scriptsize] at (55,10) {11};
\end{tikzpicture}

%% file: fig_p3_blowup_v.tex
\begin{tikzpicture}[x=0.27cm,y=0.36cm]
\fill[black!8] (0,-0.45) rectangle (15,10+0.45); \node[font=\scriptsize] at (7.5,10+0.95) {$\kappa$ of $12$};
\fill[black!8] (31,-0.45) rectangle (46,10+0.45); \node[font=\scriptsize] at (38.5,10+0.95) {$\kappa$ of $23$};
\fill[black!8] (46,-0.45) rectangle (55,10+0.45); \node[font=\scriptsize] at (50.5,10+0.95) {beads};
\draw[blue!70!black,thick] (0,10) -- (1,10) -- (2,10) -- (3,10) -- (4,9) -- (5,9) -- (6,9) -- (7,9) -- (8,9) -- (9,9) -- (10,8) -- (11,8) -- (12,8) -- (13,8) -- (14,7) -- (15,6) -- (16,6) -- (17,6) -- (18,5) -- (19,5) -- (20,5) -- (21,5) -- (22,4) -- (23,4) -- (24,4) -- (25,3) -- (26,3) -- (27,3) -- (28,2) -- (29,2) -- (30,2) -- (31,2) -- (32,2) -- (33,2) -- (34,2) -- (35,2) -- (36,2) -- (37,2) -- (38,2) -- (39,2) -- (40,2) -- (41,2) -- (42,2) -- (43,2) -- (44,2) -- (45,2) -- (46,2) -- (47,2) -- (48,2) -- (49,2) -- (50,2) -- (51,2) -- (52,2) -- (53,1) -- (54,0) -- (55,0);
\node[left,font=\scriptsize] at (0,10) {1};
\node[right,font=\scriptsize] at (55,0) {1};
\draw[blue!70!black,thick] (0,9) -- (1,9) -- (2,9) -- (3,8) -- (4,8) -- (5,8) -- (6,7) -- (7,7) -- (8,6) -- (9,6) -- (10,6) -- (11,6) -- (12,5) -- (13,5) -- (14,5) -- (15,5) -- (16,5) -- (17,4) -- (18,4) -- (19,4) -- (20,4) -- (21,3) -- (22,3) -- (23,3) -- (24,2) -- (25,2) -- (26,2) -- (27,1) -- (28,1) -- (29,1) -- (30,1) -- (31,1) -- (32,1) -- (33,1) -- (34,1) -- (35,1) -- (36,1) -- (37,1) -- (38,1) -- (39,1) -- (40,1) -- (41,1) -- (42,1) -- (43,1) -- (44,1) -- (45,1) -- (46,1) -- (47,1) -- (48,1) -- (49,1) -- (50,1) -- (51,1) -- (52,1) -- (53,2) -- (54,2) -- (55,1);
\node[left,font=\scriptsize] at (0,9) {2};
\node[right,font=\scriptsize] at (55,1) {2};
\draw[blue!70!black,thick] (0,8) -- (1,8) -- (2,7) -- (3,7) -- (4,7) -- (5,6) -- (6,6) -- (7,5) -- (8,5) -- (9,5) -- (10,5) -- (11,4) -- (12,4) -- (13,4) -- (14,4) -- (15,4) -- (16,3) -- (17,3) -- (18,3) -- (19,3) -- (20,2) -- (21,2) -- (22,2) -- (23,1) -- (24,1) -- (25,1) -- (26,0) -- (27,0) -- (28,0) -- (29,0) -- (30,0) -- (31,0) -- (32,0) -- (33,0) -- (34,0) -- (35,0) -- (36,0) -- (37,0) -- (38,0) -- (39,0) -- (40,0) -- (41,0) -- (42,0) -- (43,0) -- (44,0) -- (45,0) -- (46,0) -- (47,0) -- (48,0) -- (49,0) -- (50,0) -- (51,0) -- (52,0) -- (53,0) -- (54,1) -- (55,2);
\node[left,font=\scriptsize] at (0,8) {3};
\node[right,font=\scriptsize] at (55,2) {3};
\draw[black,thick] (0,7) -- (1,6) -- (2,6) -- (3,6) -- (4,6) -- (5,7) -- (6,8) -- (7,8) -- (8,8) -- (9,7) -- (10,7) -- (11,7) -- (12,7) -- (13,6) -- (14,6) -- (15,7) -- (16,7) -- (17,7) -- (18,7) -- (19,6) -- (20,6) -- (21,6) -- (22,6) -- (23,6) -- (24,6) -- (25,6) -- (26,6) -- (27,6) -- (28,6) -- (29,5) -- (30,4) -- (31,3) -- (32,3) -- (33,3) -- (34,3) -- (35,3) -- (36,3) -- (37,3) -- (38,3) -- (39,3) -- (40,3) -- (41,3) -- (42,3) -- (43,3) -- (44,3) -- (45,3) -- (46,3) -- (47,3) -- (48,3) -- (49,3) -- (50,3) -- (51,3) -- (52,3) -- (53,3) -- (54,3) -- (55,3);
\node[left,font=\scriptsize] at (0,7) {4};
\node[right,font=\scriptsize] at (55,3) {4};
\draw[red!70!black,thick] (0,6) -- (1,7) -- (2,8) -- (3,9) -- (4,10) -- (5,10) -- (6,10) -- (7,10) -- (8,10) -- (9,10) -- (10,10) -- (11,10) -- (12,10) -- (13,10) -- (14,10) -- (15,10) -- (16,10) -- (17,10) -- (18,10) -- (19,10) -- (20,10) -- (21,10) -- (22,10) -- (23,10) -- (24,10) -- (25,10) -- (26,10) -- (27,10) -- (28,10) -- (29,10) -- (30,10) -- (31,10) -- (32,10) -- (33,10) -- (34,10) -- (35,9) -- (36,9) -- (37,9) -- (38,9) -- (39,9) -- (40,9) -- (41,8) -- (42,8) -- (43,8) -- (44,8) -- (45,7) -- (46,6) -- (47,6) -- (48,6) -- (49,6) -- (50,5) -- (51,4) -- (52,4) -- (53,4) -- (54,4) -- (55,4);
\node[left,font=\scriptsize] at (0,6) {5};
\node[right,font=\scriptsize] at (55,4) {5};
\draw[red!70!black,thick] (0,5) -- (1,5) -- (2,5) -- (3,5) -- (4,5) -- (5,5) -- (6,5) -- (7,6) -- (8,7) -- (9,8) -- (10,9) -- (11,9) -- (12,9) -- (13,9) -- (14,9) -- (15,9) -- (16,9) -- (17,9) -- (18,9) -- (19,9) -- (20,9) -- (21,9) -- (22,9) -- (23,9) -- (24,9) -- (25,9) -- (26,9) -- (27,9) -- (28,9) -- (29,9) -- (30,9) -- (31,9) -- (32,9) -- (33,9) -- (34,8) -- (35,8) -- (36,8) -- (37,7) -- (38,7) -- (39,6) -- (40,6) -- (41,6) -- (42,6) -- (43,5) -- (44,5) -- (45,5) -- (46,5) -- (47,5) -- (48,5) -- (49,5) -- (50,6) -- (51,6) -- (52,5) -- (53,5) -- (54,5) -- (55,5);
\node[left,font=\scriptsize] at (0,5) {6};
\node[right,font=\scriptsize] at (55,5) {6};
\draw[red!70!black,thick] (0,4) -- (1,4) -- (2,4) -- (3,4) -- (4,4) -- (5,4) -- (6,4) -- (7,4) -- (8,4) -- (9,4) -- (10,4) -- (11,5) -- (12,6) -- (13,7) -- (14,8) -- (15,8) -- (16,8) -- (17,8) -- (18,8) -- (19,8) -- (20,8) -- (21,8) -- (22,8) -- (23,8) -- (24,8) -- (25,8) -- (26,8) -- (27,8) -- (28,8) -- (29,8) -- (30,8) -- (31,8) -- (32,8) -- (33,7) -- (34,7) -- (35,7) -- (36,6) -- (37,6) -- (38,5) -- (39,5) -- (40,5) -- (41,5) -- (42,4) -- (43,4) -- (44,4) -- (45,4) -- (46,4) -- (47,4) -- (48,4) -- (49,4) -- (50,4) -- (51,5) -- (52,6) -- (53,6) -- (54,6) -- (55,6);
\node[left,font=\scriptsize] at (0,4) {7};
\node[right,font=\scriptsize] at (55,6) {7};
\draw[black!55,thick] (0,3) -- (1,3) -- (2,3) -- (3,3) -- (4,3) -- (5,3) -- (6,3) -- (7,3) -- (8,3) -- (9,3) -- (10,3) -- (11,3) -- (12,3) -- (13,3) -- (14,3) -- (15,3) -- (16,4) -- (17,5) -- (18,6) -- (19,7) -- (20,7) -- (21,7) -- (22,7) -- (23,7) -- (24,7) -- (25,7) -- (26,7) -- (27,7) -- (28,7) -- (29,7) -- (30,7) -- (31,7) -- (32,6) -- (33,6) -- (34,6) -- (35,6) -- (36,7) -- (37,8) -- (38,8) -- (39,8) -- (40,7) -- (41,7) -- (42,7) -- (43,7) -- (44,6) -- (45,6) -- (46,7) -- (47,7) -- (48,7) -- (49,7) -- (50,7) -- (51,7) -- (52,7) -- (53,7) -- (54,7) -- (55,7);
\node[left,font=\scriptsize] at (0,3) {8};
\node[right,font=\scriptsize] at (55,7) {8};
\draw[green!50!black,thick] (0,2) -- (1,2) -- (2,2) -- (3,2) -- (4,2) -- (5,2) -- (6,2) -- (7,2) -- (8,2) -- (9,2) -- (10,2) -- (11,2) -- (12,2) -- (13,2) -- (14,2) -- (15,2) -- (16,2) -- (17,2) -- (18,2) -- (19,2) -- (20,3) -- (21,4) -- (22,5) -- (23,5) -- (24,5) -- (25,5) -- (26,5) -- (27,5) -- (28,5) -- (29,6) -- (30,6) -- (31,6) -- (32,7) -- (33,8) -- (34,9) -- (35,10) -- (36,10) -- (37,10) -- (38,10) -- (39,10) -- (40,10) -- (41,10) -- (42,10) -- (43,10) -- (44,10) -- (45,10) -- (46,10) -- (47,9) -- (48,8) -- (49,8) -- (50,8) -- (51,8) -- (52,8) -- (53,8) -- (54,8) -- (55,8);
\node[left,font=\scriptsize] at (0,2) {9};
\node[right,font=\scriptsize] at (55,8) {9};
\draw[green!50!black,thick] (0,1) -- (1,1) -- (2,1) -- (3,1) -- (4,1) -- (5,1) -- (6,1) -- (7,1) -- (8,1) -- (9,1) -- (10,1) -- (11,1) -- (12,1) -- (13,1) -- (14,1) -- (15,1) -- (16,1) -- (17,1) -- (18,1) -- (19,1) -- (20,1) -- (21,1) -- (22,1) -- (23,2) -- (24,3) -- (25,4) -- (26,4) -- (27,4) -- (28,4) -- (29,4) -- (30,5) -- (31,5) -- (32,5) -- (33,5) -- (34,5) -- (35,5) -- (36,5) -- (37,5) -- (38,6) -- (39,7) -- (40,8) -- (41,9) -- (42,9) -- (43,9) -- (44,9) -- (45,9) -- (46,9) -- (47,10) -- (48,10) -- (49,9) -- (50,9) -- (51,9) -- (52,9) -- (53,9) -- (54,9) -- (55,9);
\node[left,font=\scriptsize] at (0,1) {10};
\node[right,font=\scriptsize] at (55,9) {10};
\draw[green!50!black,thick] (0,0) -- (1,0) -- (2,0) -- (3,0) -- (4,0) -- (5,0) -- (6,0) -- (7,0) -- (8,0) -- (9,0) -- (10,0) -- (11,0) -- (12,0) -- (13,0) -- (14,0) -- (15,0) -- (16,0) -- (17,0) -- (18,0) -- (19,0) -- (20,0) -- (21,0) -- (22,0) -- (23,0) -- (24,0) -- (25,0) -- (26,1) -- (27,2) -- (28,3) -- (29,3) -- (30,3) -- (31,4) -- (32,4) -- (33,4) -- (34,4) -- (35,4) -- (36,4) -- (37,4) -- (38,4) -- (39,4) -- (40,4) -- (41,4) -- (42,5) -- (43,6) -- (44,7) -- (45,8) -- (46,8) -- (47,8) -- (48,9) -- (49,10) -- (50,10) -- (51,10) -- (52,10) -- (53,10) -- (54,10) -- (55,10);
\node[left,font=\scriptsize] at (0,0) {11};
\node[right,font=\scriptsize] at (55,10) {11};
\end{tikzpicture}

%% file: main.bbl
\addcontentsline{toc}{part}{References}
\begin{thebibliography}{99}
\bibitem{AMP} O.~Aichholzer, W.~Mulzer, A.~Pilz, Flip distance between triangulations of a simple polygon is
NP-complete, \emph{Discrete Comput.\ Geom.} 54 (2015) 368--389, \doi{10.1007/s00454-015-9709-7}.
\bibitem{AlimontiKann} P.~Alimonti, V.~Kann, Some APX-completeness results for cubic graphs,
\emph{Theoret.\ Comput.\ Sci.} 237 (2000) 123--134, \doi{10.1016/S0304-3975(98)00158-3}.
\bibitem{Bergold} H.~Bergold, \emph{Signotopes and Convex Drawings}, PhD thesis, Freie Universit\"at Berlin, 2023,
\doi{10.17169/refubium-42367}.
\bibitem{BFS} H.~Bergold, S.~Felsner, M.~Scheucher, An extension theorem for signotopes, \emph{Combin.\ Theory}
5(4) (2025) \#2, \doi{10.5070/C65465668}.
\bibitem{BLSWZ} A.~Bj\"orner, M.~Las Vergnas, B.~Sturmfels, N.~White, G.~M.~Ziegler, \emph{Oriented Matroids},
2nd ed., Encyclopedia Math.\ Appl.\ 46, Cambridge University Press, Cambridge, 1999,
\doi{10.1017/CBO9780511586507}.
\bibitem{BFR} O.~Bodini, T.~Fernique, \'E.~R\'emila, Distances on lozenge tilings, in: \emph{Discrete Geometry for
Computer Imagery (DGCI 2009)}, Lecture Notes in Comput.\ Sci.\ 5810, Springer, Berlin, 2009, 240--251,
\doi{10.1007/978-3-642-04397-0_21}.
\bibitem{BFRR} O.~Bodini, T.~Fernique, M.~Rao, \'E.~R\'emila, Distances on rhombus tilings, \emph{Theoret.\
Comput.\ Sci.} 412 (2011) 4787--4794, \doi{10.1016/j.tcs.2011.04.015}.
\bibitem{BokowskiRohlfs} J.~Bokowski, H.~Rohlfs, On a mutation problem for oriented matroids, \emph{European J.\
Combin.} 22 (2001) 617--626, \doi{10.1006/eujc.2000.0483}.
\bibitem{Dorfer} J.~Dorfer, Flip distance between triangulations of convex polygons is NP-complete,
arXiv:2602.22874 (2026).
\bibitem{DumitrescuMandal} A.~Dumitrescu, R.~Mandal, New lower bounds for the number of pseudoline arrangements,
in: \emph{Proceedings of the Thirtieth Annual ACM-SIAM Symposium on Discrete Algorithms (SODA)}, SIAM, Philadelphia,
2019, 410--425, \doi{10.1137/1.9781611975482.26}.
\bibitem{Elnitsky} S.~Elnitsky, Rhombic tilings of polygons and classes of reduced words in Coxeter groups,
\emph{J.\ Combin.\ Theory Ser.\ A} 77 (1997) 193--221, \doi{10.1006/jcta.1997.2723}.
\bibitem{FelsnerValtr} S.~Felsner, P.~Valtr, Coding and counting arrangements of pseudolines, \emph{Discrete
Comput.\ Geom.} 46 (2011) 405--416, \doi{10.1007/s00454-011-9366-4}.
\bibitem{FW00} S.~Felsner, H.~Weil, A theorem on higher Bruhat orders, \emph{Discrete Comput.\ Geom.} 23 (2000)
121--127, \doi{10.1007/PL00009485}.
\bibitem{FW} S.~Felsner, H.~Weil, Sweeps, arrangements and signotopes, \emph{Discrete Appl.\ Math.} 109 (2001)
67--94, \doi{10.1016/S0166-218X(00)00232-8}.
\bibitem{GMS} G.~Gutierres, R.~Mamede, J.~L.~Santos, Commutation classes of the reduced words for the longest
element of $\mathfrak S_n$, \emph{Electron.\ J.\ Combin.} 27(2) (2020) P2.21, \doi{10.37236/9481}.
\bibitem{GMS22} G.~Gutierres, R.~Mamede, J.~L.~Santos, Diameter of the commutation classes graph of a permutation,
\emph{European J.\ Combin.} 103 (2022) 103525, \doi{10.1016/j.ejc.2022.103525}.
\bibitem{GuKnauer} Q.~Gu, K.~Knauer, Neither simpliciality nor mutation connectivity: conjectures of Las
Vergnas and Cordovil--Las Vergnas fail, arXiv:2609.15850 (2026).
\bibitem{KV} M.~M.~Kapranov, V.~A.~Voevodsky, Combinatorial-geometric aspects of polycategory theory: pasting
schemes and higher Bruhat orders (list of results), \emph{Cahiers Topologie G\'eom.\ Diff\'erentielle Cat\'eg.} 32
(1991) 11--27.
\bibitem{Karp} R.~M.~Karp, Reducibility among combinatorial problems, in: R.~E.~Miller, J.~W.~Thatcher (eds.),
\emph{Complexity of Computer Computations}, Plenum Press, New York, 1972, 85--103,
\doi{10.1007/978-1-4684-2001-2_9}.
\bibitem{KnauerMarc} K.~Knauer, T.~Marc, Corners and simpliciality in oriented matroids and partial cubes,
\emph{European J.\ Combin.} 112 (2023) 103714, \doi{10.1016/j.ejc.2023.103714}.
\bibitem{LubiwPathak} A.~Lubiw, V.~Pathak, Flip distance between two triangulations of a point set is
NP-complete, \emph{Comput.\ Geom.} 49 (2015) 17--23, \doi{10.1016/j.comgeo.2014.11.001}.
\bibitem{MS} Yu.~I.~Manin, V.~V.~Schechtman, Arrangements of hyperplanes, higher braid groups and higher Bruhat
orders, in: J.~Coates, R.~Greenberg, B.~Mazur, I.~Satake (eds.), \emph{Algebraic Number Theory -- in honor of
K.~Iwasawa}, Adv.\ Stud.\ Pure Math.\ 17, Academic Press, Boston, 1989, 289--308, \doi{10.2969/aspm/01710289}.
\bibitem{Matsumoto} H.~Matsumoto, G\'en\'erateurs et relations des groupes de Weyl g\'en\'eralis\'es,
\emph{C.~R. Acad.\ Sci.\ Paris} 258 (1964) 3419--3422.
\bibitem{Miyata} H.~Miyata, On combinatorial properties of points and polynomial curves, arXiv:1703.04963
(2017).
\bibitem{Opatrny} J.~Opatrny, Total ordering problem, \emph{SIAM J.\ Comput.} 8 (1979) 111--114,
\doi{10.1137/0208008}.
\bibitem{Pilz} A.~Pilz, Flip distance between triangulations of a planar point set is APX-hard,
\emph{Comput.\ Geom.} 47 (2014) 589--604, \doi{10.1016/j.comgeo.2014.01.001}.
\bibitem{Radtke} Y.~Alves Radtke, S.~Felsner, J.~Obenaus, S.~Roch, M.~Scheucher, B.~Vogtenhuber, Flip graph
connectivity for arrangements of pseudolines and pseudocircles, in: \emph{Proceedings of the 2024 Annual ACM-SIAM
Symposium on Discrete Algorithms (SODA)}, SIAM, Philadelphia, 2024, 4849--4871,
\doi{10.1137/1.9781611977912.172}.
\bibitem{Rambau} J.~Rambau, Triangulations of cyclic polytopes and higher Bruhat orders, \emph{Mathematika} 44
(1997) 162--194, \doi{10.1112/S0025579300012055}.
\bibitem{ReinerRoichman} V.~Reiner, Y.~Roichman, Diameter of graphs of reduced words and galleries, \emph{Trans.\
Amer.\ Math.\ Soc.} 365 (2013) 2779--2802, \doi{10.1090/S0002-9947-2012-05719-9}.
\bibitem{Handbook} J.~Richter-Gebert, G.~M.~Ziegler, Oriented matroids, in: J.~E.~Goodman, J.~O'Rourke,
C.~D.~T\'oth (eds.), \emph{Handbook of Discrete and Computational Geometry}, 3rd ed., Chapman \& Hall/CRC, Boca
Raton, 2017, 159--184, \doi{10.1201/9781315119601}.
\bibitem{Ringel} G.~Ringel, Teilungen der Ebene durch Geraden oder topologische Geraden, \emph{Math.\ Z.} 64
(1956) 79--102, \doi{10.1007/BF01166556}.
\bibitem{RoudneffSturmfels} J.-P.~Roudneff, B.~Sturmfels, Simplicial cells in arrangements and mutations of
oriented matroids, \emph{Geom.\ Dedicata} 27 (1988) 153--170, \doi{10.1007/BF00151346}.
\bibitem{Tits} J.~Tits, Le probl\`eme des mots dans les groupes de Coxeter, in: \emph{Symposia Mathematica},
Vol.~1 (INDAM, Rome, 1967/68), Academic Press, London, 1969, 175--185.
\bibitem{Wilhelmi} M.~Wilhelmi, Mutations and (Non-)Euclideaness in oriented matroids, arXiv:2501.12951 (2025).
\bibitem{Ziegler} G.~M.~Ziegler, Higher Bruhat orders and cyclic hyperplane arrangements, \emph{Topology} 32
(1993) 259--279, \doi{10.1016/0040-9383(93)90019-R}.
\end{thebibliography}
